%% file: main.tex
\documentclass[11pt]{article} 

\usepackage[
  top=1in,
  bottom=1in,
  left=1in,
  right=1in
]{geometry}
\usepackage{amsmath} 
\usepackage{amsfonts,amssymb}
\usepackage{comment}
\usepackage{algorithm}
\usepackage{algpseudocode}
\usepackage{amsmath}
\usepackage{makecell}
\usepackage{graphicx} 
\usepackage{enumitem} 
\setlist[itemize]{leftmargin=*}
\setlist[enumerate]{leftmargin=*}
\usepackage{mathrsfs}
\usepackage{setspace} 
\usepackage{bbm}
\usepackage{mathtools}
\usepackage{url}  
\usepackage[dvipsnames]{xcolor}
\usepackage[
    linktoc=page,
    colorlinks=true,
    linkcolor=RoyalBlue,
    citecolor=BrickRed,
    urlcolor=black
]{hyperref} 
\input{preamble}

\usepackage[numbers]{natbib} 
\usepackage{amsthm}

\newtheorem{theorem}{Theorem}[section]

\newtheorem{lem}{Lemma}[section]

\newtheorem{rem}{Remark}[section]

\newtheorem{definition}{Definition}[section]
 \numberwithin{equation}{section}

\usepackage{bm}
\usepackage{subcaption}

\DeclareMathOperator{\sign}{sign}

\DeclareMathOperator{\sfT}{\mathsf{T}}

\DeclareMathOperator{\dist}{dist}

\DeclareMathOperator{\relu}{ReLU}
 
\allowdisplaybreaks

\title{Instance Optimal Sparse Recovery  from Nonlinear\\ Observations: A Unified Framework}
\date{\today}

\author{Junren Chen\thanks{Department of Statistics, Columbia University.  (emails: \texttt{jc6315@columbia.edu,\,mm4338@columbia.edu})} 
\and Arian Maleki\footnotemark[1]
}

\begin{document}

\maketitle
  
\begin{abstract}
Sparse recovery has found numerous applications. As most real-world signals and parameters are only approximately sparse, a fundamental property for sparse recovery algorithms is the robustness to model error. This is characterized by the instance optimality in compressed sensing, which concerns sparse recovery from \emph{linear} measurements. While restricted isometry property (RIP) yields instance optimality in compressed sensing, it remains much less clear how to achieve instance optimal sparse recovery from   \emph{nonlinear} observations. This paper develops a unified framework to this end. The main ingredient is a signal-dependent restricted approximate invertibility condition (RAIC) of some ``gradient'', which leads to the instance optimality of iterative hard thresholding. Under Gaussian designs, we apply the proposed framework to phaseless, one-bit, and ReLU measurements, which correspond to the problems of sparse phase retrieval, one-bit compressed sensing, and sparse ReLU regression, respectively. For sparse phase retrieval, we propose a variant of thresholded amplitude flow and show its
  instance optimality under $O(s^3)$  measurements (up to logarithmic factors), where $s$ is the sparsity level. To our best knowledge, this is the first instance optimal efficient algorithm for sparse phase retrieval and complements \cite{gao2016stable} that achieved this via a computationally intractable program. In one-bit compressed sensing, we establish the instance optimality of normalized binary iterative hard thresholding and strengthen the recent result of \cite{matsumoto2024binary}. In sparse ReLU regression, it is shown that a slight variant of the algorithm in \cite{soltanolkotabi2017learning} is instance optimal. 
  Moreover, $(\ell_2,\ell_2)$ non-uniform instance optimal guarantees are obtained for these problems. The analysis is built upon a number of high-dimensional concentration bounds, including bounds on restricted eigenvalues and a novel instance-dependent hyperplane tessellation result. 
\end{abstract}

\medskip

\noindent\textit{Mathematics Subject Classification (2020).} Primary 94A12; Secondary 94A20, 60B20, 60D05.

\smallskip

\noindent\textit{Keywords.} instance optimality; compressed sensing; nonlinear observations; iterative hard thresholding; sparse phase retrieval;  hyperplane tessellation; oracle inequality 

\medskip 

\noindent 

\section{Introduction}
In   compressed sensing that concerns sparse recovery   from linear measurements, instance optimality is a basic notion that characterizes the robustness of an algorithm to model error, where model error refers to the error arising when the true signal is not exactly sparse. Let the $\ell_q~(q\ge 1)$ norm of $\bu$ be $\|\bu\|_q = (\sum_{i}|u_i|^q)^{1/q}$, an estimator $\hat{\bx}$ constructed from the sensing matrix and measurements is said to be $(\ell_q,\ell_p)$-instance optimal of order $s$ if \citep{Foucart2013AMI} 
\begin{align}
    \|\hat{\bx}-\bx\|_q \le \frac{C}{s^{1/p-1/q}}\|\bx-\bx_{[s]}\|_p,\quad\forall \bx\in\mathbb{R}^n\label{csgeneraliop}
\end{align}
holds for some constant $C>0$. Throughout this paper, $\bu_{[s]}$ denotes the best $s$-sparse approximation of $\bu$; ties are broken by selecting the support with the smallest indices.  In a nutshell, an instance optimal algorithm is capable of recovering all signals to an error proportional to a certain distance to $\Sigma^n_s= \{\bu\in \mathbb{R}^n:\|\bu\|_0\le s\}$.

  To start, we consider the sparse recovery of an $n$-dimensional vector $\bx$ from the linear measurements $\by=\bA\bx$ for some measurement matrix $\bA$. It is well known that a number of algorithms---such as basis pursuit denoising \citep{cai2013sparse,traonmilin2018stable}, iterative hard thresholding \citep{blumensath2009iterative}, CoSaMP \citep{needell2009cosamp}, hard thresholding pursuit \citep{foucart2011hard}---achieve the $(\ell_2,\ell_1)$-instance optimality 
\begin{align}\label{csl2l111}
    \|\hat{\bx}-\bx\|_2 \le \frac{C}{\sqrt{s}}\|\bx-\bx_{[s]}\|_1,\quad\forall \bx\in\mathbb{R}^n,  
\end{align}
as long as $\bA$
  satisfies a certain restricted isometry property (RIP). 
Also, there is a negative result that  the $(\ell_2,\ell_2)$-instance optimality 
\[ \|\hat{\bx}-\bx\|_2 \le C\|\bx-\bx_{[s]}\|_2,\quad\forall \bx\in\mathbb{R}^n\]
is not possible unless $m\gtrsim n$ \citep{cohen2009compressed,bourrier2014fundamental}. Nonetheless, the non-uniform $(\ell_2,\ell_2)$-instance optimality, which only concerns the recovery of a fixed $\bx\in \mathbb{R}^n$, is possible in high dimensions (i.e., $m\ll n$). It was proved by \cite{wojtaszczyk2010stability} that, under $O(s\log\frac{en}{s})$ Gaussian linear measurements, basis pursuit denoising achieves 
\begin{align}\label{nonulinear}
    \|\hat{\bx}-\bx\|_2 \le C\|\bx-\bx_{[s]}\|_2,\quad\textrm{with high probability (w.h.p.)}, 
\end{align}
for any $\bx\in\mathbb{R}^n$ oblivious to the sensing matrix. This is also  known as instance optimality in probability.

Nonlinear measurements, however, are ubiquitous in signal processing, statistics, and machine learning. The first such example is one-bit quantization, which helps increase the sampling rates in signal processing and overcome the communication cost bottleneck in distributed machine learning. In the setting of sparse recovery, this leads to the extensively studied problem of one-bit compressed sensing \citep{boufounos20081}. Another prominent example is that of phase retrieval \citep{shechtman2015phase}, inspired by many imaging problems where the sensors unavoidably lose the phase information and capture the magnitude only. 
In high dimensions, it is standard to leverage the sparsity of the signals, which then inspires the problem of sparse phase retrieval \citep{jaganathan2016phase}. Further examples, to name just a few, include ReLU regression \citep{soltanolkotabi2017learning}, modulo measurements \citep{bhandari2020unlimited}, saturated measurements \cite{foucart2017sparse}, phase-only measurements \citep{oppenheim1981importance}, one-bit phaseless measurements  \citep{chen2024one},  and more general unknown link functions \citep{plan2016generalized}.

Since most real-world signals and parameters are only approximately sparse rather than exactly sparse \citep{mallat2009wavelet}, we ask in this paper the following question:
\begin{gather*}
\textit{How to achieve instance optimal sparse recovery from nonlinear observations?} 
\end{gather*}
Observe that (\ref{csgeneraliop}) from compressed sensing exhibits the two defining features: (i) a model error term that vanishes for $\bx\in\Sigma^n_s$, and (ii) the uniformity for all signals in $\mathbb{R}^n$. Yet, defining instance optimality in general nonlinear  sparse recovery problems already involves some subtleties. First of all, the entire signal space may not be $\mathbb{R}^n$. For instance, in the recovery of $\bx$ from $\by=\sign(\bA\bx)$ where $\sign(a)=\mathbbm{1}(a\ge 0) - \mathbbm{1}(a<0)$, i.e., the one-bit sensing problem, the signal norm is completely lost and it is standard to assume the signals live in the  unit Euclidean
sphere $\mathbb{S}^{n-1}$. In this paper, we shall introduce a generic notation $\bar{\calX}$ to denote the signal space. 
Perhaps more importantly, while (\ref{csgeneraliop}) guarantees exact recovery for $\bx\in\Sigma^n_s$, this is not possible for some nonlinear problems (again like one-bit sensing), where exact recovery is not possible even when the signal is $s$-sparse. In these problems, the best we can hope for is to achieve the statistically optimal error rate for $\bx\in\Sigma^n_s$. We will formalize these discussions in Section \ref{sec:unified}. 

\subsection{Main Contributions}
We develop a unified framework for proving instance optimality of iterative hard thresholding (IHT) in sparse recovery from nonlinear measurements. Given a sparsity level $k$ and a problem-dependent gradient map $\bh_{\bx}:\mathbb{R}^n\to\mathbb{R}^n$, IHT refers to the iteration
\[
\bx_{t+1}=H_k\big(\bx_t-\eta\bh_{\bx}(\bx_t)\big),
\]
where $H_k(\bu)=\bu_{[k]}$ is the hard-thresholding operator. Our approach first decomposes the signal space $\bar{\calX}$ into $\calX$ and $\calX^c=\bar{\calX}\setminus\calX$, where $\calX$ collects signals that are suitably close to $\Sigma_s^n$, and particularly contains $\Sigma_s^n\cap\bar{\calX}$.

The main challenge lies in establishing the instance optimal recovery for $\bx\in\calX$. To this end, our key technical ingredient is a structured condition on $\bh_{\bx}$---a \emph{signal-dependent} variant of the \emph{restricted approximate invertibility condition} (RAIC)---which directly yields instance optimality over $\calX$. In prior works, the RAIC uniformly controls the deviation of $\bh_{\bx}(\bu)$ from the ideal invertibility step $\bu-\bx$ for exactly sparse pairs $(\bu,\bx)$. Our variant extends this control to $\bx\in\calX$ at the cost of an additional signal-dependent error term. Further details are given in Section \ref{sec:unifieda}. 

As applications, we establish instance optimality of IHT-type algorithms for sparse phase retrieval, one-bit compressed sensing and ReLU regression  with Gaussian designs. We highlight two representative results below.

\begin{itemize}
    \item In sparse phase retrieval  (i.e., recovery of $\bx$ with a sparse prior from $\by=|\bA\bx|$), our main result (Theorem \ref{thm:uni}) shows that a variant of thresholded amplitude flow, viewed as IHT with the amplitude-based $\ell_2$ loss, produces a sequence $\{\bx_t\}_{t=0}^\infty$ satisfying 
    \begin{align}
        \lim_{t\to \infty}\dist(\bx_t,\bx)\le  \frac{C\|\bx-\bx_{[s]}\|_1}{\sqrt{s}},\quad \forall \bx\in\mathbb{R}^n\label{sprmain}
    \end{align}
    w.h.p. using $\tilde{O}(s^3)$ measurements. Here, $\dist(\bu,\bv)=\min\{\|\bu+\bv\|_2,\|\bu-\bv\|_2\}$ is the standard error metric in phase retrieval.

    We also obtain a non-uniform version (Theorem \ref{thm:nonuni}): for any fixed $\bx\in \mathbb{R}^n$,  it holds w.h.p. that
    \begin{align}
        \lim_{t\to\infty}\dist(\bx_t,\bx)\le  C\|\bx-\bx_{[s]}\|_2.\label{sprnonueq}
    \end{align} 

    \item In one-bit compressed sensing (i.e., recovery of $\bx$ with a sparse prior from $\by=\sign(\bA\bx)$), our main result (Theorem \ref{thm:1bcsiop}) establishes the instance optimality of normalized binary iterative hard thresholding (\texttt{NBIHT}), which is essentially a procedure of normalized IHT with the ReLU loss:
    \begin{align}
        \sup_{t\ge 2\log(m/s)}\|\bx_t-\bx\|_2\le \tilde{O}\big(\frac{s}{m}\big)+O\bigg(\frac{\|\bx-\bx_{[s]}\|_1}{\sqrt{s}}\log\Big(\frac{1}{\|\bx-\bx_{[s]}\|_1/\sqrt{s}\wedge 1/2}\Big)\bigg),\quad \forall \bx\in\mathbb{S}^{n-1}.\label{1bcsmain}
    \end{align}
    Here, our techniques generate  a logarithmic factor of $\log(\frac{1}{\|\bx-\bx_{[s]}\|_1/\sqrt{s}\wedge 1/2})$, where $a\wedge b=\min\{a,b\}$, in the model error term.

    The non-uniform counterpart (Theorem \ref{thm:1bcsnonuiop}) gives that, for any fixed $\bx\in\mathbb{S}^{n-1}$, w.h.p., 
    \begin{align}
        \sup_{t\ge 2\log(m/s)}\|\bx_t-\bx\|_2\le \tilde{O}\big(\frac{s}{m}\big)+O(\|\bx-\bx_{[s]}\|_2).\label{1bcsnonueq}
    \end{align}  
\end{itemize}

Overall, our framework yields instance optimal guarantees for nonlinear sparse recovery that are comparable to the classical compressed sensing results in \eqref{csl2l111} and \eqref{nonulinear}. The analysis involves several high-dimensional probability tools, including concentration bounds, embedding guarantees, and restricted eigenvalue bounds. Some of these ingredients, especially Lemma \ref{lem:iophyper}, appear to be new and may be of independent interest.

\subsection{Related Works}\label{sec:relatedwork}
This paper is related to the classical works of instance optimality in compressed sensing, which we briefly recap in Equations (\ref{csgeneraliop})--(\ref{nonulinear}); see \cite[Sections 6 \& 11]{Foucart2013AMI} and references therein for further details. For more recent progress, see \cite{petersen2022robust} for instance. 
These developments are restricted to linear measurements. 
In the following, we shall review several lines of works that are most relevant to our paper. 

\paragraph{High-dimensional estimation via RAIC.} Provably optimal and computationally efficient algorithms have recently been obtained for several statistical estimation problems by establishing the restricted approximate invertibility condition (RAIC) for suitable gradient map. Such an RAIC requires that  
the gradient well approximates the ideal descent step over a low-dimensional set of structured signals (such as $\Sigma^n_s$), and implies both the convergence and the statistical performance of some nonconvex optimization algorithm. Worked examples include one-bit compressed sensing \citep{matsumoto2024binary,friedlander2021nbiht,chen2024optimal}, binary generalized linear models \citep{matsumoto2025learning}, nonlinear tensor estimation \citep{chen2025unified}, logistic regression \cite{chen2026finite}, one-bit phase retrieval \citep{chen2024one}, among others. 



\paragraph{Sparse phase retrieval (SPR).} 
 A variety of efficient solvers have been proposed for  recovering $\bx\in\Sigma^n_s$ from $\by=|
\bA\bx|$ (i.e., the SPR problem), including a convex program \citep{li2013sparse} and several nonconvex algorithms \citep{cai2013sparse,wang2017sparse,soltanolkotabi2019structured}, typically adapted from their counterparts for unstructured signals; see, e.g., \cite{candes2013phaselift,candes2015phase,wang2017solving,zhang2017nonconvex}. The best known guarantees for these algorithms  ensure the {\it non-uniform exact recovery}: if $m=\tilde{\Omega}(s^2)$, then for any fixed $\bx\in\Sigma^n_s$, w.h.p., some algorithm recovers $\pm \bx$ exactly. 
These exact recovery guarantees are not instance optimal, as they provide no recovery guarantee for $\bx\notin\Sigma^n_s$. To the best of our knowledge, no previous efficient sparse phase retrieval solver comes with guarantees for non-sparse (particularly approximately sparse) signals.

\noindent 
On the other hand, \cite{gao2016stable} used a strong RIP to show that the constrained $\ell_1$ minimization program
\begin{align}\label{l1minspr}
    \hat{\bx}_{\ell_1} = \arg\min_{\bu}\|\bu\|_1,\quad \textrm{subject to }
    |\bA\bu|=\by,
\end{align}
is $(\ell_2,\ell_1)$-instance optimal. In particular, when $m\gtrsim s\log(en/s)$, w.h.p., \[ \dist(\hat{\bx}_{\ell_1},\bx) 
    \le  \frac{C\|\bx-\bx_{[s]}\|_1}{\sqrt{s}},\quad \forall \bx\in\mathbb{R}^n.\]  
    This result was recently extended to the complex-valued setting in \cite{xia2025instance}. However, due to the nonconvex constraint, the program  \eqref{l1minspr} is in general computationally intractable.

    \noindent 
    To our best knowledge,   Theorems \ref{thm:uni} and \ref{thm:nonuni} in this paper provide the first  instance optimal guarantees for an efficient SPR solver, uniformly for all $\bx\in\mathbb{R}^n$ under $m=\tilde{\Omega}(s^3)$, and for a fixed $\bx\in\mathbb{R}^n$ under $m=\tilde{\Omega}(s^2)$. 

\paragraph{One-bit compressed sensing (1bCS).} 
Most existing 1bCS results concern the recovery of exactly sparse signals $\bx\in \Sigma^{n,*}_s:=\Sigma^n_s\cap \mathbb{S}^{n-1}$ from $\by=\sign(\bA\bx)$, in which the information-theoretic optimal uniform error rate reads 
\begin{align}
    \inf_{\hat{\bx}=\hat{\bx}(\bA,\by)}
    \sup_{\bx\in\Sigma^{n,*}_s}\|\hat{\bx}-\bx\|_2
    =\tilde{\Theta}\big(s/m\big) \label{1bcsuniformerrorrate}
\end{align} 
provided $m\gtrsim s\log(en/s)$ \citep{jacques2013robust}. It was recently shown by \cite{matsumoto2024binary} that this rate is achieved by the  normalized binary iterative hard thresholding  (\texttt{NBIHT}), that is,
\[
   \sup_{\bx\in\Sigma^{n,*}_s} \|\hat{\bx}_{\rm nbiht}-\bx\|_2 
    =  \tilde{O}(s/m),\quad \textrm{(w.h.p.),}\]
    where $\hat{\bx}_{\rm nbiht}$ denotes the output of NBIHT with sufficiently many iterations.

    \noindent 
To account for model error,  prior works considered the recovery of \emph{approximately sparse} signals in 
$\sqrt{s}\mathbb{B}_1^n\cap \mathbb{S}^{n-1}$. Under $m\gtrsim s\log(en/s)$, the information-theoretic optimal uniform error rate is
\begin{align}\label{1bcsapproximate}
     \inf_{\hat{\bx}=\hat{\bx}(\bA,\by)}\sup_{\bx\in\sqrt{s}\mathbb{B}_1^n\cap \mathbb{S}^{n-1}}
    \|\hat{\bx}-\bx\|_2
    =\tilde{\Theta} \bigg(\Big(\frac{s}{m}\Big)^{1/3}\bigg),
\end{align}
where the upper bound is achieved by  projected gradient descent \citep{chen2024optimal} and Adaboost \citep{chinot2022adaboost}, and the lower bound is established recently in  \cite{chen2026near}. However, (\ref{1bcsapproximate})  is not instance optimal as it gives the same worst-case error rate  for all signals. On the other hand, a desirable instance optimal guarantee should provide signal-dependent error bounds that recover the information-theoretic optimal rate $\tilde{O}(s/m)$ over $\bx\in \Sigma^{n,*}_s$, and deteriorate continuously as  the distance of $\bx$ from $\Sigma^n_s$ increases.  Our Theorem \ref{thm:1bcsiop} is the first such result for 1bCS, to the best of our knowledge. 

\subsection{Additional Related Works} 
We review several lines of additional related works.

\paragraph{Instance optimality under nonlinear measurements.}
\cite{keriven2018instance} pursued a similar goal of seeking instance optimality in   sparse recovery problems with nonlinear observations. However, the decoder is typically intractable and no non-asymptotic instance optimal guarantees are provided.  \cite{gao2016stable,xia2025instance} established instance optimality for a computationally intractable program in sparse phase retrieval, while \cite{chen2024robust} proved $(\ell_2,\ell_1)$ instance optimality for an efficient algorithm in phase-only compressed sensing. These results are tailored to specific nonlinear models and do not provide a unified approach.

\paragraph{Recovery of approximately sparse signals.} As with (\ref{1bcsapproximate}),
robustness to model error has indeed been studied under nonlinear observations by considering approximately sparse  signals living in a scaled $\ell_1$ ball. 
Beyond quantization, \cite{plan2016generalized,plan2017high} developed efficient algorithms for single-index models and established the rate $\tilde{O}(\sqrt{s/m})$ for $\bx\in\Sigma^{n,*}_s$, while a slower rate $\tilde{O}((s/m)^{1/4})$ for $\bx\in \sqrt{s}\mathbb{B}_1^n\cap\mathbb{S}^{n-1}$. These results concern the worst-case recovery error and are not instance-dependent. The present paper provides a different treatment to model error under nonlinear observations and typically yields improved error rate for signals being sufficiently close to $\Sigma^n_s$; see, e.g., Remark \ref{rem:1bcsimprove}.  

\paragraph{Embedding and hyperplane tessellation.} Most prior analyses   of nonlinear measurements relied on different high-dimensional embedding  phenomena \citep{plan2012robust,plan2013one,plan2014dimension,xu2020quantized,dirksen2021non,jacques2021importance}, and our work is no exception in this regard. Here we focus on hyperplane tessellations, or binary embeddings. 
Particularly, one asks how well a map from a set $\calU$ to the binary cube $\{-1,1\}^m$, induced by $m$ random hyperplanes, preserves Euclidean distances. The seminal work of \cite{plan2014dimension} showed that a number of Gaussian hyperplanes proportional to the Gaussian width of $\calU$ suffices for uniform tessellation over $\calU$. This result was sharpened and extended to local embeddings by \cite{oymak2015near}; see also recent developments in \cite{dirksen2021non,dirksen2022sharp,jung2021quantized,dirksen2024fast,chen2024one,dirksen2025resolution}. 
In this work, a key component of our   1bCS analysis is a new signal-dependent hyperplane tessellation result  that may be of independent interest.

\paragraph{Oracle inequality.} Our non-uniform instance optimal guarantees are closely related to oracle inequalities in high-dimensional statistics, which   decompose the prediction/estimation error into statistical error and model error. Consider the sparse linear regression model $\by = \bA\bx + \bm{\epsilon}$ with $\bA$ satisfying order-$s$ restricted eigenvalue condition and $\bm{\epsilon}\sim N(0,\sigma^2\bI_m)$. A representative  oracle inequality  states that the Lasso estimator satisfies, w.h.p.,  
\begin{align}
\frac{1}{m}\|\bA(\hat{\bx}_{\rm Lasso}-\bx)\|_2^2 \le  \inf_{\bu\in \Sigma^n_s}\frac{1}{m}\|\bA(\bu-\bx)\|_2^2 + \frac{C\sigma^2 s\log n}{m}. 
    \label{sharpeq}
\end{align}
See, e.g., \cite{koltchinskii2011nuclear,bellec2018slope}. In statistics, this shows that Lasso tolerates model misspecification: the true parameter $\bx$ need not  be $s$-sparse, and the price to pay is simply   $\inf_{\bu\in \Sigma^n_s}\frac{1}{m}\|\bA(\bu-\bx)\|_2^2$.

\noindent 
Mathematically, it is possible to write our results in a similar form. For instance, let $\{\bx_t\}_{t\ge 0}$ be the iterates of \texttt{NBIHT}, our 1bCS result in (\ref{1bcsnonueq}) already gives 
\[\sup_{t\ge 2\log(m/s)}\|\bx_t-\bx\|_2\le O\Big(\inf_{\bu\in\Sigma^n_s}\|\bu-\bx\|_2\Big)+\tilde{O}\Big(\frac{s}{m}\Big);\]
combining with the local binary embedding result from \citep{oymak2015near},  this further implies that 
\[\sup_{t\ge 2\log(m/s)}\frac{1}{m}\|\sign(\bA \bx_t)-\sign(\bA\bx)\|_2^2 \lesssim\inf_{\bu\in \Sigma^{n,*}_s}\frac{1}{m}\|\sign(\bA\bu)-\sign(\bA\bx)\|_2^2+\frac{s}{m}\]up to logarithmic factors.

\noindent
Yet we point out some important distinctions. First, most oracle inequalities  in statistics are for regularized empirical risk minimizers (or M-estimators) in sparse \emph{linear} regression, and technically, they are often established by       optimality conditions   and certain deterministic assumptions on the design. Also, the oracle inequalities may not be immediately achieved by a polynomial-time algorithm, for instance, when nonconvex regularizer is used. See \cite{bickel2009simultaneous,koltchinskii2011nuclear,bunea2007sparsity,elsener2018sharp,stucky2017sharp,bellec2018slope}, and see also \cite[Sections 9--10]{wainwright2019high} and references therein.  In contrast, our framework is developed for \emph{nonlinear} regression problems, and the guarantees are achieved by the computationally fast procedure of IHT, whose analysis is based on the RAIC. On the other hand, some of these works (e.g.   \cite{koltchinskii2011nuclear,bellec2018slope,stucky2017sharp}) established {\it sharp} oracle inequalities with the leading constant of the model error term being exactly $1$ (e.g., the above Equation (\ref{sharpeq})), which is beyond the scope of this work.


\subsection{Notation and Overview}
We use $C,C_i,c,c_i$ to denote universal constants whose values vary from line to line. For $T_1,T_2>0$, we write $T_1 \lesssim T_2$ or $T_1=O(T_2)$ to denote that $T_1\le CT_2$ for some constant $C>0$, and write $T_1\gtrsim T_2$ or $T_1=\Omega(T_2)$ to denote that $T_1\ge CT_2$ for some $C>0$; if both hold, then we write $T_1\asymp T_2$ or $T_1=\Theta(T_2)$. We shall use $\tilde{O}(\cdot),\,\tilde{\Omega}(\cdot),\,\tilde{\Theta}(\cdot)$ to further hide logarithmic factors in $m,n$ in $O(\cdot),\,\Omega(\cdot),\,\Theta(\cdot)$, respectively. The $\psi_2$ norm of a random variable $X$ is defined as $\|X\|_{\psi_2}= \inf\{t>0:\mathbb{E}\exp(X^2/t^2)\le 2\}$, and of a random vector $\bX\in\mathbb{R}^n$ is $\|\bX\|_{\psi_2}=\sup_{\bv\in \mathbb{S}^{n-1}}\|\bv^\top\bX\|_{\psi_2}$. For some sparsity level $s\in[n]:=\{1,2,\cdots,n\}$, $H_s:\mathbb{R}^n\to \mathbb{R}^n$ is the hard thresholding operator defined as $H_s(\bu)=\bu_{[s]}$, where ties are broken by selecting the support with the smallest indices. We use $\mathbb{B}_2^n$ to denote the unit $\ell_2$ ball in $\mathbb{R}^n$, and write $\mathbb{B}_2^n(\lambda) :=\lambda\mathbb{B}_2^n$, $\mathbb{B}_2^n(\bu;\lambda) = \bu + \lambda \mathbb{B}_2^n.$ More notation will be introduced when appropriate.  

 \paragraph{Overview.} In Section \ref{sec:unified}, we develop a unified framework for instance optimal sparse recovery via (normalized) IHT. The major applications of the framework appear in Sections \ref{sec3:spr}--\ref{sec:relu}, where we show that  (normalized) IHT is instance optimal for SPR, 1bCS and sparse ReLU regression. 
 We close the paper with concluding remarks in Section \ref{sec:conclu}. The complete proofs of our results appear in the appendices. 

\section{Unified Framework}\label{sec:unified}

\subsection{RAIC and (Normalized) IHT}\label{sec:aic_iht}
For the recovery of $\bx$, we construct a gradient map $\bh_{\bx}: \mathbb{R}^n\to \mathbb{R}^n$, so that $\bh_{\bx}(\bu)$ serves as the gradient at $\bu$. We focus on the algorithm of iterative hard thresholding (IHT)
\begin{align}\label{ihtalg}\tag{IHT}
    \bx_{t+1} = H_{s}(\bx_t - \eta \cdot\bh_{\bx}(\bx_t)),\quad t=0,1,\cdots,
\end{align}
for some sparsity level $s$.

If the signal lives in the unit sphere $\mathbb{S}^{n-1}$,\footnote{This can be generalized to $\bx\in \beta\mathbb{S}^{n-1}$ for some known $\beta>0$, i.e., the setting where the signal norm is known a priori.} we may utilize an additional normalization step, yielding normalized IHT (NIHT)
\begin{align}\label{nihtalg}\tag{NIHT}
    \bx_{t+1} = \frac{H_{s}(\bx_t - \eta \cdot\bh_{\bx}(\bx_t))}{\|H_{s}(\bx_t - \eta \cdot\bh_{\bx}(\bx_t))\|_2},\quad t=0,1,\cdots.
\end{align}

We define the top-$k$ $\ell_2$ norm as  
\begin{align*}
    \|\bv\|_{2,k} = \|\bv_{[k]}\|_2= \sup_{\bu\in \Sigma^{n,*}_k} \langle \bu ,\bv\rangle,
\end{align*}
  and then  define the  RAIC in the following.  
  
\begin{definition} \label{def:sparseraic}
For a fixed target signal $\bx$, we say that the gradient $\bh_{\bx}(\bu)$ satisfies the RAIC of order $k$ over a constraint set $\calU\subset\Sigma^n_k$ with an error function $R_{\bx}(\bu)$, if it holds that  
\begin{align*} 
    \big\|\bu-\bx-\bh_{\bx}(\bu)\big\|_{2,2k} \le R_{\bx}(\bu)\,,\quad \forall\bu\in\calU\,. 
\end{align*} 
We write this as \[\bh_{\bx}(\bu)\sim {\rm RAIC}(\Sigma^{n,*}_{2k};\calU,\bx,R_{\bx}(\bu)).\]
\end{definition}

The following results state that an RAIC with large enough constraint set (in the sense of (\ref{dlower11}), (\ref{dlower22})) and error function $\mu_1\|\bu-\bx\|_2+\mu_2$ yields the linear convergence of (N)IHT with a proper initialization. Their proofs are provided in Section \ref{sec:proofthm12}. 

\begin{theorem}[Convergence of \ref{ihtalg}]\label{raiciht}
     For a fixed $\bx\in \mathbb{R}^n$, suppose that \[\eta\bh_{\bx}(\bu)\sim {\rm RAIC}(\Sigma^{n,*}_{2s};\calU,\bx,\mu_1\|\bu-\bx\|_2+\mu_2).\] If
     $ \mu_1<\frac{1}{2}$ and  for some $0<R_{\bx}^{\rm loc}\le \infty$ 
     \begin{align} \label{dlower11}
        \calU\supset \Sigma^n_s \cap \mathbb{B}_2^n(\bx;R_{\bx}^{\rm loc}),\quad \textrm{where}~\,R_{\bx}^{\rm loc}>\frac{2\mu_2+3\|\bx-\bx_{[s]}\|_2}{1-2\mu_1}, 
     \end{align}
  then $\{\bx_t\}_{t\ge 0}$ generated by running \ref{ihtalg} \[\bx_{t+1}=H_s(\bx_t-\eta\cdot \bh_{\bx}(\bx_t)),\quad t\ge 0\]
     with $\bx_0\in \Sigma^n_s \cap \mathbb{B}_2^n(\bx;R_{\bx}^{\rm loc})$ satisfies 
\begin{align}\label{converge11}
    \|\bx_t-\bx\|_2\le (2\mu_1)^t \|\bx_0-\bx\|_2 + \frac{2\mu_2+3\|\bx-\bx_{[s]}\|_2}{1-2\mu_1},\quad t\ge 0. 
\end{align}     
\end{theorem}
\begin{theorem}[Convergence of \ref{nihtalg}] \label{raicniht} 
    For a target signal $\bx\in \mathbb{S}^{n-1}$, suppose that \[\eta\bh_{\bx}(\bu)\sim {\rm RAIC}(\Sigma^{n,*}_{2s};\calU,\bx,\mu_1\|\bu-\bx\|_2+\mu_2).\] 
    If $ \mu_1 < \frac{1}{4}$ and for some $0<R_{\bx}^{\rm loc}\le \infty$
    \begin{align}\label{dlower22}
       \calU\supset \Sigma^{n,*}_s\cap \mathbb{B}_2^n(\bx;R_{\bx}^{\rm loc}), \quad \textrm{where}~R_{\bx}^{\rm loc}> \frac{4\mu_2+6\|\bx-\bx_{[s]}\|_2}{1-4\mu_1},
    \end{align}
    then $\{\bx_t\}_{t\ge 0}$ generated by running \ref{nihtalg}
    \begin{align*}
        \bx_{t+1} = \frac{H_s(\bx_t-\eta\cdot \bh_{\bx}(\bx_t))}{\|H_s(\bx_t-\eta\cdot \bh_{\bx}(\bx_t))\|_2} ,\quad t\ge 0
    \end{align*}
    with $\bx_0\in \Sigma^{n,*}_s\cap \mathbb{B}_2^n(\bx;R_{\bx}^{\rm loc})$ satisfies 
    \begin{align}\label{converge22}
        \|\bx_t-\bx\|_2\le (4\mu_1)^t \|\bx_0-\bx\|_2 + \frac{4\mu_2+6\|\bx-\bx_{[s]}\|_2}{1-4\mu_1},\quad t\ge 0.  
    \end{align}
\end{theorem}
\subsection{From Signal-Dependent RAIC to Instance Optimality}\label{sec:unifieda}
Our goal is to use (N)IHT to achieve instance optimality
\begin{align}\label{uniformiop}
    \|\hat{\bx}_{\rm (N)IHT}-\bx\|_2 \le \calE_{\rm opt}(t;\bx) + \mathcal{E}_{\rm stat}+\mathcal{E}_{{\rm iop},s}(\bx),\quad \forall \bx\in  \bar{\calX},
\end{align}
where $\calE_{\rm opt}(t;\bx)$ denotes the optimization error that exponentially decays to $0$ when $t\to \infty$, $\mathcal{E}_{\rm stat}$ is the optimal statistical error that does not depend on the signal, and $\mathcal{E}_{{\rm iop},s}(\bx)$ characterizes the instance optimality of the algorithm. To be more concrete, in this paper, we establish $ \calE_{\rm opt}(t;\bx) =  O(\rho^t\|\bx\|_2)$ for some $\rho\in(0,1)$, indicating the linear convergence of (N)IHT. The statistical error, $\calE_{\rm stat}$, vanishes in SPR and ReLU regression, while is at the order   $\tilde{O}(\frac{s}{m})$ in 1bCS, matching the uniform recovery lower bound in (\ref{1bcsuniformerrorrate}). Moreover, $\mathcal{E}_{{\rm iop},s}(\bx)$ typically takes the form of
\begin{align}
    \tau_s(\bx) := \sum_{j\ge 1}\|\bx_{[(j+1)s]}-\bx_{[js]}\|_2.  
\end{align}


\subsubsection{(Uniform) Instance Optimality} \label{sec:frameuniiop}
For each  $\bx\in \bar{\calX}$ suppose that we work with the gradient map $\bh_{\bx}:\mathbb{R}^n\to \mathbb{R}^n$. 
Based on RAIC and its implications, we propose a unified framework for establishing instance optimal guarantee as in (\ref{uniformiop}), which  consists of the following steps:  
\begin{enumerate}
    \item \textbf{(Step 1: Decomposition of $\bar{\calX}$)} Choose $\calX \subset \bar{\calX}$ and separately treat $\bx\in \calX$ and $\bx\in \calX^c:= \bar{\calX}\setminus \calX$. In general, $\calX$ is the set of signals that are closer to $\Sigma^n_s$ and have small $\mathcal{E}_{{\rm iop},s}(\bx)$. More concretely, in the subsequent examples, we set 
    \begin{align}\label{unicalX}
        \calX= \calX_{\rm unif}:=\{\bu\in \bar{\calX}:\tau_s(\bu)\le c_*\|\bu\|_2\}
    \end{align}
    for some small enough universal constant $c_*>0.$

   
    
    
    
    \item \textbf{(Step 2: Instance optimality for $\bx\in\calX$)} The analysis of $\calX$ is based on establishing a signal-dependent RAIC for the gradients of all $\bx\in \calX$: 
    \begin{align}\label{iopraic}
        \eta\bh_{\bx}(\bu)\sim {\rm RAIC}\big(\Sigma^{n,*}_{2s};\calU_{\bx},\bx,\mu_1\|\bu-\bx\|_2+\mu_2+e(\bx)\big),\quad \forall\bx\in\calX,
    \end{align}
    where the error function is signal-dependent because of the term $e(\bx)$. 

    If $\bar{\calX}=\mathbb{R}^n$ and \ref{ihtalg} is considered, then we require that 
    \begin{gather}\label{con1iht}
\mu_1<\frac{1}{2}, 
    \\
        \calU_{\bx} \supset \Sigma^{n}_s \cap \mathbb{B}_2^n(\bx;R_{\bx}^{\rm loc})~~\textrm{for some }R_{\bx}^{\rm loc}\in(0,\infty], \\
         R_{\bx}^{\rm loc} > \frac{2\mu_2+2e(\bx)+3\|\bx-\bx_{[s]}\|_2}{1-2\mu_1}, \label{con3iht} \\
         \bx_0 \in \Sigma^n_s \cap \mathbb{B}_2^n (\bx;R_{\bx}^{\rm loc}) \label{iniiht}
    \end{gather}
    hold for all $\bx\in\calX$.

    If $\bar{\calX}=\mathbb{S}^{n-1}$ and \ref{nihtalg} is considered, then  we require that 
    \begin{gather}\label{con1niht}
    \mu_1<\frac{1}{4},
    \\ \label{nihtdx}
         \calU_{\bx} \supset \Sigma^{n,*}_s \cap \mathbb{B}_2^n(\bx;R_{\bx}^{\rm loc})~~\textrm{for some }R_{\bx}^{\rm loc}\in(0,\infty],\\\label{con3niht}
         R_{\bx}^{\rm loc}>\frac{4\mu_2+4e(\bx)+6\|\bx-\bx_{[s]}\|_2}{1-4\mu_1},\\
         \bx_0\in\Sigma^{n,*}_s\cap \mathbb{B}_2^n(\bx;R_{\bx}^{\rm loc}). \label{ininiht}
    \end{gather}
        hold for all $\bx\in\calX$. 

        
    \item \textbf{(Step 3: Instance optimality for $\bx\in\calX^c$)} The analysis of $\calX^c$ is based on the observation that the desired bound (\ref{uniformiop}) becomes a crude one due to the large $\mathcal{E}_{{\rm iop},s}(\bx)$. As such, a simple analysis, together with algorithmic regularization if needed, is sufficient. 
\end{enumerate}

To see why Step 2 leads to the instance optimality over $\bx\in\calX$,  we provide the following two statements. 

\begin{theorem}[\ref{ihtalg} is instance optimal for $\bx\in\calX$] \label{thm:ihtiopX} Assume that (\ref{iopraic}) and (\ref{con1iht})--(\ref{iniiht}) hold for all $\bx\in\calX$, then \ref{ihtalg} satisfies 
\begin{align}\label{ihtXconclusion}
    \|\bx_t-\bx\|_2\le (2\mu_1)^t\|\bx_0-\bx\|_2 +  \frac{2\mu_2+2e(\bx)+3\|\bx-\bx_{[s]}\|_2}{1-2\mu_1},~~ \forall t\ge 0,~\forall \bx\in \calX.
\end{align}
\end{theorem}
\begin{proof}
The statement follows by applying Theorem \ref{raiciht} to every $\bx\in\calX$.  
\end{proof}
\begin{theorem}[\ref{nihtalg} is instance optimal for $\bx\in\calX$] \label{thm:nihtiopX} Assume that (\ref{iopraic}) and (\ref{con1niht})--(\ref{ininiht}) hold for all $\bx\in\calX$, then \ref{nihtalg} satisfies
\begin{align}\label{nihtXconclusion}
    \|\bx_t-\bx\|_2\le (4\mu_1)^t\|\bx_0-\bx\|_2 + \frac{4\mu_2+4e(\bx)+6\|\bx-\bx_{[s]}\|_2}{1-4\mu_1},~~ \forall t\ge 0,~\forall\bx\in\calX. 
\end{align}
\end{theorem}
\begin{proof}
    The statement follows by applying Theorem \ref{raicniht} to every $\bx\in\calX$.  
\end{proof}

\begin{rem}
  In short, both algorithms linearly converge to $\ell_2$ error at the order of $O\big(\mu_2+e(\bx)+\|\bx-\bx_{[s]}\|_2\big)$, uniformly for all $\bx\in\calX$. Our hope is that   $O(\mu_2)$ captures the (optimal) statistical error as with $\mathcal{E}_{\rm stat}$ in (\ref{uniformiop}), while $O(e(\bx)+\|\bx-\bx_{[s]}\|_2)$ captures the instance optimality as with the approximation error $\mathcal{E}_{{\rm iop},s}(\bx)$ in (\ref{uniformiop}). For the latter, we point out that $O(\|\bx-\bx_{[s]}\|_2)$ is intrinsic to (N)IHT since the $s$-sparse iterates always satisfy 
\(\|\bx_t-\bx\|_2 \ge \|\bx-\bx_{[s]}\|_2.\)    
\end{rem} 
\begin{rem}
    If $R^{\rm loc}_{\bx}<\infty$, then we have to separately obtain a warm initializer obeying (\ref{iniiht}) or (\ref{ininiht}). This may require a separate procedure and independent analysis. 
\end{rem}

Since $\bx\in \calX^c$ is settled by a separate argument, we then arrive at (\ref{uniformiop}). See also Figure \ref{fig:proofscheme} for an intuitive explanation of our proof scheme.  
\begin{figure}[ht!] 
    \centering
    \includegraphics[width=0.5\linewidth]{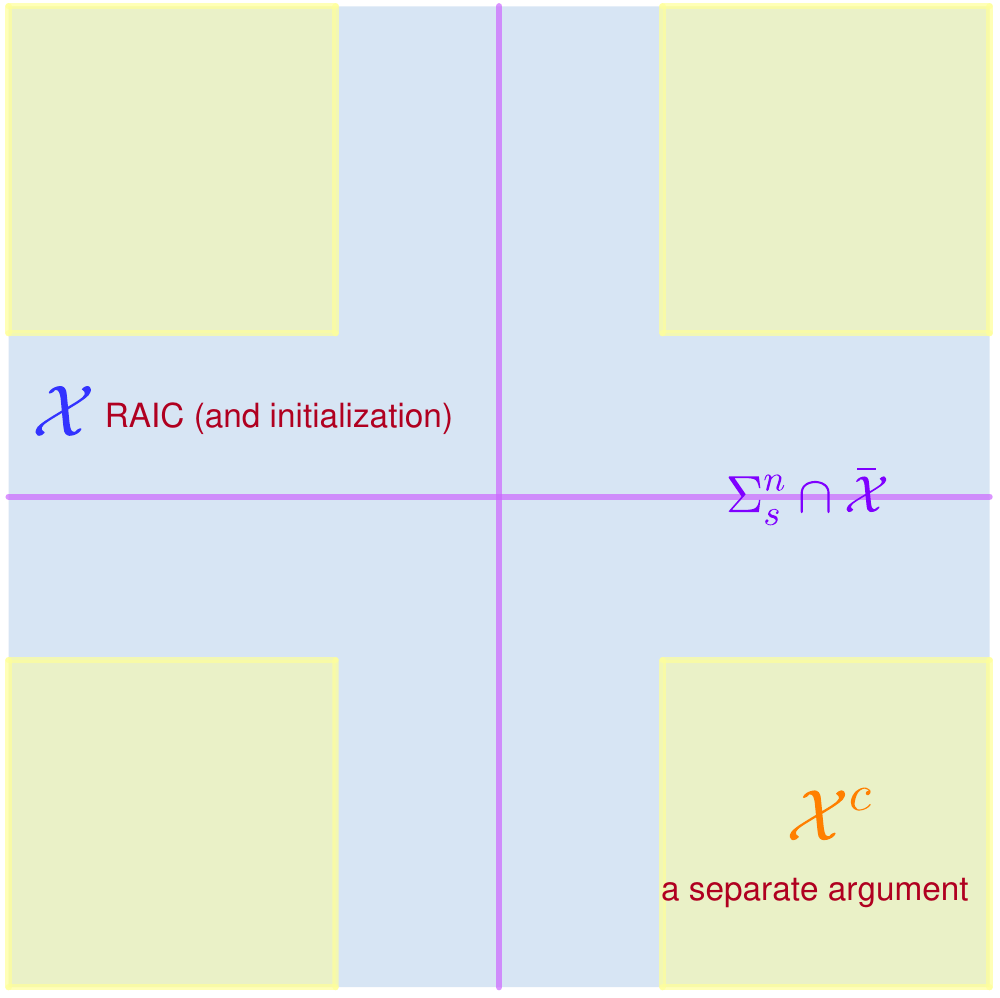} 
    \caption{Our approach to instance optimal sparse recovery. The figure  considers $\bar{\calX}=\mathbb{R}^2$   and shows   $\calX=\{\bu:\|\bu-\bu_{[1]}\|_2\le c_*\}$ (the blue region) as a ``relaxation'' of $\Sigma^2_1=\{\bu:\|\bu\|_0\le 1\}$ (the purple lines). The proposed approach treats the blue region $\calX$ by RAIC (and initialization analysis, if needed), then covers the yellow region $\calX^c$ by a separate argument.   
    \label{fig:proofscheme}}
\end{figure}

\subsubsection{Non-Uniform Instance Optimality}\label{sec:framenonuiop}
 We are also interested in the non-uniform instance optimality
\begin{align}\label{nonuiop}
    \|\hat{\bx}-\bx\|_2\le \calE_{\rm opt}(t;\bx)+ \mathcal{E}_{\rm stat} + \mathcal{E}_{{\rm iop},s}(\bx),\quad\textrm{for a fixed $\bx\in\bar{\calX}$}.
\end{align} 
This is weaker than (\ref{uniformiop}) in the sense that it only guarantees the recovery of a fixed $\bx$. On the other hand, the non-uniform instance optimality often allows for a sharper model error term $\calE_{\rm iop,s}(\bx)$ at the order of 
\begin{align}\label{nonunimodelerrorterm}
    \delta_{s}(\bx):= \|\bx-\bx_{[s]}\|_2. 
\end{align}

 Given the gradient map $\bh_{\bx}:\mathbb{R}^n\to \mathbb{R}^n$, we adapt our unified framework for (\ref{nonuiop}) in the following:  
\begin{enumerate}
    \item \textbf{(Step 1: Decomposition of $\bar{\calX}$)} Choose $\calX\subset \bar{\calX}$ as a class of approximately sparse signals. In subsequent concrete examples, we set
    \begin{align}
        \label{calXnonuniset}
        \calX= \calX_{\rm nonunif}:= \{\bu\in \bar{\calX}:\delta_s(\bu)\le c_*\|\bu\|_2\}
    \end{align}
    for some small enough universal constant $c_*>0.$
    \item \textbf{(Step 2: Instance optimality if $\bx\in\calX$)} If $\bx\in \calX$, then we seek to establish 
    \begin{align*}
        \eta\bh_{\bx}(\bu) \sim {\rm RAIC}(\Sigma^{n,*}_{2s};\calU_{\bx},\bx,\mu_1\|\bu-\bx\|_2+\mu_2+e(\bx)), 
    \end{align*}
    along with (\ref{con1iht})--(\ref{iniiht}) for \ref{ihtalg}, or along with (\ref{con1niht})--(\ref{ininiht}) for \ref{nihtalg}. 
    \item \textbf{(Step 3: Instance optimality if $\bx\in\calX^c$)} If $\bx\in \calX^c$, we provide a separate argument.  
\end{enumerate}

    While the framework is developed for nonlinear observations, as a proof-of-concept, readers may refer to \cite[Section 1]{chen2026supp} to see that the framework   applies to compressed sensing and recovers (\ref{csl2l111}) and (\ref{nonulinear}) for IHT.

\section{Sparse Phase Retrieval}\label{sec3:spr}
Sparse phase retrieval (SPR) concerns the recovery of $\bx$ with a sparse prior\footnote{This means that the signal is assumed to be sparse or approximately sparse, but need not belong exactly to $\Sigma^n_s$. Indeed, our desired instance-optimal guarantees apply to all $\bx\in\overline{\calX}$, for example all $\bx\in\mathbb{R}^n$ in SPR. The same convention applies in the rest of the paper.} from $\by=|\bA\bx|$, or equivalently, $\{y_i=|\ba_i^\top\bx|\}_{i=1}^m$. We work with  a Gaussian matrix $\bA=[\ba_1,\cdots, \ba_m]^\top\sim N^{m\times n}(0,1)$. To our best knowledge, there is no known   efficient and instance optimal SPR algorithm. The main aim of this section is to devise the first such algorithm.

\subsection{Algorithm}

Our algorithm is a variant of the thresholded amplitude flow \citep{soltanolkotabi2019structured,wang2017sparse}. See also \cite{zhang2017nonconvex,wang2017solving} for the counterparts of unstructured signals. 

\paragraph{\ref{ihtalg} with amplitude-based loss.} Note that a specific subgradient of the amplitude-based $\ell_2$ loss
\[\calL_{\rm amp}(\bu) = \frac{1}{2m}\sum_{i=1}^m\big(|\ba_i^\top\bu|-y_i\big)^2\]
is given by 
\[\partial \calL_{\rm amp}(\bu) = \frac{1}{m}\sum_{i=1}^m\big(|\ba_i^\top\bu|-y_i\big)\sign(\ba_i^\top\bu)\ba_i.\]
Substituting $y_i=|\ba_i^\top\bx|$, we set 
\begin{align}
    \bh_{\bx}(\bu):=\frac{1}{m}\sum_{i=1}^m\big(|\ba_i^\top\bu|-|\ba_i^\top\bx|\big)\sign(\ba_i^\top\bu)\ba_i.\label{hxpr}
\end{align}
We shall use a fixed step size $1$, and thus the \ref{ihtalg} is given by  \begin{align}
    \bx_{t+1} = H_{s}(\bx_t - \bh_{\bx}(\bx_t))\,,\quad t= 0 ,1,\cdots. \label{IHTproriginal}
\end{align}

\paragraph{Initialization.} In light of $\mathbb{E}y_i=\mathbb{E}|\ba_i^\top\bx|=\sqrt{2/\pi}\|\bx\|_2$, we let \[\lambda_{\bx}=\sqrt{\frac{\pi}{2}}\frac{1}{m}\sum_{i=1}^m y_i\] 
be the norm estimate for $\|\bx\|_2$. The mechanism for spectral initialization is to construct a data matrix whose leading eigenvector aligns well with the true signal. To this end, we define the truncation function as $$\sfT_\xi(a)  = \begin{cases}
    ~~~a~,\quad&\text{if}~~|a|\le \xi \\
    \xi\sign(a)~,\quad &\text{if}~~|a|> \xi
\end{cases}$$
and consider the matrix 
\begin{align}
    \hat{\bS}_{\bx} = \frac{1}{m}\sum_{i=1}^m \sfT_{2\lambda_{\bx}}(y_i)\ba_i\ba_i^\top.\label{hatSxxx}
\end{align}
To find its sparse leading eigenvector, we need an estimate of the support set of $\bx$; we shall use the size-$s$ index set $\calI_{\bx}$ corresponding to the largest $s$ diagonal entries of $\hat{\bS}_{\bx}$.  
Let $\calS_1\subset [m]$ and $\calS_2\subset[n]$, for $\bS\in\mathbb{R}^{m\times n}$, we shall let $[\bS]_{\calS_1,\calS_2}$ be the $m\times n$ matrix obtained from $\bS$ by setting the rows not in $\calS_1$ and the columns not in $\calS_2$ to zero. We now find the (normalized) leading eigenvector of $ [\hat{\bS}_{\bx}]_{\calI_{\bx},\calI_{\bx}}$, which we denote by $\bv_{\bx}$, as the estimate of the signal direction ${\bx}/{\|\bx\|_2}$. Taken collectively, we obtain an initial guess of $\bx$: $$\bx_0:= \lambda_{\bx}\bv_{\bx}.$$

\paragraph{An additional projection.} While the above \ref{ihtalg} and initialization are well developed in prior works (e.g., \cite{cai2016optimal,jagatap2019sample,soltanolkotabi2019structured}), we incorporate an additional projection into \ref{ihtalg}. Specifically, 
 we utilize  the projection onto the ball \[\mathbb{B}^n_2(2\lambda_{\bx}):=\{\bu\in \mathbb{R}^n:\|\bu\|_2\le 2\lambda_{\bx}\},\]
to handle signals that are distant from $\Sigma^n_s$ (precisely, signals in $\calX^c$ in our framework). Therefore, our final update rule is
\begin{align}\label{finalsprupdate}
    \bx_{t+1} = P_{\mathbb{B}_2^n(2\lambda_{\bx})}\big(H_{s}(\bx_t - \bh_{\bx}(\bx_t))\big)\,,\quad t= 0 ,1,\cdots
\end{align}
instead of (\ref{IHTproriginal}).

Combining these pieces, we are now ready to formally introduce our  algorithm, referred to as bounded thresholded amplitude flow (\texttt{BTAF}),  in Algorithm \ref{alg:btaf}. 
  
\begin{algorithm}[ht!]   
\caption{Bounded Thresholded Amplitude Flow (\texttt{BTAF})}\label{alg:btaf}
{\setstretch{1.15}        
\begin{algorithmic}[1]   
  \Require Matrix $\bA$, observations $y=|\bA\bx|$, sparsity level $s$

  \State \textbf{Step 1: Initialization}
  \State $\lambda_{\bx} = \sqrt{\tfrac{\pi}{2}}\;\tfrac{1}{m}\sum_{i=1}^m y_i$
  \State $\hat{\bS}_{\bx} = \tfrac{1}{m}\sum_{i=1}^m \sfT_{2\lambda_{\bx}}(y_i)\,\ba_i\ba_i^\top$
  \State $\calI_{\bx} \gets$ indices of the $s$ largest diagonal entries of $\hat{\bS}_{\bx}$
  \State $\bv_{\bx} \gets$ leading eigenvector of $[\hat{\bS}_{\bx}]_{\calI_{\bx},\calI_{\bx}}$
  \State $\bx_{0} = \lambda_{\bx}\,\bv_{\bx}$

  \State \textbf{Step 2: Refinement}
  \For{$t = 0,1,2,\dots$} 
    \State $\displaystyle
       \bx_{t+1} = P_{\mathbb{B}_2^n(2\lambda_{\bx})}
       \left(H_{s}\bigg(\bx_t - \frac{1}{m}\sum_{i=1}^m(|\ba_i^\top\bx_t|-y_i)\sign(\ba_i^\top\bx_t)\ba_i\bigg)\right)$
  \EndFor

  \Ensure Sequence $\{\bx_t\}_{t\ge 0}$
\end{algorithmic}
}
\end{algorithm}

\subsection{Instance Optimality}\label{sec:iopspr}
Under $\tilde{O}(s^3)$ Gaussian measurements, our main result shows the   instance optimality of \texttt{BTAF}. 

\begin{theorem}\label{thm:uni}Under Gaussian  matrix $\bA$,
    if \[m\gtrsim s^3\log\Big(\frac{en}{s}\Big)\log^2\bigg(\frac{em}{s\log(en/s)}\bigg) \]  then with probability at least $1-C_1\exp(-c_2s\log\frac{en}{s})$, the $\{\bx_t\}_{t\ge 0}$ generated by Algorithm \ref{alg:btaf} satisfies 
    \begin{align}\label{l2l1iop}
        \dist(\bx_{t},\bx) \le \frac{\|\bx\|_2}{2^t}+C_2\tau_s(\bx),\quad \forall\,t\ge 0,~\bx\in \mathbb{R}^n. 
    \end{align}  
    In particular, if $s$ is even and $s=2s_0$, then with the same probability, 
    \[\dist(\bx_t,\bx)\le \frac{\|\bx\|_2}{2^t}+\frac{C_3\|\bx-\bx_{[s_0]}\|_1}{\sqrt{s_0}},\quad \forall \,t\ge 0,~\bx\in\mathbb{R}^n.\]
\end{theorem}
\begin{proof}
    We give technical overview below and provide the complete proof in Appendix \ref{thm1}.  
\end{proof}
\begin{rem}
    [The $s^3$ complexity] The sample complexity $\tilde{O}(s^3)$ in Theorem \ref{thm:uni} exhibits an essential gap from the best known $\tilde{O}(s^2)$ for existing efficient algorithms to exactly recover any $s$-sparse signals. It is an open question whether $\tilde{O}(s^3)$ for efficient and instance optimal sparse phase retrieval is improvable. 
\end{rem} 

\subsubsection{Technical Overview for Theorem \ref{thm:uni}}  
   With no loss of generality, we can assume $\|\bx_0-\bx\|_2\le \|\bx_0+\bx\|_2$ and hence $\dist(\bx_0,\bx)=\|\bx_0-\bx\|_2$. (Otherwise, we can instead treat $-\bx$ as the true signal and use the same argument with minor adaptation.) For $\bh_{\bx}$ chosen as in (\ref{hxpr}), the result is proved by  following the three steps in the unified framework in Section \ref{sec:frameuniiop}: 
\begin{enumerate}
    \item {\bf (Decomposition of $\mathbb{R}^n$)} As with (\ref{unicalX}), we set 
    \(
        \calX:=\left\{\bu \in \mathbb{R}^n: \tau_s(\bu)\le c_*\|\bu\|_2\right\}\) 
    for some small enough $c_*$.  

    \item {\bf (Instance optimality for $\bx\in\calX$)} We first analyze (\ref{IHTproriginal}) and then transfer the result to (\ref{finalsprupdate}) by showing that $P_{\mathbb{B}_2^n(2\lambda_{\bx})}$ has no effect on the iterates associated with $\bx\in\calX$. 
    
    \textbf{RAIC.} For some universal constant $c>0$,  
    we show that the signal-dependent RAIC 
    \begin{align} \label{uraicspr}
        \bh_{\bx}(\bu)\sim {\rm RAIC}\bigg(\Sigma^{n,*}_{2s};\Sigma^n_s
        \cap \mathbb{B}_2^n(\bx;c\|\bx\|_2),\bx,\frac{\|\bu-\bx\|_2}{4}+4\tau_s(\bx)\bigg),\quad\forall\bx\in\calX\setminus\{0\}
    \end{align}  
    holds (w.h.p.). It is not hard to check that (\ref{con1iht})--(\ref{con3iht}) hold under sufficiently small $c_*$, with 
    \[\mu_1=\frac{1}{4},~~\mu_2=0,~~R_{\bx}^{\rm loc}=c\|\bx\|_2,\quad\text{and}~~e(\bx)=4\tau_s(\bx).\]

    \textbf{Initialization.} Moreover, we prove that under $m\gtrsim s^3$ (up to logarithmic factors), w.h.p., 
    \begin{align} \label{uniinispr}
        \bx_0 \in \Sigma^{n}_s\cap \mathbb{B}_2^n(\bx;c\|\bx\|_2),\quad \forall \bx\in \calX\setminus\{0\}.
    \end{align}  
    By Theorem \ref{thm:ihtiopX}, IHT in (\ref{IHTproriginal}) with $\bx_0$ from Algorithm \ref{alg:btaf} satisfies (\ref{ihtXconclusion}), i.e.,   
    \begin{align}
        \|\bx_t-\bx\|_2 \le \frac{c\|\bx\|_2}{2^t}+22\tau_s(\bx),\quad \forall t\ge 0,~\forall\bx\in\calX\setminus\{0\}\label{sprcompressible}
    \end{align}
    as desired.

    \textbf{Incorporating $P_{\mathbb{B}_2^n(2\lambda_{\bx})}$.} We further show that $P_{\mathbb{B}_2^n(2\lambda_{\bx})}$ has no effect on the iterates of $\bx\in\calX$. Therefore, $\{\bx_t\}_{t\ge 0}$ corresponding to $\bx\in\calX$ satisfy (\ref{sprcompressible}).

    \item {\bf (Instance optimality for $\bx\in\calX^c$)} Note that the desired error bound (\ref{l2l1iop}) for $\bx\in \calX^c$ scales as a constant. Consequently, a simple argument based on the projection onto $\mathbb{B}_2^n(2\lambda_{\bx})$ is sufficient for all these signals.  
\end{enumerate}

The two most technical parts lie in the proof of the RAIC and the analysis of the initialization.

\paragraph{Establishing the Uniform RAIC (\ref{uraicspr}).} The RAIC over $\bx\in\Sigma^n_s$ has been established in \cite[Theorem B.9]{chen2025unified}, hence we only need an extension to $\bx\in\calX$. The major effect of the model error turns out to be characterized by  the gradient mismatch term $\|\bh_{\bx}(\bu)-\bh_{\bx_{[s]}}(\bu)\|_{2s,2}$. Due to the Lipschitz continuity,   standard sparse decomposition and sparse eigenvalue   bound of $\bA$   establish 
\begin{align}
    \|\bh_{\bx}(\bu)-\bh_{\bx_{[s]}}(\bu)\|_{2s,2}\lesssim \tau_s(\bx),\quad \forall \bx\in \calX\setminus\{0\}.\label{mismatchbound}
\end{align} 

\paragraph{Establishing the Uniform Initialization Bound (\ref{uniinispr}).} Existing analyses of the spectral initialization only concern a fixed $\bx\in\Sigma^n_s$, while we have to establish (\ref{uniinispr}) for all $\bx\in\calX\setminus\{0\}$. By adapting some existing technicalities 
(see, e.g.,  \cite{jagatap2019sample}) and  accounting for the model error $\bx-\bx_{[s]}$, we establish a deterministic initialization error bound in Lemma \ref{lem:inideter}.   The problem then boils down to establishing (\ref{normaccc}), (\ref{operaconcen}), (\ref{maxconcen}) for all $\bx\in\calX\setminus\{0\}$, where $\tilde{\bS}_{\bx}$ defined in (\ref{idealdm}) serves as a surrogate of $\hat{\bS}_{\bx}$.

In view of $\lambda_{\bx} = \sqrt{\frac{\pi}{2}}\frac{\|\bA\bx\|_1}{m}$, it turns out that the uniform (\ref{normaccc}) follows from an $\ell_1$ embedding result due to \cite{plan2014dimension}. On the other hand, the analyses of uniform (\ref{operaconcen}) and  (\ref{maxconcen}) appear more delicate, where our aims are to show 
\begin{gather}\label{sparserestrict}
    \sup_{\bx\in\calX\setminus\{0\}}\sup_{\bu\in\Sigma^{n,*}_s}\bigg|\bu^\top \bigg(\frac{\hat{\bS}_{\bx}-\mathbb{E}[\tilde{\bS}_{\bx}]}{\|\bx\|_2}\bigg)\bu\bigg|\lesssim1\,,\\  \label{ejrestrict}
     \sup_{\bx\in\calX\setminus\{0\}}\sup_{\bu\in\{\be_j:j\in[n]\}}\bigg|\bu^\top \bigg(\frac{\hat{\bS}_{\bx}-\mathbb{E}[\tilde{\bS}_{\bx}]}{\|\bx\|_2}\bigg)\bu\bigg|\lesssim \frac{1}{s}\,.
\end{gather}
Note that these concern the restricted eigenvalues of a set of random matrices $\{(\hat{\bS}_{\bx}-\mathbb{E}[\tilde{\bS}_{\bx}])/\|\bx\|_2:\bx\in\calX\setminus\{0\}\}$. 
By empirical process theory  we establish in Lemma \ref{lem:Phixuniformbound} the following: 
\begin{align}\label{unisparseRE}
    \sup_{\bx\in\calX\setminus\{0\}}\sup_{\bu\in\Sigma^{n,*}_s}\bigg|\bu^\top \bigg(\frac{\hat{\bS}_{\bx}-\mathbb{E}[\tilde{\bS}_{\bx}]}{\|\bx\|_2}\bigg)\bu\bigg| =\tilde{O}\bigg(\sqrt{\frac{s}{m}}\bigg);
\end{align}  
As it turns out, (despite a smaller range of $\bu$) we are not able to establish an essentially sharper bound on the empirical process in (\ref{ejrestrict}). Hence, we directly reuse the bound $\tilde{O}(\sqrt{s/m})$ for it. Therefore, ignoring logarithmic factors, (\ref{ejrestrict}) is ensured by  \(\sqrt{\frac{s}{m}}\lesssim \frac{1}{s},\) 
which precisely dictates the $s^3$ sample size.

\subsection{Non-Uniform Instance Optimality}
We further establish a non-uniform instance optimal guarantee for any fixed $\bx\in\mathbb{R}^n$ under $\tilde{O}(s^2)$ measurements.  

\begin{theorem}\label{thm:nonuni}
    Let $\bx$ be any fixed signal in $\mathbb{R}^n$.  
    Under Gaussian matrix $\bA$, if 
    $m\gtrsim s^2\log n$, then with probability at least $1-C_1n^{-1}-C_2\exp(-c_3s\log(en/s))$,  $\{\bx_t\}_{t\ge 0}$ generated by Algorithm \ref{alg:btaf} satisfies 
    \begin{align}\label{l2l2iop}
        \dist(\bx_{t},\bx) \le \frac{\|\bx\|_2}{2^t}+ C\delta_{s}(\bx),\quad \forall \, t\ge 0\,.
    \end{align}   
\end{theorem}
\begin{proof}
    The proof can be found in Appendix \ref{thm2}. 
\end{proof}
\begin{rem}
    When transitioning to non-uniform instance optimality, we have an improvement on the model error term $\mathcal{E}_{{\rm iop},s}(\bx)$, that is, $O(\delta_{s}(\bx))$ in (\ref{l2l2iop}) which is tighter than $O(\tau_{s}(\bx))$ in (\ref{l2l1iop}). This is consistent with existing instance optimal guarantees for compressed sensing recalled in (\ref{csl2l111}) and (\ref{nonulinear}).\label{rem:sharpernonuni}
\end{rem}
        
    \begin{rem}
         Note that $\tilde{O}(s^2)$ coincides with the best-known sample complexity of computationally efficient algorithms that guarantee exact recovery for arbitrary \(s\)-sparse signals. In this sense, instance optimality for a fixed signal comes essentially for free. 
    \end{rem}



The proof of Theorem \ref{thm:nonuni} largely follows from that of Theorem \ref{thm:uni}, except that the  RAIC and initialization bound only need to hold for a fixed $\bx$. As such, we shall only discuss some main distinctions. In the RAIC, the effect of model error is again captured by $\|\bh_{\bx}(\bu)-\bh_{\bx_{[s]}}(\bu)\|_{2s,2}$, while
for a fixed $\bx$ this satisfies a sharper bound than (\ref{mismatchbound}): \[\|\bh_{\bx}(\bu)-\bh_{\bx_{[s]}}(\bu)\|_{2s,2}\lesssim \|\bx-\bx_{[s]}\|_2.\] 
This is exactly what improves $\tau_s(\bx)$ to $\delta_s(\bx)$. For the initialization, we only need to guarantee the following non-uniform counterpart of (\ref{sparserestrict})--(\ref{ejrestrict}):
\begin{align}\label{nonunirestrictspr}\sup_{\bu\in\Sigma^{n,*}_s}\bigg|\bu^\top \bigg(\frac{\hat{\bS}_{\bx}-\mathbb{E}[\tilde{\bS}_{\bx}]}{\|\bx\|_2}\bigg)\bu\bigg|\lesssim1,\qquad \sup_{\bu\in\{\be_j:j\in[n]\}}\bigg|\bu^\top \bigg(\frac{\hat{\bS}_{\bx}-\mathbb{E}[\tilde{\bS}_{\bx}]}{\|\bx\|_2}\bigg)\bu\bigg|\lesssim \frac{1}{s}.\end{align}
Due to the transition from ``all $\bx$''  to ``a fixed $\bx$'', we establish substantially tighter bound 
\begin{gather*}
     \sup_{\bu\in\{\be_j:j\in[n]\}}\bigg|\bu^\top \bigg(\frac{\hat{\bS}_{\bx}-\mathbb{E}[\tilde{\bS}_{\bx}]}{\|\bx\|_2}\bigg)\bu\bigg| =\tilde{O}\bigg(\frac{1}{\sqrt{m}}\bigg).
\end{gather*}
As such, ignoring logarithmic factors, we need $\frac{1}{\sqrt{m}}\le \frac{1}{s}$ to ensure the second condition in (\ref{nonunirestrictspr}), which determines the $s^2$ sample complexity.


\section{One-Bit Compressed Sensing}\label{sec4:1bcs}
One-bit compressed sensing (1bCS) concerns the recovery of $\bx\in \bar{\calX}= \mathbb{S}^{n-1}$ with a sparse prior from $\by=\sign(\bA\bx)$, where $\bA\sim N^{m\times n}(0,1)$ is a Gaussian matrix. 
The most relevant developments of 1bCS have been reviewed in Section \ref{sec:relatedwork}. Here, we establish the instance optimality of \texttt{NBIHT}.  

\subsection{Algorithm}
We shall show that \texttt{NBIHT} is a specific instance of our \ref{nihtalg} algorithm. We adopt the ReLU loss function 
\[\calL_{\rm ReLU}(\bu) = \frac{1}{2m}\sum_{i=1}^m\big(|\ba_i^\top\bu|-y_i\ba_i^\top\bu\big)\]
which is a convex relaxation of the hamming distance loss (or $0$-$1$ loss). A subgradient of the ReLU loss is given by 
\[\partial\calL_{\rm ReLU}(\bu) = \frac{1}{2m}\sum_{i=1}^m(\sign(\ba_i^\top\bu)-y_i)\ba_i.\]
Combining with $y_i=\sign(\ba_i^\top\bx)$, we shall set
\begin{align}
    \bh_{\bx}(\bu) = \frac{1}{2m}\sum_{i=1}^m (\sign(\ba_i^\top\bu)-\sign(\ba_i^\top\bx))\ba_i.\label{1bcsgradient}
\end{align}
Taking the step size $\eta=\sqrt{2\pi}$, our \ref{nihtalg} specializes to the following algorithm, referred to as \texttt{NBIHT} in the literature.

\begin{algorithm}[ht!]   
\caption{Normalized Binary Iterative Hard Thresholding (\texttt{NBIHT})}\label{alg:nbiht}
{\setstretch{1.15}        
\begin{algorithmic}[1]   
  \Require Matrix $\bA$, observations $\by=\sign(\bA\bx)$, sparsity level $s$

  \State \textbf{Step 1: Initialization}
  \State Let $\bx_0 = \be_1$

  \State \textbf{Step 2: Refinement}
  \For{$t = 0,1,2,\dots$}
    \State $\displaystyle \bx_{t+1} = \frac{H_s\big(\bx_t - \sqrt{\frac{\pi}{2}}\frac{1}{m}\sum_{i=1}^m(\sign(\ba_i^\top\bx_t)-y_i)\ba_i\big)}{\big\|H_s\big(\bx_t - \sqrt{\frac{\pi}{2}}\frac{1}{m}\sum_{i=1}^m(\sign(\ba_i^\top\bx_t)-y_i)\ba_i\big)\big\|_2},
         $
  \EndFor

  \Ensure Sequence $\{\bx_t\}_{t\ge 0}$
\end{algorithmic}
}
\end{algorithm}
\subsection{Instance Optimality}\label{sec:iop1bcs1}
Our main result in this section shows  the instance optimality of \texttt{NBIHT}. 
\begin{theorem}\label{thm:1bcsiop}
    Under Gaussian matrix $\bA$, if
 $
        m\gtrsim s\log \frac{en}{s}$, 
    then with probability at least $1-\exp(-cs\log\frac{en}{s})$, the output of Algorithm \ref{alg:nbiht} satisfies 
    \begin{align*}
        \|\bx_t-\bx\|_2 \le 2\Big(\frac{2}{5}\Big)^t+ \frac{C_1s}{m} \log\Big(\frac{mn}{s^2}\Big)\log^{1/2}\Big(\frac{m}{s}\Big)  + C_2 \tau_s(\bx) \log\bigg(\frac{1}{\tau_s(\bx)\wedge \frac{1}{2}}\bigg) ,\quad \forall t\ge 0,~\forall \bx\in\mathbb{S}^{n-1}.  
    \end{align*} 
   In particular, if $s$ is even and $s=2s_0$, then 
    \begin{align*}
        \|\bx_t-\bx\|_2&\le 2\Big(\frac{2}{5}\Big)^t + \frac{C_3s_0}{m}\log\Big(\frac{mn}{s_0^2}\Big)\log^{1/2}\Big(\frac{m}{s_0}\Big)\\
        &+ \frac{C_4\|\bx-\bx_{[s_0]}\|_1}{\sqrt{s_0}}\log\bigg(\frac{1}{s_0^{-1/2}\|\bx-\bx_{[s_0]}\|_1\wedge \frac{1}{2}}\bigg),\quad \forall t\ge 0,~\forall\bx\in \mathbb{S}^ {n-1}. 
    \end{align*}
\end{theorem}
\begin{proof}
    We provide technical overview below and provide the complete proof in Appendix \ref{app:iop1bcsproof}.
\end{proof}
\begin{rem}
  Unlike in SPR, it is not possible to perfectly recover $\bx$  from the one-bit measurements (even if $\bx\in \Sigma^n_s$). Therefore, the statistical error term $O(\frac{s}{m}\log(\frac{mn}{s^2})\log^{1/2}(\frac{m}{s}))$ appears in the error rate, which is optimal up to logarithmic factors \citep{jacques2013robust}. To our best knowledge, \texttt{NBIHT} is the only   known efficient algorithm     that provably achieves the error rate $\tilde{O}(\frac{s}{m})$ over $s$-sparse signals, due to the recent work of Matsumoto and Mazumdar \citep{matsumoto2024binary}. Our Theorem \ref{thm:1bcsiop} upgrades their main result to an instance optimal one. 
\end{rem}

 \begin{rem}
Note that the model error term is consistent with compressed sensing (\ref{csl2l111}) and sparse phase retrieval (Theorem \ref{thm:uni}) but exhibits an extra logarithmic factor. In fact, some of the arguments in our proofs necessarily generate additional logarithmic factor. It is an open question whether this logarithmic factor can be avoided by a more refined argument.        
 \end{rem}

\subsubsection{Technical Overview for Theorem \ref{thm:1bcsiop}}   
For the gradient map in (\ref{1bcsgradient}), the proof follows the three steps in the unified framework:
\begin{enumerate}
    \item {\bf (Decomposition of $\mathbb{S}^{n-1}$)} For some small enough universal constant $c_*>0$, we set \(\calX:=\{\bu\in \mathbb{S}^{n-1}:\tau_s(\bu)\le c_*\}\) as in (\ref{unicalX});  
    \item {\bf (Instance optimality for $\bx\in\calX$)} For $\bx\in\calX$, we establish the signal-dependent RAIC     
    \begin{align}\nn
        &\sqrt{2\pi}\bh_{\bx}(\bu)\sim {\rm RAIC}\bigg(\Sigma^{n,*}_{2s};\Sigma^{n,*}_s,\bx,\frac{\|\bu-\bx\|_2}{10} + \frac{C_1s}{m}\log\Big(\frac{mn}{s^2}\Big)\log^{1/2}\Big(\frac{m}{s}\Big)+ e(\bx)\bigg),~\forall \bx\in\calX, \\
        &\textrm{where}~e(\bx) \asymp \tau_s(\bx) \log\bigg(\frac{1}{\tau_s(\bx)}\bigg). \label{1bcsraic}
    \end{align}  
    This satisfies (\ref{con1niht})--(\ref{ininiht}) with 
    \(\mu_1=\frac{1}{10},~R_{\bx}^{\rm loc}=\infty,~\bx_0=\be_1\) for all $\bx\in \calX$, 
    and therefore yields (\ref{nihtXconclusion}) for the \texttt{NBIHT} iterates $\{\bx_t\}_{t\ge 0}$: 
    \begin{align}\label{nbihtXbound}
        \|\bx_t-\bx\|_2\le 2\Big(\frac{2}{5}\Big) ^t + O\Big(\frac{s}{m}\log\Big(\frac{mn}{s^2}\Big)\log^{1/2}\Big(\frac{m}{s}\Big)+e(\bx)\Big),\quad \forall t\ge 0,\,\bx\in\calX. 
    \end{align}
    
    \item {\bf (Instance optimality for $\bx\in\calX^c$)} For $\bx\in\calX^c$, the crude bound $\|\bx_t-\bx\|_2\le 2$ is enough, in light of \(
        \|\bx_t-\bx\|_2\le 2\le \frac{2}{c_*}\tau_s(\bx).\)
\end{enumerate}

The main bulk of techniques is devoted to establishing (\ref{1bcsraic}) for all $\bx\in\calX$. We shall build on the RAIC for all $\bx\in \Sigma^{n,*}_s$ recently established in \cite{matsumoto2024binary} to prove the optimality of \texttt{NBIHT}. Here, we quote a statement from \cite{chen2024optimal} that is closer to our development (see Lemma \ref{lem:1bcsraic1}), which delivers 
\begin{align}
\big\|\bu-\bx-\sqrt{2\pi}\cdot \bh_{\bx}(\bu)\big\|_{2s,2}\le \frac{\|\bu-\bx\|_2}{10}+\tilde{O}\bigg(\frac{s}{m}\bigg),\quad \forall \bu,\bx\in\Sigma^{n,*}_s.
\end{align} 
It remains to seek an extension to $\bx\in\calX$. Like in the proof of Theorem \ref{thm:uni}, the effect of model error  turns out to be mainly driven by   the gradient mismatch term \begin{align}
    \|\bh_{\bx}(\bu)-\bh_{\bx_{[s]}}(\bu)\|_{2s,2}=\bigg\|\frac{1}{2m}\sum_{i=1}^m(\sign(\ba_i^\top\bx)-\sign(\ba_i^\top\bx_{[s]}))\ba_i\bigg\|_{2s,2}. \label{gmismatch1bcs}
\end{align} Yet, the discontinuity of $\sign$ forbids a simple Lipschitzness-based argument as in SPR.

Observe that the $i$-th contributor is non-zero only if $\sign(\ba_i^\top\bx)\ne \sign(\ba_i^\top\bx_{[s]})$, our strategy is to bound \[L_{\bx}:=|\{i\in[m]:\sign(\ba_i^\top\bx)\ne \sign(\ba_i^\top\bx_{[s]})\}|\]  This bound, combined with Lemma \ref{lem:maxlsum},  delivers a worst-case upper bound. Following this approach, we have to establish signal-dependent bound on $L_{\bx}$. This seems plausible in light of \[L_{\bx}\sim {\rm Binomial}(m,\pi^{-1}\arccos(\langle \bx,\bx_{[s]}/\|\bx_{[s]}\|_2\rangle)),\] which concentrates about $m\pi^{-1}\arccos(\langle \bx,\bx_{[s]}/\|\bx_{[s]}\|_2\rangle)$.

\paragraph{Limitation of existing results.} Geometrically, the study of $L_{\bx}$ can be interpreted as a \emph{hyperplane tessellation} problem: arrange $m$ Gaussian hyperplanes, how many of them separate $\bx$ and $\bx_{[s]}$, uniformly over all $\bx\in\calX$? However, existing results (see, e.g., \cite{plan2014dimension,oymak2015near,dirksen2021non,dirksen2022sharp}) are not instance-dependent but rather provide a uniform bound that is identical for all $\bx$ of interest, and hence do not serve our purpose. For instance,  Theorem 2.2 of \cite{oymak2015near} yields the following high-probability concentration bound on   $S_{\bu,\bv}:=|\{i\in[m]:\sign(\ba_i^\top\bu)\ne\sign(\ba_i^\top\bv)\}|$: 
\begin{align}
    \bigg|\frac{S_{\bu,\bv}}{m} - \pi^{-1}\arccos(\langle \bu,\bv\rangle)\bigg|=\tilde{O}\bigg(\big(\frac{s}{m}\big)^{1/4}\bigg),\quad \forall \bu,\bv\in\calX.
\end{align}
This implies  
\begin{align}
    \frac{L_{\bx}}{m}=\frac{S_{\bx,\bx_{[s]}/\|\bx_{[s]}\|_2}}{m}\lesssim \|\bx-\bx_{[s]}\|_2+\tilde{O}\bigg(\big(\frac{s}{m}\big)^{1/4}\bigg),\quad\forall\bx\in\calX\label{existinghyper}
\end{align}  in our context. However, this is not sufficient for establishing Theorem \ref{thm:1bcsiop}, since
the concentration error term $\tilde{O}((s/m)^{1/4})$ carries over to the statistical error and compromises entirely our ability to attain the optimal statistical rate $\tilde{O}(s/m)$ for $\bx\in\Sigma^{n,*}_s$.

\paragraph{Technical contribution.} We establish a novel instance-dependent hyperplane tessellation bound 
\begin{align}
    \frac{L_{\bx}}{m}\le  \tilde{O}(\tau_s(\bx))+\tilde{O}\Big(\frac{s}{m}\Big),\quad \forall \bx\in\calX;\label{eqiophyper}
\end{align}
see Lemma \ref{lem:iophyper}. The constant term $\tilde{O}(s/m)$  is essentially sharper than $\tilde{O}((s/m)^{1/4})$ in (\ref{existinghyper}) and allows us to retain the optimal statistical error for all $\bx\in\calX$. While the price to pay is the degradation from $O(\|\bx-\bx_{[s]}\|_2)$ to $\tilde{O}(\tau_s(\bx))$ in the model error, this is nonetheless consistent with the model error terms in instance optimal compressed sensing and SPR, up to logarithmic factors.

In view of $L_{\bx} = S_{\bx,\bx_{[s]}/\|\bx_{[s]}\|_2}$, the main inspiration for our establishment of (\ref{eqiophyper}) is the \emph{asymmetry} between $\bx$ and $\bx_{[s]}/\|\bx_{[s]}\|_2$: 
\[\bx\in\calX,\qquad \bx_{[s]}/\|\bx_{[s]}\|_2\in \Sigma^{n,*}_s,\]
and in a certain sense $\Sigma^{n,*}_s$   is a much thinner set than $\calX$. Our developments generalize to $\{S_{\bu,\bv}:\bu\in \calW,~\bv\in \calW'\}$ for some $\calW'$ being much thinner than $\calW$, which is in contrast to $\{S_{\bu,\bv}:\bu,\bv\in\calW\}$ treated in prior works; see Figure \ref{fig:twofigures} for an intuitive illustration. 

\begin{figure}[ht!]
    \centering
    \includegraphics[width=0.3\linewidth]{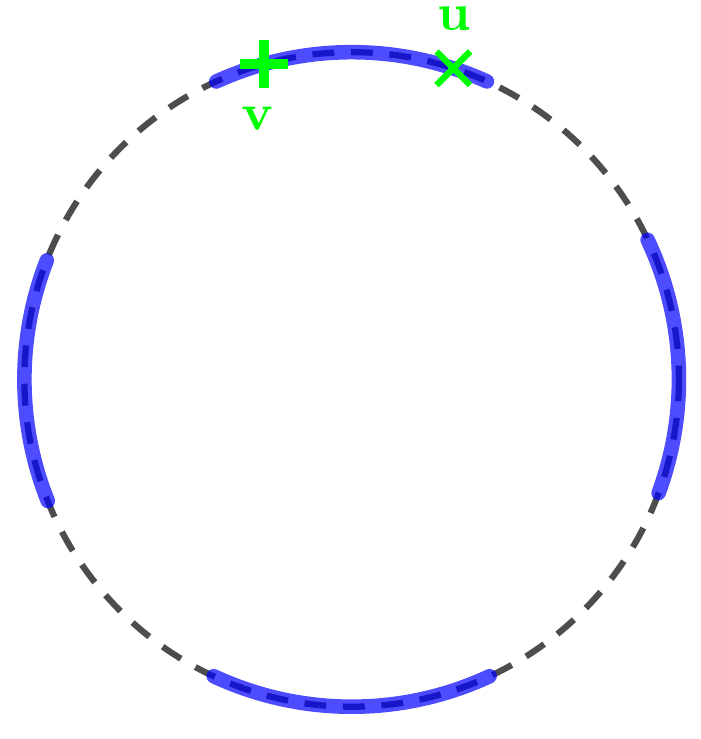}
    \qquad 
    \includegraphics[width=0.3\linewidth]{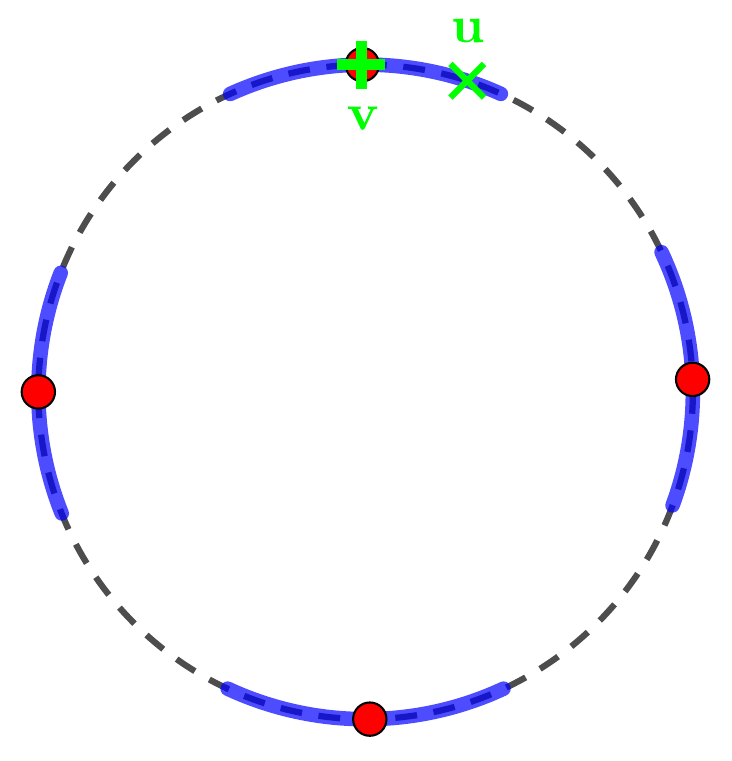}
    \caption{The blue set $\calW$ represents approximately sparse vectors, while the red set $\calW':=\Sigma^{2,*}_1$ is a smaller subset. Classical hyperplane tessellation results relate $S_{\bu,\bv}$ to $\pi^{-1}\arccos(\langle \bu,\bv\rangle)$ when both $\bu,\bv$ lie in  $\calW$ (left), while our argument exploits the asymmetry that $\bv$ lies in the smaller set $\calW'$ (right).}
    \label{fig:twofigures}
\end{figure}

 The rest of this proof remains rather involved. In a nutshell,
 we note that  \begin{align} L_{\bx}&\le \bigg|\bigg\{i\in[m]: \bigg|\ba_i^\top\frac{\bx_{[s]}}{\|\bx_{[s]}\|_2}\bigg|\le \eta\bigg\}\bigg|  +  \bigg|\bigg\{i\in[m]:\bigg|\ba_i^\top\bigg(\bx-\frac{\bx_{[s]}}{\|\bx_{[s]}\|_2}\bigg)\bigg|>\frac{\eta}{2}\bigg\}\bigg|,\quad \forall \eta>0,
 \label{boundLxstart}\end{align}
 and separately bound the two terms   over a wide enough range of  $\eta$. Finally, we choose $\eta$ as a function of $\bx\in\calX$ to render a signal-dependent bound. See Appendix \ref{app:iop1bcsproof} for further details.



\subsection{Non-Uniform Instance Optimality}\label{sec:nonu1bcs}
  We have the following non-uniform instance optimality in the recovery of a fixed $\bx\in \mathbb{S}^{n-1}$. Unlike Theorem \ref{thm:1bcsiop}, the above result does not lose logarithmic factor in the model error  term.

\begin{theorem} \label{thm:1bcsnonuiop}
    Fix $\bx\in \mathbb{S}^{n-1}$. Under Gaussian matrix $\bA$, if
 $m\gtrsim s\log \frac{en}{s}$, 
    then with probability at least $1-\exp(-cs\log\frac{en}{s})$, for any $t\ge 0$, 
    \begin{align*}
        \|\bx_t-\bx\|_2 \le 2\Big(\frac{2}{5}\Big)^t + \frac{C_1s}{m}\log\Big(\frac{en}{s}\Big)\log^{1/2}\Big(\frac{m}{s}\Big) + C_2\delta_s(\bx).
    \end{align*}
\end{theorem}
 \begin{proof}
      The proof is given in Appendix \ref{app:prove1bcsnonuni}. 
 \end{proof}
\begin{rem}\label{rem:1bcsimprove}    
  Previous works handle model error by treating the  approximately $s$-sparse signals in $\sqrt{s}\mathbb{B}_1^n\cap \mathbb{S}^{n-1}$, with the near-optimal uniform error rate being $\tilde{O}((s/m)^{1/3})$; see (\ref{1bcsapproximate}).  Ignoring logarithmic factors, our   Theorem \ref{thm:1bcsiop}  improves on (\ref{1bcsapproximate}) over all signals $\bx$ obeying 
  $\tau_s(\bx) \lesssim (\frac{s}{m})^{1/3}$,  and likewise our non-uniform Theorem \ref{thm:1bcsnonuiop} gives a faster rate for $\bx$ satisfying  $\delta_s(\bx)\lesssim (\frac{s}{m})^{1/3}$. 
\end{rem}

Compared to the proof of Theorem \ref{thm:1bcsiop},
the main distinction lies in the establishment of the RAIC, especially in bounding the gradient mismatch term (\ref{gmismatch1bcs}). Indeed we can reuse the previous strategy: for the fixed $\bx$,     a simple Chernoff bound   yields  \(L_{\bx}\lesssim s\log\frac{en}{s}+m\delta_s(\bx),\)  and then   Lemma \ref{lem:maxlsum} gives the bound $O(\delta_s(\bx)\log^{1/2}(e/\delta_s(\bx)))$. However, this exhibits an additional logarithmic factor compared to the desired $O(\delta_s(\bx))$ in (\ref{desiredraicnonu1bcs}). The removal of this logarithmic factor requires a more fine-grained analysis, which first   decomposes the gradient into the components along the orthogonal directions $\bx-\frac{\bx_{[s]}}{\|\bx_{[s]}\|_2}$, $\bx+\frac{\bx_{[s]}}{\|\bx_{[s]}\|_2}$, and ${\rm span}(\bx,\bx_{[s]})^\perp$, and then establishes concentration inequalities on each of the three components; see Appendix \ref{app:prove1bcsnonuni}. Note that similar argument  has appeared in \cite{matsumoto2024binary} for a different purpose.


\section{Sparse ReLU Regression}\label{sec:relu} 
 Sparse ReLU regression is the problem of recovering   $\bx\in \mathbb{R}^n$ with an $s$-sparse prior from $\by=\relu(\bA\bx)$, where $\relu(a)=\max\{a,0\}= \frac{a+|a|}{2}$ is the ReLU activation in neural networks. Note that the parameter estimation problem in ReLU regression was considered in \cite{soltanolkotabi2017learning,kalan2019fitting,mazumdar2019learning}.

\subsection{Algorithm}
Our algorithm is a slight variant of the projected gradient descent in \cite{soltanolkotabi2017learning}. 
\paragraph{\ref{ihtalg} with the ReLU-based $\ell_2$ loss.}  We adopt the ReLU-based $\ell_2$ loss
\[\calL_{\bx}(\bu) = \frac{1}{m}\sum_{i=1}^m [\relu(\ba_i^\top\bu)-y_i]^2\]
to fit ReLUs. Note that one can choose a ``generalized gradient'' of $\calL_{\bx}(\bu)$ as 
\[\partial\calL_{\bx}(\bu)=\frac{1}{m}\sum_{i=1}^m\big[\relu(\ba_i^\top\bu)-y_i\big]\big[1+\sign(\ba_i^\top\bu)\big]\ba_i,\]
and therefore we substitute $y_i=\relu(\ba_i^\top\bx)$ to introduce
\begin{align}
    \label{hxrelu}\bh_{\bx}(\bu) = \frac{1}{m}\sum_{i=1}^m\big[\relu(\ba_i^\top\bu)-\relu(\ba_i^\top\bx)\big]\big[1+\sign(\ba_i^\top\bu)\big]\ba_i. 
\end{align}
Using a fixed step size of $\eta =1$, we shall consider the \ref{ihtalg}
\begin{align}\label{oriihtrelu}
    \bx_{t+1} = H_s(\bx_t - \bh_{\bx}(\bx_t)),\quad t\ge 0. 
\end{align}

\paragraph{Initialization.} We consider the projection-based estimator of \cite{plan2017high,xu2020quantized}. By the rotational invariance of $\ba_i$, 
\[\mathbb{E}(y_i\ba_i) = \mathbb{E}[\relu(\ba_i^\top\bx)\ba_i]=\mathbb{E}\bigg[\relu(\ba_i^\top\bx)\ba_i^\top\frac{\bx}{\|\bx\|_2}\bigg]\frac{\bx}{\|\bx\|_2}=\frac{\bx}{2},\]
the idea of constructing $\bx_0$ is to project the average $\frac{1}{m}\sum_{i=1}^m2y_i\ba_i$ onto $\Sigma^n_s$: 
\[\bx_0 = H_s\bigg(\frac{1}{m}\sum_{i=1}^m 2y_i\ba_i\bigg).\]
In fact, \cite{soltanolkotabi2017learning} leveraged the same initialization. 
\paragraph{An additional projection.} In the same spirit as \texttt{BTAF} (Algorithm \ref{alg:btaf}) for SPR, we modify \cite{soltanolkotabi2017learning} by incorporating an additional projection into (\ref{oriihtrelu}) to ensure the boundedness of the iterates. We shall use $2\|\bx_0\|_2$ as a crude upper bound on $\|\bx\|_2$ and then utilize the projection onto $\mathbb{B}_2^n(2\|\bx\|_2)$, changing (\ref{oriihtrelu}) to \[\bx_{t+1}=P_{\mathbb{B}_2^n(2\|\bx_0\|_2)}\big(H_s(\bx_t - \bh_{\bx}(\bx_t))\big).\]

Putting the pieces together, we formally introduce the full procedure, referred to as bounded ReLU iterative hard thresholding (\texttt{BReLUIHT}), as Algorithm \ref{alg:reluiht}.  

\begin{algorithm}[ht!]   
\caption{Bounded ReLU Iterative Hard Thresholding (\texttt{BReLUIHT})}\label{alg:reluiht}
{\setstretch{1.15}        
\begin{algorithmic}[1]   
  \Require Matrix $\bA$, measurements $\by=\relu(\bA\bx)$, sparsity level $s$

  \State \textbf{Step 1: Initialization}
  \State Compute $\bx_0 = H_s\big(\frac{1}{m}\sum_{i=1}^m 2y_i\ba_i\big)$  

  \State \textbf{Step 2: Refinement}
  \For{$t = 0,1,2,\dots$}
    \State $\displaystyle \bx_{t+1} = P_{\mathbb{B}_2^n(2\|\bx_0\|_2)}\bigg(H_s\bigg(\bx_t-\frac{1}{m}\sum_{i=1}^m\big[\relu(\ba_i^\top\bx_t)-y_i\big]\big[1+\sign(\ba_i^\top\bx_t)\big]\ba_i \bigg)\bigg),
         $
  \EndFor

  \Ensure Sequence $\{\bx_t\}_{t\ge 0}$
\end{algorithmic}
}
\end{algorithm}

For $\bx=0$, \texttt{BReLUIHT} returns $\bx_t=0$ and   achieves exact recovery. As such, we restrict our attention to nonzero $\bx$ in the subsequent analysis. We will not discuss the technical details in this section as they are largely analogous to the SPR analysis.

\subsection{Instance Optimality}\label{sec:ioprelu}
Our main result establishes the instance optimality of \texttt{BReLUIHT}, providing the analogue of (\ref{csl2l111}) for sparse ReLU regression.
\begin{theorem}
    \label{thm:reluuniform}
    Under Gaussian matrix $\bA$, if $m\gtrsim s\log\frac{en}{s}$, then with probability at least $1-C\exp(-cs\log\frac{en}{s})$, $\{\bx_t\}_{t\ge0}$ generated by Algorithm \ref{alg:reluiht} satisfies
    \[\|\bx_t-\bx\|_2 \le \frac{c_1\|\bx\|_2}{2^t}+C_2\tau_s(\bx),\quad \forall t\ge 0,~\forall\bx\in\mathbb{R}^n.\]
    In particular, if $s$ is even and $s=2s_0$, then 
    \[\|\bx_t-\bx\|_2\le \frac{c_3\|\bx\|_2}{2^t}+\frac{C_4\|\bx-\bx_{[s_0]}\|_1}{\sqrt{s_0}},\quad\forall t\ge 0,~\forall\bx\in\mathbb{R}^n.\]
\end{theorem}
\begin{proof}
    The proof is provided in Appendix \ref{app:uniioprelu}.  
\end{proof}



\subsection{Non-Uniform Instance Optimality}\label{sec:nonurelu}
We have the following result for the non-uniform recovery of a fixed $\bx$, which finds the analogue of (\ref{nonulinear}) in sparse ReLU regression. 
\begin{theorem}
    \label{thm:relunon-uniform}
    Under Gaussian matrix $\bA$, for any fixed $\bx\in\mathbb{R}^n $, if $m\gtrsim s\log\frac{en}{s}$, then with probability at least $1-C\exp(-cs\log\frac{en}{s})$, $\{\bx_t\}_{t\ge 0}$ generated by Algorithm \ref{alg:reluiht} satisfies 
    \[\|\bx_t-\bx\|_2 \le \frac{c_1\|\bx\|_2}{2^t} + C_2\delta_{s}(\bx),\quad\forall t\ge 0.\]
\end{theorem} 
\begin{proof}
    The proof appears in Appendix \ref{app:nonuioprelu}.
\end{proof}

 
 
 
\section{Conclusion}\label{sec:conclu}
Instance optimality is a standard notion in compressed sensing that characterizes the robustness of the algorithms to model error. Its non-uniform counterpart also connects naturally to oracle inequalities in statistics. In this work, we developed a general framework to establish instance optimality for sparse recovery problems with nonlinear measurements. The main ingredient is the signal-dependent RAIC for a suitable gradient map, which leads to instance optimality of IHT. We established such RAICs  for sparse phase retrieval, one-bit compressed sensing and sparse ReLU regression  with Gaussian designs and obtained instance optimal guarantees. In fact, our framework also applies to phase-only compressed sensing and establishes instance optimal guarantee comparable to \cite{chen2024robust}; see \cite[Section 2]{chen2026supp}. There remain a number of interesting questions for future research, such as extension to non-Gaussian designs, sharp oracle inequalities under nonlinear observations, the tuning of the sparsity level, among many others.

\bibliography{libr} 
\appendix


\section{Proofs for Convergence of \ref{ihtalg} and \ref{nihtalg}}\label{sec:proofthm12}
We start with a useful lemma which controls the ``post-hard-thresholding error'' by top-$2s$ $\ell_2$ norm. 
\begin{lem}
    \label{dualbound} 
    For any $\bu\in\mathbb{R}^n$ and any $\bv\in \Sigma^n_s$, we have
    $\|H_{s}(\bu)-\bv\|_2\le 2\|\bu-\bv\|_{2,2s}\,.$
\end{lem}
 \begin{proof} 
 For $\calT\subset [n]$, recall that $\bu_{\calT}$ is obtained from $\bu$ by setting the entries not in $\calT$ as $0$.
     Let $\calT_1$, $\calT_2$ be the support sets of $H_s(\bu)$ and $\bv$, respectively, then \[\|H_s(\bu)-\bv\|_2\le \|H_s(\bu)-\bu_{\calT_1\cup \calT_2}\|_2+\|\bv-\bu_{\calT_1\cup \calT_2}\|_2\le 2 \|\bv-\bu_{\calT_1\cup \calT_2}\|_2,\] 
     where the second inequality holds because $H_s(\bu)$ is the best $s$-sparse approximation to $\bu_{\calT_1\cup \calT_2}$. The result then follows from  $\|\bv-\bu_{\calT_1\cup \calT_2}\|_2=\|(\bu-\bv)_{\calT_1\cup \calT_2}\|_2\le \|\bu-\bv\|_{2,2s}.$  
 \end{proof}
\subsection{Proofs of Theorems \ref{raiciht}--\ref{raicniht} (Convergence of \ref{ihtalg} and \ref{nihtalg})}
\begin{proof}
    We prove the two results simultaneously. Let $q=1$ for \ref{ihtalg} and $q=2$ for \ref{nihtalg}, and write
    \(
        \rho:=2q\mu_1,~
        b:=2q\mu_2+3q\|\bx-\bx_{[s]}\|_2.\)
    Under the assumptions of the corresponding theorem, $\rho<1$ and $R_{\bx}^{\rm loc}>b/(1-\rho)$. We define a sequence $\{f_t\}_{t\ge 0}$ by    
    \begin{align*}f_0=\|\bx_0-\bx\|_2,\qquad 
        f_{t+1}=\rho f_t+b,\qquad t\ge 0.
    \end{align*}
    It is not hard to find that, for any $t\ge 0$,
    \begin{align}\label{unifiedft}
        f_t
        &=\rho^t\|\bx_0-\bx\|_2+\big(1-\rho^t\big)\frac{b}{1-\rho}\\
        &\le \rho^t\|\bx_0-\bx\|_2+\frac{b}{1-\rho}.\nn
    \end{align}
    In light of $\|\bx_0-\bx\|_2\le R_{\bx}^{\rm loc}$ and $b/(1-\rho)<R_{\bx}^{\rm loc}$, we further obtain
    \[
        f_t\le \max\left\{\|\bx_0-\bx\|_2,\frac{b}{1-\rho}\right\}\le R_{\bx}^{\rm loc},\quad \forall t\ge 0.\]

    Therefore, all that remains is to prove $\|\bx_t-\bx\|_2\le f_t$ for any $t\ge 0$, and we use induction to achieve this. We first derive a common one-step estimate.  
    For \ref{ihtalg}, triangle inequality and Lemma \ref{dualbound} give
    \begin{align*}
        \|\bx_{t+1}-\bx\|_2
        &=\|H_s(\bx_t-\eta\cdot\bh_{\bx}(\bx_t))-\bx\|_2\\
        &\stackrel{(a)}{\le}\|H_s(\bx_t-\eta\cdot\bh_{\bx}(\bx_t))-\bx_{[s]}\|_2+\|\bx-\bx_{[s]}\|_2\\
        &\stackrel{(b)}{\le}2\|\bx_t-\eta\cdot\bh_{\bx}(\bx_t)-\bx_{[s]}\|_{2s,2}+\|\bx-\bx_{[s]}\|_2\\
        &\stackrel{(c)}{\le}2\|\bx_t-\eta\cdot\bh_{\bx}(\bx_t)-\bx\|_{2s,2}+3\|\bx-\bx_{[s]}\|_2.
    \end{align*}
    For \ref{nihtalg}, using additionally $\bx\in\mathbb{S}^{n-1}$, we have
    \begin{align*}
        \|\bx_{t+1}-\bx\|_2
        &=\left\|\frac{H_s(\bx_t-\eta\cdot\bh_{\bx}(\bx_t))}{\|H_s(\bx_t-\eta\cdot\bh_{\bx}(\bx_t))\|_2}-\bx\right\|_2\\
        &\stackrel{(a)}{\le}2\|H_s(\bx_t-\eta\cdot\bh_{\bx}(\bx_t))-\bx\|_2\\
        &\stackrel{(b)}{\le}2\|H_s(\bx_t-\eta\cdot\bh_{\bx}(\bx_t))-\bx_{[s]}\|_2+2\|\bx-\bx_{[s]}\|_2\\
        &\stackrel{(c)}{\le}4\|\bx_t-\eta\cdot\bh_{\bx}(\bx_t)-\bx_{[s]}\|_{2s,2}+2\|\bx-\bx_{[s]}\|_2\\
        &\stackrel{(d)}{\le}4\|\bx_t-\eta\cdot\bh_{\bx}(\bx_t)-\bx\|_{2s,2}+6\|\bx-\bx_{[s]}\|_2.
    \end{align*}
    In the first display, $(a)$ and $(c)$ are due to triangle inequality, and $(b)$ follows from Lemma \ref{dualbound}. In the second display, $(a)$ is a standard bound that can be found in \cite[Equation (2.12)]{chen2024robust} for instance, $(b)$ and $(d)$ are due to triangle inequality, and $(c)$ follows from Lemma \ref{dualbound}. Consequently, for either algorithm,
    \begin{align}\label{unifiedstep}
        \|\bx_{t+1}-\bx\|_2
        \le 2q\|\bx_t-\bx-\eta\cdot\bh_{\bx}(\bx_t)\|_{2s,2}
        +3q\|\bx-\bx_{[s]}\|_2.
    \end{align}

    We now prove $\|\bx_t-\bx\|_2\le f_t$ by induction. The base case holds trivially. Suppose that $\|\bx_t-\bx\|_2\le f_t$. Since $f_t\le R_{\bx}^{\rm loc}$, we have $\|\bx_t-\bx\|_2\le R_{\bx}^{\rm loc}$. Moreover, $\bx_t\in\Sigma_s^n$ for \ref{ihtalg}, while $\bx_t\in\Sigma_s^{n,*}$ for \ref{nihtalg}. Thus, in light of (\ref{dlower11}) and (\ref{dlower22}), we have $\bx_t\in\calU$, and the assumed RAIC yields
    \begin{align*}
        \|\bx_t-\bx-\eta\cdot\bh_{\bx}(\bx_t)\|_{2s,2}
        \le \mu_1\|\bx_t-\bx\|_2+\mu_2.
    \end{align*}
    Substituting this into (\ref{unifiedstep}) and using $\|\bx_t-\bx\|_2\le f_t$, we obtain
    \[
        \|\bx_{t+1}-\bx\|_2\le 2q\mu_1\|\bx_t-\bx\|_2+2q\mu_2+3q\|\bx-\bx_{[s]}\|_2 \le \rho f_t+b=f_{t+1}.
    \] 
    This completes the induction. Finally, substituting $q=1$ and $q=2$ into (\ref{unifiedft}) yields (\ref{converge11}) and (\ref{converge22}), respectively, and completes the proof.
\end{proof}

\section{Gaussian Width and Covering Number}  
Our subsequent proofs build on some concentration lemmas and existing results   involving
Gaussian width and covering number of $\calU\subset \mathbb{R}^n$. Here, we shall pause to provide a brief introduction. 
With $\bg\sim N(0,\bI_n)$, the Gaussian width of $\calU\subset \mathbb{R}^n$ is defined as 
\[\omega(\calU) = \mathbb{E}\sup_{\bu\in\calU} \bg^\top \bu.\]
Also, the covering number of $\calU$ at radius $r$, denoted by $\calN(\calU,r)$, is the minimal number of radius-$r$ $\ell_2$ balls needed to cover $\calU$. As an equivalent notion, $\log\calN(\calU,r)$ is referred to as the metric entropy of $\calU$ at radius $r$. In the following, we collect the Gaussian width and metric entropy estimates used in this paper. 
\begin{lem} 
Let $\{\be_j\}_{j=1}^n$ be the canonical basis in $\mathbb{R}^n$, $s\in[n]$, $r\in (0,2)$,  and $\calX_1:=\{\bu\in \mathbb{S}^{n-1}:\tau_s(\bu)\le c_*\|\bu\|_2\}$ for some $c_*\in(0,1)$. Then for some universal constant $C,$
    \begin{gather}
    \label{gwej}\omega\big(\{\be_j\}_{j=1}^n\big)\le C\sqrt{\log n};
    \\
        \label{gwsparse}\omega(\Sigma^{n,*}_s)\le \omega(\sqrt{s}\mathbb{B}_1^n\cap \mathbb{S}^{n-1}) \le C \sqrt{s\log\frac{en}{s}};\\
    \label{sparsecovering}
    \log \calN(\Sigma^{n,*}_s,r) \le s\log\Big(\frac{12n}{sr}\Big);\\
    \omega(\calX_1)\le C\sqrt{s\log\frac{en}{s}}. \label{Xstargwbound}
    \end{gather}
\end{lem} 
\begin{proof}
    For the first three inequalities, see, e.g., \cite{vershynin2018high,plan2012robust,plan2013one}. We now prove (\ref{Xstargwbound}).  For $\bw\in \calX_1$, 
     \begin{align*}
         &\|\bw\|_1 = \|\bw_{[s]}\|_1 + \sum_{i\ge 1}\|\bw_{[(i+1)s]}-\bw_{[is]}\|_1 \le \sqrt{s}\|\bw_{[s]}\|_2+\sum_{i\ge 1}\sqrt{s}\|\bw_{[(i+1)s]}-\bw_{[is]}\|_2 \le (1+c_*)\sqrt{s}, 
     \end{align*}
     and hence 
     \(
         \calX_1\subset \mathbb{B}_1^n((1+c_*)\sqrt{s}) \cap \mathbb{S}^{n-1}.\)
     Thus, (\ref{gwsparse}) gives \(
         \omega(\calX_1)\le \omega( \mathbb{B}_1^n\big((1+c_*)\sqrt{s}\big)\cap\mathbb{S}^{n-1})\lesssim \sqrt{s\log (en/s)}. 
     \)
\end{proof}

\section{Proofs for Sparse Phase Retrieval}\label{sec:proof}
Throughout this appendix, we work with $\bh_{\bx}(\bu) = \frac{1}{m}\sum_{i=1}^m(|\ba_i^\top\bu|-|\ba_i^\top\bx|)\sign(\ba_i^\top\bu)\ba_i$ and   consider   $\bx\ne 0$.

 \subsection{Proof of Theorem \ref{thm:uni} (Instance Optimal SPR)}\label{thm1}
We   let
\(
    \calX:=\{\bu \in \mathbb{R}^n: \tau_s(\bu)\le c_*\|\bu\|_2\}  
\)
for some sufficiently small $c_*>0$. 
For  convenience, we shorten the sparse restricted eigenvalue of  a matrix $\bM\in \mathbb{R}^{m\times n}$ as \begin{align}
    \Gamma_s(\bM):=\sup_{\bu\in \Sigma^{n,*}_s}\|\bM\bu\|_2.\label{defineGamma}
\end{align}  
 We note the following bound for a Gaussian matrix. 
\begin{lem}[See, e.g., {\cite[Section 9]{vershynin2018high}}] \label{ripupper}
     $\bA\sim N^{m\times n}(0,1)$ satisfies   
\begin{align}\label{upperrip}
    \mathbb{P}\bigg(\Gamma_s \Big(\frac{\bA}{\sqrt{m}}\Big) \le 1 + \sqrt{\frac{Cs\log(en/s)}{m}}\bigg) \ge 1-2\exp(-s\log(en/s))\,.
\end{align}
\end{lem}

\subsubsection{Proof of the RAIC in (\ref{uraicspr})}
 
For some $c>0$, we seek to establish (\ref{uraicspr}). By definition, this is equivalent to 
\begin{align}
    \big\|\bu-\bx-\bh_{\bx}(\bu)\big\|_{2s,2}\le  \frac{\|\bu-\bx\|_2}{4}+4\tau_s(\bx),~~
   \forall\bu \in  \Sigma^{n}_s\cap \mathbb{B}_2^n(\bx;c\|\bx\|_2),~\forall \bx\in\calX\setminus\{0\}.&\label{desiredraicspr}
\end{align}
We start with the decomposition 
\begin{align} 
    \|\bu-\bx-\bh_{\bx}(\bu)\|_{2s,2} &\le \|\bu - \bx_{[s]} - \bh_{\bx_{[s]}}(\bu)\|_{2s,2}  \label{sprdecompose}+ \|\bx-\bx_{[s]}\|_2 + \big\|\bh_{\bx}(\bu)-\bh_{\bx_{[s]}}(\bu)\big\|_{2s,2}
\end{align}

\subsubsection*{(i) Bounding $ \|\bu - \bx_{[s]} - \bh_{\bx_{[s]}}(\bu)\|_{2s,2}$}

Noticing that $\bu$ and $\bx_{[s]}$ are both $s$-sparse, we control the first term $\|\bu-\bx_{[s]}-\bh_{\bx_{[s]}}(\bu)\|_{2s,2}$ by using the RAIC established in \cite{chen2025unified}. 
 
\begin{lem} \label{lem:raic} If $m\gtrsim s\log\frac{en}{s}$, then with probability at least $1-c_1\exp(-c_2m)$,   
\begin{align*}
    \|\bu-\bx-\bh_{\bx}(\bu)\| _{2,2s}\le \frac{1}{4}\|\bu-\bx\|_2,\quad \forall \textrm{$\bu,\bx\in \Sigma^n_s\setminus\{0\}$ obeying $\|\bu - \bx\|_2\le \tilde{c}\|\bx\|_2$}
\end{align*}
holds for some small enough universal constant $\tilde{c}>0$.  
\end{lem}
\begin{proof} 
   By Definition \ref{def:sparseraic}, the claim is equivalent to 
   \[\bh_{\bx}(\bu)\sim {\rm RAIC}\Big(\Sigma^{n,*}_{2s};\Sigma^n_s\cap \mathbb{B}_2^n(\bx;\tilde{c}\|\bx\|_2),\bx,\frac{1}{4}\|\bu-\bx\|_2\Big),\qquad \forall \bx\in\Sigma^n_s\setminus\{0\}.\]
   The result then directly follows from \cite[Theorem B.9]{chen2025unified} with $\calC=\Sigma^n_s$, $\calK=\Sigma^{n,*}_{2s}$, together with the estimates (\ref{gwsparse})--(\ref{sparsecovering}).  
\end{proof}

 For $\bx\in \calX\setminus\{0\}$ and $\bu\in \Sigma^n_s\cap \mathbb{B}_2^n(\bx;c\|\bx\|_2)$, note that 
\begin{align*}
    \|\bu-\bx_{[s]}\|_2 \le \|\bu-\bx\|_2 +\|\bx-\bx_{[s]}\|_2\le c\|\bx\|_2 + \tau_s(\bx) \le (c+c_*)\|\bx\|_2. 
\end{align*}
Therefore, by setting $c\le \frac{\tilde{c}}{4}$ and $c_*\le \frac{\tilde{c}}{4}$ for the $\tilde{c}$ in Lemma \ref{lem:raic}, we obtain the uniform bound  
\begin{align}\label{prraic_sparse}
    \big\|\bu-\bx_{[s]}-\bh_{\bx_{[s]}}(\bu)\big\|_{2s,2} \le\frac{\|\bu-\bx_{[s]}\|_2 }{4} \le \frac{\|\bu-\bx\|_2}{4}+ \frac{\|\bx-\bx_{[s]}\|_2}{4}. 
\end{align}

\subsubsection*{(ii) Bounding $\|\bh_{\bx}(\bu)-\bh_{\bx_{[s]}}(\bu)\|_{2s,2}$}
 For all $\bx\in\calX\setminus\{0\}$ and $\bu\in \Sigma^n_s \cap \mathbb{B}_2^n(\bx;c\|\bx\|_2)$,
\begin{align} \nn
    &\big\|\bh_{\bx}(\bu)-\bh_{\bx_{[s]}}(\bu)\big\|_{2s,2} = \sup_{\bw\in \Sigma^{n,*}_{2s}}\big\langle \bw, \bh_{\bx}(\bu)- \bh_{\bx_{[s]}}(\bu)\big\rangle \\\nn
    &= \sup_{\bw\in \Sigma^{n,*}_{2s}} \frac{1}{m}\sum_{
    i=1
    }^m \big(|\ba_i^\top\bx|-|\ba_i^\top\bx_{[s]}|\big)\sign(\ba_i^\top\bu)\ba_i^\top\bw \\\nn
    &\stackrel{(a)}{\le} \sup_{\bw\in \Sigma^{n,*}_{2s}} \frac{1}{m}\sum_{i=1}^m \big|\ba_i^\top(\bx-\bx_{[s]})\big||\ba_i^\top\bw|
    \\\nn
    &\stackrel{(b)}{\le} \sup_{\bw\in \Sigma^{n,*}_{2s}} \frac{1}{m}\sum_{i=1}^m \sum_{j\ge 1}\big|\ba_i^\top(\bx_{[(j+1)s]}-\bx_{[js]})\big||\ba_i^\top\bw| \\\nn
    &\stackrel{(c)}{\le} \sum_{j\ge 1}\|\bx_{[(j+1)s]}-\bx_{[js]}\|_2 \sup_{\bw\in \Sigma^{n,*}_{2s}}\sup_{\bv\in \Sigma^{n,*}_s}  \frac{1}{m}\sum_{i=1}^m |\bv^\top\ba_i\ba_i^\top\bw| \\\nn
    &\stackrel{(d)}{\le} \sum_{j\ge 1}\|\bx_{[(j+1)s]}-\bx_{[js]}\|_2 \sup_{\bw\in \Sigma^{n,*}_{2s}}\frac{\|\bA\bv\|_2}{\sqrt{m}}\sup_{\bv\in \Sigma^{n,*}_s} \frac{\|\bA\bw\|_2}{\sqrt{m}}  \\\nn
    &\stackrel{(e)}{\le} \sum_{j\ge 1}\|\bx_{[(j+1)s]}-\bx_{[js]}\|_2 \cdot\Gamma_{2s}\Big(\frac{\bA}{\sqrt{m}}\Big)\cdot\Gamma_s\Big(\frac{\bA}{\sqrt{m}}\Big)
    \\ \label{gradientmismatch}
    &\stackrel{(f)}{\le} 2\sum_{j\ge 1}\|\bx_{[(j+1)s]}-\bx_{[js]}\|_2 = 2\tau_s(\bx), 
\end{align} 
where $(a)$ is due to triangle inequality, in $(b)$ we use $\bx-\bx_{[s]} = \sum_{j\ge 1}(\bx_{[(j+1)s]}-\bx_{[js]})$ and triangle inequality, in $(c)$ we suppose $\bx_{[(j+1)s]}-\bx_{[js]}\ne 0$ and utilize $\frac{\bx_{[(j+1)s]}-\bx_{[js]}}{\|\bx_{[(j+1)s]}-\bx_{[js]}\|_2}\in \Sigma^{n,*}_s$, $(d)$ is due to Cauchy--Schwarz inequality, $(e)$ follows from the definition of $\Gamma_s(\cdot)$ in (\ref{defineGamma}), and $(f)$ holds with the promised probability in light of Lemma \ref{ripupper}.

Substituting (\ref{prraic_sparse}) and (\ref{gradientmismatch}) into  
  (\ref{sprdecompose}) and using $\|\bx-\bx_{[s]}\|_2\le \tau_s(\bx)$, we arrive at (\ref{desiredraicspr}).  

\subsubsection{Proof of the Initialization Guarantee (\ref{uniinispr})} 
Since $\bx_0\in \Sigma^n_s$ by construction, to ensure (\ref{iniiht}), it remains to show 
\begin{align}
    \dist(\bx_0,\bx)\le c\|\bx\|_2,\quad \forall\bx\in\calX\label{sprinix0}
\end{align} 
for the constant $c$ in (\ref{desiredraicspr}). 
For convenience we let $\bar{\bx}:=\bx/\|\bx\|_2$ and define 
\begin{align}
    \tilde{\bS}_{\bx}:=\frac{1}{m}\sum_{i=1}^m \sfT_{2\|\bx\|_2}(y_i)\ba_i\ba_i^\top\label{idealdm}
\end{align}
as a surrogate of the data matrix $\hat{\bS}_{\bx}=\frac{1}{m}\sum_{i=1}^m T_{2\lambda_{\bx}}(y_i)\ba_i\ba_i^\top.$ 
The following lemma computes $\mathbb{E}[\tilde{\bS}_{\bx}].$

\begin{lem}\label{lem:popu}
    For any $\bx\in \mathbb{R}^n$, we have $\mathbb{E}[\tilde{\bS}_{\bx}] = \|\bx\|_2\big(c_\flat \bI_n+ c_{\diamond}\bar{\bx}\bar{\bx}^\top\big)$ where $c_{\flat}\approx 0.7809,~c_{\diamond}\approx 0.6899$ are universal constants.  
\end{lem}
\begin{proof}
     The proof is a straightforward calculation based on the rotational invariance of $\ba_i$. We omit the details.  
\end{proof}

\subsubsection*{(iii) A Deterministic Initialization Error Bound}
  We  then provide a deterministic bound on the initialization error for $\bx\in\calX\setminus\{0\}$, i.e., the nonzero signals with small enough $
  \frac{\tau_s(\bx)}{\|\bx\|_2}$.  In fact, the bound is valid for a larger set of signals with small $\frac{\|\bx-\bx_{[s]}\|_2}{\|\bx\|_2}$; see (\ref{smalll2}) below.  
 The proof is analogous to existing analysis   (e.g., \cite{jagatap2019sample}) but needs additional technicalities to handle model error. 
 
\begin{lem}\label{lem:inideter}
    Let $\bx$ be nonzero. There exist universal constants $C>c>0$ such that the following hold. For any $\bar{c}\in(0,c)$, if  
    \begin{gather} \label{smalll2}
    \frac{\|\bx-\bx_{[s]}\|_2}{\|\bx\|_2}\le \bar{c}\,,
    \\\frac{|\lambda_{\bx}-\|\bx\|_2|}{\|\bx\|_2}\le \bar{c}\,, \label{normaccc}
    \\
        \label{operaconcen}
        \Phi(\bx):=\frac{1}{\|\bx\|_2}\sup_{\bu\in\Sigma^{n,*}_{s}} \Big|\bu^\top \big(\hat{\bS}_{\bx}-\mathbb{E}[\tilde{\bS}_{\bx}]\big)\bu\Big| \le \bar{c}\,, \\
        \label{maxconcen}
        \Psi(\bx):= \frac{1}{\|\bx\|_2}\sup_{\bu\in\{\be_j\}_{j=1}^n} \Big|\bu^\top \big(\hat{\bS}_{\bx}-\mathbb{E}[\tilde{\bS}_{\bx}]\big)\bu\Big| \le \frac{\bar{c}}{s}\,,
    \end{gather} 
    then  
    \(
    \dist(\bx_0,\bx) \le C\sqrt{\bar{c}}\|\bx\|_2. 
    \)
\end{lem}
\begin{proof}
    We start with 
\begin{align}\label{decmmpose111}
\dist(\bx_0,\bx) \le \dist(\bx_0,\bx_{[s]}) + \|\bx-\bx_{[s]}\|_2\stackrel{{\rm(\ref{smalll2})}}{\le} \dist(\bx_0,\bx_{[s]}) + c_0\|\bx\|_2, 
\end{align}
and notice that 
      \begin{align}\nn
          &\dist(\bx_0,\bx_{[s]}) = \dist(\lambda_{\bx}\bv_{\bx},\bx_{[s]})\\\nn
          & \le \dist\big(\lambda_{\bx}\bv_{\bx} ,\|\bx\|_2\bv_{\bx}\big)+\dist\bigg(\|\bx\|_2\bv_{\bx}, \frac{\|\bx\|_2\bx_{\calI_{\bx}}}{\|\bx_{\calI_{\bx}}\|_2}\bigg) + \dist\bigg(\frac{\|\bx\|_2\bx_{\calI_{\bx}}}{\|\bx_{\calI_{\bx}}\|_2},\bx_{[s]}\bigg)\\\label{decomposeini}
          & \le \big|\lambda_{\bx}-\|\bx\|_2\big| + \|\bx\|_2 \dist\bigg(\bv_{\bx},\frac{\bx_{\calI_{\bx}}}{\|\bx_{\calI_{\bx}}\|_2}\bigg)+ \|\bx\|_2\bigg\| \frac{\bx_{\calI_{\bx}}}{\|\bx_{\calI_{\bx}}\|_2} - \frac{\bx_{[s]}}{\|\bx\|_2}\bigg\|_2,
      \end{align}
      where in the last inequality we use triangle inequality and $\dist(\bu,\bv)\le\|\bu-\bv\|_2$.
We are left to bound the second and third terms in (\ref{decomposeini}).

\paragraph{Bounding the third term in (\ref{decomposeini}).} By triangle inequality,  
\begin{align}\label{supperr}
    \bigg\| \frac{\bx_{\calI_{\bx}}}{\|\bx_{\calI_{\bx}}\|_2} - \frac{\bx_{[s]}}{\|\bx\|_2}\bigg\|_2 \le \bigg\| \frac{\bx_{\calI_{\bx}}}{\|\bx_{\calI_{\bx}}\|_2} - \frac{\bx_{[s]}}{\|\bx_{[s]}\|_2}\bigg\|_2 + \bigg\| \frac{\bx_{[s]}}{\|\bx_{[s]}\|_2} - \frac{\bx_{[s]}}{\|\bx\|_2}\bigg\|_2\,,
\end{align}
and moreover, 
\begin{align}\label{bound2ndt}
    \bigg\| \frac{\bx_{[s]}}{\|\bx_{[s]}\|_2} - \frac{\bx_{[s]}}{\|\bx\|_2}\bigg\|_2 = \bigg|1-\frac{\|\bx_{[s]}\|_2}{\|\bx\|_2}\bigg|\le \frac{\|\bx-\bx_{[s]}\|_2}{\|\bx\|_2} \stackrel{{\rm(\ref{smalll2})}}{\le} \bar{c}\,.
\end{align}
For the first term on the right-hand side of (\ref{supperr}),  by triangle inequality, 
\begin{align} \nn 
    &\bigg\|\frac{\bx_{\calI_{\bx}}}{\|\bx_{\calI_{\bx}}\|_2}-\frac{\bx_{[s]}}{\|\bx_{[s]}\|_2}\bigg\|_2 \le \bigg\|\frac{\bx_{\calI_{\bx}}}{\|\bx_{\calI_{\bx}}\|_2}-\frac{\bx_{\calI_{\bx}}}{\|\bx_{[s]}\|_2}\bigg\|_2 + \bigg\|\frac{\bx_{\calI_{\bx}}-\bx_{[s]}}{\|\bx_{[s]}\|_2}\bigg\|_2\\
    & = \frac{|\|\bx_{[s]}\|_2-\|\bx_{\calI_{\bx}}\|_2|+\|\bx_{\calI_{\bx}}-\bx_{[s]}\|_2}{\|\bx_{[s]}\|_2}\le \frac{2\|\bx_{\calI_{\bx}}-\bx_{[s]}\|_2}{\|\bx_{[s]}\|_2}\,.\label{norminusnor}
\end{align}
We let $I_0\subset [n]$ be a size-$s$ index set such that 
  $\supp(\bx_{[s]})\subset I_0$ and notice   
\begin{align}\label{decompose1}
    \|\bx_{\calI_{\bx}}-\bx_{[s]}\|_2 \le \|\bx_{\calI_{\bx} \setminus I_0}\|_2+\|\bx_{I_0\setminus\calI_{\bx}}\|_2,
\end{align}
where $\|\bx_{\calI_{\bx} \setminus I_0}\|_2$ can be bounded by
\begin{align}
    \|\bx_{\calI_{\bx}\setminus I_0}\|_2\le\|\bx-\bx_{[s]}\|_2 \stackrel{{\rm(\ref{smalll2})}}{\le} \bar{c}\|\bx\|_2.\label{bound1term}
\end{align} 
To bound $\|\bx_{ I_0 \setminus\calI_{\bx}}\|_2$, with no loss of generality we suppose $I_0\setminus \calI_{\bx} \ne \varnothing$, which together with $|I_0|=|\calI_{\bx}|=s$ leads to 
\begin{align*}
    |\calI_{x}\setminus I_0| = |I_0\setminus \calI_x| \ge 1.
\end{align*}
In turn, for any $j\in I_0\setminus \calI_x$, there exists $k_i\in \calI_x\setminus I_0$ such that $k_i\ne k_{i'}$ if $i\ne i'$, and   
 by the construction of $\calI_{\bx}$, 
\begin{align}\label{defineIx}
    [\hat{\bS}_{\bx}]_{j,j} \le  [\hat{\bS}_{\bx}]_{k_j,k_j}\,,\quad \forall j \in I_0\setminus \calI_x.
\end{align}
Moreover, (\ref{maxconcen}) along with Lemma \ref{lem:popu} implies 
\begin{gather} 
       [\hat{\bS}_{\bx}]_{j,j} - \|\bx\|_2 \bigg(c_{\flat} + \frac{c_\diamond x_j^2}{\|\bx\|_2^2}\bigg) \ge  -\frac{\bar{c}\|\bx\|_2}{s}\,,\\
          [\hat{\bS}_{\bx}]_{k_j,k_j} - \|\bx\|_2 \bigg(c_{\flat} + \frac{c_\diamond x_{k_j}^2}{\|\bx\|_2^2}\bigg) \le  \frac{\bar{c}\|\bx\|_2}{s}\,.
 \label{maxnormuse}
\end{gather}
 Combining (\ref{defineIx})--(\ref{maxnormuse}) yields 
 \begin{align}
     x_j^2 \le x_{k_j}^2+\frac{2\bar{c}}{c_{\diamond}s}\|\bx\|_2^2\,,\quad \forall j \in I_0\setminus \calI_x,\label{boundjbykj}
 \end{align}  
and therefore 
     \begin{align}\nn 
         \|\bx_{I_0 \setminus\calI_{\bx}}\|_2 &= \bigg(\sum_{j\in I_0\setminus \calI_{\bx}}x_j^2\bigg)^{1/2} \\\nn &\stackrel{(a)}{\le} \bigg(\sum_{j\in I_0\setminus \calI_{\bx}}\Big[x_{k_j}^2 + \frac{2\bar{c}\|\bx\|_2^2}{c_\diamond s}\Big]\bigg)^{1/2}\\\nn 
         &\stackrel{(b)}{\le} \bigg(\|\bx-\bx_{[s]}\|_2^2+\frac{2\bar{c}\|\bx\|_2^2}{c_\diamond}\bigg)^{1/2} \\
         &\stackrel{(c)}{\le} \bigg(c_0^2 + \frac{2\bar{c}}{c_\diamond}\bigg)^{1/2}\|\bx\|_2, \label{2knotselect} 
     \end{align}
     where $(a)$ follows from Equation (\ref{boundjbykj}), $(b)$ holds because $k_i\ne k _{i'}$ for $i\ne i'$ and $|I_0\setminus\calI_x|\le s$, and $(c)$ is due to Equation (\ref{smalll2}).

     Substituting (\ref{bound1term}) and (\ref{2knotselect}) into (\ref{decompose1}) yields 
     \begin{align}\label{IxminusS}
         \|\bx_{\calI_x}-\bx_{[s]}\|_2 \lesssim\sqrt{\bar{c}}\|\bx\|_2.
     \end{align} 
     In light of (\ref{norminusnor}), along with $\|\bx_{[s]}\|_2\ge (1-\bar{c})\|\bx\|_2$ from (\ref{smalll2}) and taking sufficiently small $c$ (and hence $\bar{c}$), we obtain  
     \begin{align}
         \bigg\|\frac{\bx_{\calI_{\bx}}}{\|\bx_{\calI_{\bx}}\|_2}-\frac{\bx_{[s]}}{\|\bx_{[s]}\|_2}\bigg\|_2 \lesssim \frac{\sqrt{\bar{c}}}{1-\bar{c}}\lesssim \sqrt{\bar{c}}\,. 
     \end{align}
     Substituting this and (\ref{bound2ndt}) into (\ref{supperr}) leads to 
     \begin{align}\label{secondtermterm}
         \bigg\| \frac{\bx_{\calI_{\bx}}}{\|\bx_{\calI_{\bx}}\|_2} - \frac{\bx_{[s]}}{\|\bx\|_2}\bigg\|_2 \lesssim \sqrt{\bar{c}}\,.
     \end{align}

    \paragraph{Bounding the second term in (\ref{decomposeini}).} By Algorithm \ref{alg:btaf} and Lemma \ref{lem:popu}, $\bv_{\bx}$ and $\bx_{\calI_{\bx}}/\|\bx_{\calI_{\bx}}\|_2$  are the leading eigenvectors of $[\hat{\bS}_{\bx}]_{\calI_{\bx},\calI_{\bx}}$ and $[\mathbb{E}\tilde{\bS}_{\bx}]_{\calI_{\bx},\calI_{\bx}}$, respectively. Therefore,      
\[
    \|[\hat{\bS}_{\bx}]_{\calI_{\bx},\calI_{\bx}}-[\mathbb{E}\tilde{\bS}_{\bx}]_{\calI_{\bx},\calI_{\bx}}\|_{op}\stackrel{(a)}{\le} \sup_{\substack{\calI\subset [n]\\|\calI|\le s}} \big\|[\hat{\bS}_{\bx}]_{\calI,\calI}-[\mathbb{E}\tilde{\bS}_{\bx}]_{\calI,\calI}\big\|_{op} = \sup_{\bu\in\Sigma^{n,*}_s}   \Big|\bu^\top\Big(\hat{\bS}_{\bx}-\mathbb{E}[\tilde{\bS}_{\bx}]\Big)\bu\Big|\stackrel{(b)}{\le} \bar{c}\|\bx\|_2,
\] 
where $(a)$ holds because $|\calI_{\bx}|\le s$, $(b)$ is due to (\ref{operaconcen}).  

By Lemma \ref{lem:popu}, the first two eigenvalues of $$[\mathbb{E}\tilde{\bS}_{\bx}]_{\calI_{\bx},\calI_{\bx}} = \bigg[c_\flat \|\bx\|_2\bI_n +c_\diamond\frac{\bx\bx^\top}{\|\bx\|_2}\bigg]_{\calI_{\bx},\calI_{\bx}} =  c_\flat \|\bx\|_2[\bI_n]_{\calI_{\bx},\calI_{\bx}} + \frac{c_\diamond \bx_{\calI_{\bx}}\bx_{\calI_{\bx}}^\top}{\|\bx\|_2}  $$ are $c_\flat \|\bx\|_2 + \frac{c_\diamond \|\bx_{\calI_{\bx}}\|_2^2}{\|\bx\|_2}$ and $c_\flat  \|\bx\|_2$ (or $0$ if $s=1$), thus the  eigenvalue gap is at least $\frac{c_\diamond \|\bx_{\calI_{\bx}}\|_2^2}{\|\bx\|_2}$. Moreover, by (\ref{IxminusS}),
    $\|\bx_{[s]}\|_2\ge (1-\bar{c})\|\bx\|_2$,    and sufficiently small $\bar{c}$, 
$$\|\bx_{\calI_{\bx}}\|_2 \ge \|\bx_{[s]}\|_2 -\|\bx_{[s]}-\bx_{\calI_{\bx}}\|_2\ge \big(1-O(\sqrt{\bar{c}})\big)\|\bx\|_2\ge \frac{\|\bx\|_2}{2}\,.$$   Hence,
the eigengap between the first two eigenvalues of $[\mathbb{E}\tilde{\bS}_{\bx}]_{\calI_{\bx},\calI_{\bx}}$ is at least 
\(\frac{c_\diamond \|\bx_{\calI_{\bx}}\|_2^2}{\|\bx\|_2}\ge \frac{c_\diamond \|\bx\|_2}{4}.\) 
By sufficiently small $\bar{c}$ and \ref{operaconcen},    \[\|[\hat{\bS}_{\bx}]_{\calI_{\bx},\calI_{\bx}}-[\mathbb{E}\tilde{\bS}_{\bx}]_{\calI_{\bx},\calI_{\bx}}\|_{op} \le \frac{c_\diamond \|\bx\|_2}{8}.\]
Hence,  Davis-Kahan's theorem 
\citep{davis1970rotation}, along with 
the simple bound $\dist(\bu,\bv)\le \sqrt{2}\|\bu\bu^\top-\bv\bv^\top\|_{op}$ for $\bu,\bv\in\mathbb{S}^{n-1}$,   yields 
\begin{align}
    \dist\bigg(\bv_{\bx}, \frac{\bx_{\calI_{\bx}}}{\|\bx_{\calI_{\bx}}\|_2}\bigg)\le\sqrt{2}\bigg\|\bv_{\bx}\bv_{\bx}^\top - \frac{\bx_{\calI_{\bx}}\bx_{\calI_{\bx}}^\top}{\|\bx_{\calI_{\bx}}\|_2^2}\bigg\|_{op}\lesssim \bar{c}\,.\label{distbound}
\end{align}
By substituting $|\lambda_{\bx}-\|\bx\|_2|\le \bar{c}\|\bx\|_2$, (\ref{secondtermterm}), (\ref{distbound}), and (\ref{decomposeini}) into (\ref{decmmpose111}), we reach
\(
      \dist(\bx_0,\bx) \lesssim\sqrt{\bar{c}}\|\bx\|_2,\)
completing the proof.  
\end{proof}

Therefore, to achieve (\ref{sprinix0}), it is sufficient to establish (\ref{smalll2})--(\ref{maxconcen}) with small enough $\bar{c}$ for all $\bx\in\calX$. By the definition of $\calX$, (\ref{smalll2}) holds as long as $c_*$ is small enough. We proceed to establish (\ref{normaccc})--(\ref{maxconcen}) separately.

\subsubsection*{(iv) Establishing (\ref{normaccc}) for $\bx\in\calX\setminus\{0\}$} The following lemma, which follows from an $\ell_1$ embedding result due to \cite{plan2014dimension}, guarantees (\ref{normaccc}) under $m\gtrsim s\log\frac{en}{s}$. 

 \begin{lem}\label{lem:l1l2}
     If $m\gtrsim s\log\frac{en}{s}$, then with probability at least $1-2\exp(-cs\log \frac{en}{s})$,  $$\big|\lambda_{\bx}-\|\bx\|_2\big|\le C\sqrt{\frac{s\log(en/s)}{m}}\cdot \|\bx\|_2,\qquad\forall\,\bx\in\calX.$$
 \end{lem}
 \begin{proof}
    We only need to prove $$Z:=\sup_{\bx\in\calX}\left|\frac{\lambda_{\bx}}{\|\bx\|_2}-1\right|=O\bigg(\sqrt{\frac{s\log(en/s)}{m}}\bigg)\,.$$ By homogeneity, we let $\calX^*=\calX\cap \mathbb{S}^{n-1}$
    to obtain
     \begin{align*}
         Z = \sup_{\bx\in\calX}\bigg|\sqrt{\frac{\pi}{2}}\frac{1}{m}\sum_{i=1}^m \Big|\ba_i^\top\frac{\bx}{\|\bx\|_2}\Big|-1\bigg|= \sup_{\bx\in\calX^*}\bigg|\sqrt{\frac{\pi}{2}}\frac{1}{m}\sum_{i=1}^m|\ba_i^\top\bx|-1\bigg|.
     \end{align*}
     By \cite[Lemma 2.1]{plan2014dimension},  
     \begin{align}\label{l1l2rip}
         \mathbb{P}\left(Z> \frac{4\omega(\calX^*)}{\sqrt{m}}+t\right)\le 2\exp\Big(-\frac{mt^2}{2}\Big),\qquad\forall t>0.
     \end{align}
By (\ref{Xstargwbound}),
     we set $t=\sqrt{s\log(en/s)/m}$ in (\ref{l1l2rip}) to conclude the proof. 
 \end{proof}

 \subsubsection*{(v) Establishing (\ref{operaconcen})--(\ref{maxconcen}) for $\bx\in\calX\setminus\{0\}$} 
 To ensure (\ref{operaconcen})--(\ref{maxconcen}) for $\bx\in\calX\setminus\{0\}$, we shall establish upper bound on $
\sup_{\bx\in\calX\setminus\{0\}}\Phi(\bx)$ and $
\sup_{\bx\in\calX\setminus\{0\}}\Psi(\bx)$.
\begin{lem}\label{lem:Phixuniformbound}
If $m\gtrsim s\log\frac{en}{s}$, then
\begin{align}
    \sup_{\bx\in\calX\setminus\{0\}}\Psi(\bx) \le \sup_{\bx\in\calX\setminus\{0\}}\Phi(\bx)
    \lesssim
    \log\bigg(\frac{em}{s\log(en/s)}\bigg)
    \sqrt{\frac{s\log(en/s)}{m}}
    \label{Phixuniformbound}
\end{align}
holds with probability at least
$1-C\exp(-cs\log\frac{en}{s})$.
\end{lem}

 Before proving Lemma \ref{lem:Phixuniformbound}, we record a few
auxiliary results.  The first is an elementary contraction bound for
products.

\begin{lem}[Product contraction]\label{lem:productcontraction}
Let $\calP,\calQ\subset[0,1]^m$ be countable sets, and let
$\{\varepsilon_i\}_{i=1}^m$ be independent Rademacher random
variables. Then
\[
    \mathbb{E}_{\varepsilon}
    \sup_{p\in\calP,q\in\calQ}
    \bigg|\sum_{i=1}^m\varepsilon_i p_iq_i\bigg|
    \lesssim 
    \mathbb{E}_{\varepsilon}
    \sup_{p\in\calP}
    \bigg|\sum_{i=1}^m\varepsilon_i p_i\bigg| +
    \mathbb{E}_{\varepsilon}
    \sup_{q\in\calQ}
    \bigg|\sum_{i=1}^m\varepsilon_i q_i\bigg|.
\]
\end{lem}
\begin{proof}
For any $p,q\in\mathbb{R}$, we have
\(
    pq=\frac{1}{4}\big[(p+q)^2-(p-q)^2\big].\)
Therefore, by triangle inequality,
\begin{align*}
\mathbb{E}_{\varepsilon}
\sup_{p\in\calP,q\in\calQ}
\bigg|\sum_{i=1}^m\varepsilon_i p_iq_i\bigg| \le
\frac{1}{4}\mathbb{E}_{\varepsilon}
\sup_{p\in\calP,q\in\calQ}
\bigg|\sum_{i=1}^m\varepsilon_i(p_i+q_i)^2\bigg|
+
\frac{1}{4}\mathbb{E}_{\varepsilon}
\sup_{p\in\calP,q\in\calQ}
\bigg|\sum_{i=1}^m\varepsilon_i(p_i-q_i)^2\bigg|.
\end{align*}
Since $p_i+q_i\in[0,2]$ and $p_i-q_i\in[-1,1]$, the maps
$a\mapsto a^2$ on the two ranges admit globally Lipschitz
extensions, vanishing at zero, with Lipschitz constants $4$ and $2$,
respectively. Thus, the contraction principle
(see, e.g., \cite[Chapter 6]{vershynin2018high}) gives
\begin{align*}
 \mathbb{E}_{\varepsilon}
\sup_{p\in\calP,q\in\calQ}
\bigg|\sum_{i=1}^m\varepsilon_i p_iq_i\bigg| 
&\lesssim
\mathbb{E}_{\varepsilon}
\sup_{p\in\calP,q\in\calQ}
\bigg|\sum_{i=1}^m\varepsilon_i(p_i+q_i)\bigg|
+
\mathbb{E}_{\varepsilon}
\sup_{p\in\calP,q\in\calQ}
\bigg|\sum_{i=1}^m\varepsilon_i(p_i-q_i)\bigg|\nn\\
&\le
2\mathbb{E}_{\varepsilon}
\sup_{p\in\calP}
\bigg|\sum_{i=1}^m\varepsilon_i p_i\bigg|
+
2\mathbb{E}_{\varepsilon}
\sup_{q\in\calQ}
\bigg|\sum_{i=1}^m\varepsilon_i q_i\bigg|,
\end{align*}
which proves the claim.
\end{proof}

 The following Massart's concentration inequality follows by 
  taking $\eta=0.1$ in
\cite[Theorem 2]{Adamczak2008} and absorbing the resulting numerical
constants into $C$. Here, $\|f\|_{\infty}:=\sup_{x}|f(x)|$ denotes the supremum norm, where the supremum is taken over the underlying measurable space.

\begin{lem}[Massart's inequality; see
{\cite{massart2000constants}} and
{\cite[Theorem 2]{Adamczak2008}}]
\label{lem:boundedempiricalprocess}
Let $X_1,\dots,X_m$ be independent random variables taking values in
a measurable space, and let $\calF$ be a countable class of
measurable functions such that
\(
    \mathbb{E}f(X_i)=0,~
    \|f\|_{\infty}\le b
    ~
    (\forall f\in\calF),~ i\in[m].\)
Define
\[
    Z:=\sup_{f\in\calF}
    \bigg|\sum_{i=1}^m f(X_i)\bigg|,
    \qquad
    \sigma^2:=
    \sup_{f\in\calF}
    \sum_{i=1}^m\mathbb{E}f^2(X_i).
\]
Then, for every $t\ge0$, with probability at least $1-\exp(-t)$,
\[
    Z\le
    1.1\mathbb{E}Z
    +C\big(\sigma\sqrt{t}+bt\big),
\]
where $C>0$ is a universal constant.
\end{lem}

 The third lemma provides a uniform upper bound on Gaussian empirical quadratic forms over all size-$\eta m$ subsets of the samples. This concentration bound is recurring in subsequent analysis. 

 \begin{lem}   \cite[Theorem 2.10]{dirksen2021non} \label{lem:maxlsum}
 Let $\ba_1,\dots,\ba_m$ be i.i.d. copies of $N(0,\bI_n)$, and let $\calU\subset \mathbb{R}^n$.
      For any $\eta \in \{\frac{1}{m},\frac{2}{m},\cdots,\frac{m-1}{m},1\}$, the event 
      \begin{align*}
          \sup_{\bu\in \calU}\max_{\substack{\calG\subset [m]\\|\calG|= \eta m}}\,\bigg(\frac{1}{\eta m}\sum_{i\in \calG}|\ba_i^\top\bu|^2\bigg)^{1/2} \lesssim \frac{\omega(\calU)}{\sqrt{\eta m}} + \Big(\sup_{\bu\in\calU}\|\bu\|_2\Big)\sqrt{\log\frac{e}{\eta}}
      \end{align*}
      holds with probability at least $1-2\exp(-c\eta m\log(\frac{e}{\eta}))$.
  \end{lem}

  We are now ready to prove Lemma \ref{lem:Phixuniformbound}. 
\begin{proof}[Proof of Lemma \ref{lem:Phixuniformbound}]
We start with
\begin{align}
    \sup_{\bx\in\calX\setminus\{0\}}\Phi(\bx)
    &\le
    \sup_{\bx\in\calX\setminus\{0\}}
    \frac{1}{\|\bx\|_2}
    \sup_{\bu\in\Sigma^{n,*}_{s}}
    \Big|
    \bu^\top[\hat{\bS}_{\bx}-\tilde{\bS}_{\bx}]\bu
    \Big|+
    \sup_{\bx\in\calX\setminus\{0\}}
    \frac{1}{\|\bx\|_2}
    \sup_{\bu\in\Sigma^{n,*}_{s}}
    \Big|
    \bu^\top
    (\tilde{\bS}_{\bx}
    -\mathbb{E}[\tilde{\bS}_{\bx}])
    \bu
    \Big|.
    \label{decomPhix}
\end{align}
The first term can be bounded by
\begin{align}
&\sup_{\bx\in\calX\setminus\{0\}}
\frac{1}{\|\bx\|_2}
\sup_{\bu\in\Sigma^{n,*}_{s}}
\Big|
\bu^\top[\hat{\bS}_{\bx}-\tilde{\bS}_{\bx}]\bu
\Big|\nn\\
&\stackrel{(a)}{\le}
\sup_{\bx\in\calX\setminus\{0\}}
\frac{1}{\|\bx\|_2}
\sup_{\bu\in\Sigma^{n,*}_{s}}
\bigg|
\frac{1}{m}\sum_{i=1}^m
\Big[
\sfT_{2\lambda_{\bx}}(|\ba_i^\top\bx|)
-\sfT_{2\|\bx\|_2}(|\ba_i^\top\bx|)
\Big]
(\ba_i^\top\bu)^2
\bigg|\nn\\
&\stackrel{(b)}{\le}
\sup_{\bx\in\calX\setminus\{0\}}
\sup_{\bu\in\Sigma^{n,*}_{s}}
\frac{2}{m}\sum_{i=1}^m
\bigg|
\frac{\lambda_{\bx}}{\|\bx\|_2}-1
\bigg|
(\ba_i^\top\bu)^2\nn\\
&\stackrel{(c)}{\lesssim}
\sqrt{\frac{s\log(en/s)}{m}}
\cdot
\bigg[
\Gamma_s\Big(\frac{\bA}{\sqrt{m}}\Big)
\bigg]^2
\stackrel{(d)}{\lesssim}
\sqrt{\frac{s\log(en/s)}{m}},
\label{firstbound1}
\end{align}
where in $(a)$ we substitute (\ref{hatSxxx}) and
(\ref{idealdm}), $(b)$ follows from triangle inequality and
\[
    \sup_{a\in\mathbb{R}}
    |\sfT_{\xi_1}(a)-\sfT_{\xi_2}(a)|
    \le|\xi_1-\xi_2|,
\]
$(c)$ is due to the event in Lemma \ref{lem:l1l2}, and $(d)$ is due
to the event in Lemma \ref{ripupper}.

To bound the second term in (\ref{decomPhix}), we utilize the
homogeneity to restrict to
$\calX^*=\calX\cap\mathbb{S}^{n-1}$:
\begin{align}\nn
&\sup_{\bx\in\calX\setminus\{0\}}
\frac{1}{\|\bx\|_2}
\sup_{\bu\in\Sigma^{n,*}_{s}}
\Big|
\bu^\top
(\tilde{\bS}_{\bx}
-\mathbb{E}[\tilde{\bS}_{\bx}])
\bu
\Big|\nn\\
&=
\sup_{\bx\in\calX\setminus\{0\}}
\sup_{\bu\in\Sigma^{n,*}_{s}}
\bigg|
\frac{1}{m}\sum_{i=1}^m
\sfT_2(|\ba_i^\top\bar{\bx}|)
(\ba_i^\top\bu)^2
-
\mathbb{E}\Big[
\sfT_2(|\ba_i^\top\bar{\bx}|)
(\ba_i^\top\bu)^2
\Big]
\bigg|\nn\\
&=
\sup_{\bw\in\calX^*}
\sup_{\bu\in\Sigma^{n,*}_{s}}
\bigg|
\frac{1}{m}\sum_{i=1}^m
\sfT_2(|\ba_i^\top\bw|)
(\ba_i^\top\bu)^2
-
\mathbb{E}\Big[
\sfT_2(|\ba_i^\top\bw|)
(\ba_i^\top\bu)^2
\Big]
\bigg|.
\label{secondbound1}
\end{align}
The difficulty in (\ref{secondbound1}) is that
$(\ba_i^\top\bu)^2$ is unbounded. We truncate this factor, use
symmetrization and contraction for the resulting bounded process,
and then invoke Lemma \ref{lem:maxlsum} to control the discarded
tail uniformly over $\bu$. Specifically, let 
\begin{align}
      \tau
    =
    C_0
    \log\bigg(
    \frac{em}{s\log(en/s)}
    \bigg),\label{tauvalue}
\end{align}
where $C_0>0$ is a sufficiently large universal constant. By
triangle inequality, the process in (\ref{secondbound1}) is bounded
by $Z_{\rm bd}+Z_{\rm tail}$, where
\begin{align*}
Z_{\rm bd}
&:=
\sup_{\bw\in\calX^*}
\sup_{\bu\in\Sigma^{n,*}_{s}}
\bigg|
\frac{1}{m}\sum_{i=1}^m
\sfT_2(|\ba_i^\top\bw|)
\big((\ba_i^\top\bu)^2\wedge\tau\big)
-
\mathbb{E}\Big[
\sfT_2(|\ba_i^\top\bw|)
\big((\ba_i^\top\bu)^2\wedge\tau\big)
\Big]
\bigg|,
\\
Z_{\rm tail}
&:=
\sup_{\bw\in\calX^*}
\sup_{\bu\in\Sigma^{n,*}_{s}}
\bigg|
\frac{1}{m}\sum_{i=1}^m
\sfT_2(|\ba_i^\top\bw|)
\Big[
(\ba_i^\top\bu)^2
-\big((\ba_i^\top\bu)^2\wedge\tau\big)
\Big] 
-
\mathbb{E}\Big[
\sfT_2(|\ba_i^\top\bw|)
\Big(
(\ba_i^\top\bu)^2
-\big((\ba_i^\top\bu)^2\wedge\tau\big)
\Big)
\Big]
\bigg|.
\end{align*}

\paragraph{Bounding $Z_{\rm bd}$.}
Since the indexed functions are continuous in $(\bw,\bu)$, and
their expectations are continuous by bounded convergence, we may
restrict the two index sets to countable dense subsets without
changing the suprema. Let
$\{\varepsilon_i\}_{i=1}^m$ be independent Rademacher random
variables, independent of $\{\ba_i\}_{i=1}^m$. By symmetrization
(see, e.g., \cite[Chapter 6]{vershynin2018high}) and
Lemma \ref{lem:productcontraction}, applied conditionally on
$\{\ba_i\}_{i=1}^m$ with
\(
    p_i
    =
    \frac{\sfT_2(|\ba_i^\top\bw|)}{2},
    ~
    q_i
    =
    \frac{(\ba_i^\top\bu)^2\wedge\tau}{\tau},
\)
we obtain
\begin{align*}
\mathbb{E}Z_{\rm bd}
&\lesssim
\frac{1}{m}
\mathbb{E}
\sup_{\bw\in\calX^*}
\sup_{\bu\in\Sigma^{n,*}_{s}}
\bigg|
\sum_{i=1}^m
\varepsilon_i
\sfT_2(|\ba_i^\top\bw|)
\big((\ba_i^\top\bu)^2\wedge\tau\big)
\bigg|\nn\\
&\lesssim
\frac{\tau}{m}
\mathbb{E}
\sup_{\bw\in\calX^*}
\bigg|
\sum_{i=1}^m
\varepsilon_i
\sfT_2(|\ba_i^\top\bw|)
\bigg| +
\frac{1}{m}
\mathbb{E}
\sup_{\bu\in\Sigma^{n,*}_{s}}
\bigg|
\sum_{i=1}^m
\varepsilon_i
\big((\ba_i^\top\bu)^2\wedge\tau\big)
\bigg|.
\end{align*}
The maps $a\mapsto\sfT_2(|a|)$ and
$a\mapsto a^2\wedge\tau$ vanish at zero and are $1$-Lipschitz and
$2\sqrt{\tau}$-Lipschitz, respectively. Thus, another application
of the contraction principle in
\cite[Chapter 6]{vershynin2018high} yields
\begin{align*}
\mathbb{E}Z_{\rm bd}
&\lesssim
\frac{\tau}{m}
\mathbb{E}
\sup_{\bw\in\calX^*}
\bigg|
\sum_{i=1}^m
\varepsilon_i\ba_i^\top\bw
\bigg| +
\frac{\sqrt{\tau}}{m}
\mathbb{E}
\sup_{\bu\in\Sigma^{n,*}_{s}}
\bigg|
\sum_{i=1}^m
\varepsilon_i\ba_i^\top\bu
\bigg|\nn\\
&\stackrel{(a)}{=}
\frac{
\tau\omega(\calX^*)
+\sqrt{\tau}\,
\omega(\Sigma^{n,*}_s)
}{\sqrt{m}} \stackrel{(b)}{\lesssim}
\tau
\sqrt{\frac{s\log(en/s)}{m}},
\end{align*}
where $(a)$ follows from
\(
    \sum_{i=1}^m\varepsilon_i\ba_i
    \sim N(0,m\bI_n)\)
and the symmetry of $\calX^*$ and $\Sigma^{n,*}_s$, while $(b)$ follows from (\ref{Xstargwbound}) and
(\ref{gwsparse}).

We next pass from the expectation to a high-probability bound.
Apply Lemma \ref{lem:boundedempiricalprocess} to the centered
functions indexed by
$(\bw,\bu)\in\calX^*\times\Sigma^{n,*}_s$. Since
$0\le\sfT_2(|a|)\le2$, every centered function is bounded in
absolute value by $2\tau$. Moreover, for every
$\bw\in\calX^*$ and $\bu\in\Sigma^{n,*}_s$,
\begin{align*}
&\operatorname{Var}\Big(
\sfT_2(|\ba_i^\top\bw|)
\big((\ba_i^\top\bu)^2\wedge\tau\big)
\Big) \le
\mathbb{E}\bigg[
\sfT_2^2(|\ba_i^\top\bw|)
\big((\ba_i^\top\bu)^2\wedge\tau\big)^2
\bigg] \le
4\mathbb{E}(\ba_i^\top\bu)^4
=
12.
\end{align*}
Therefore, we take $t=s\log(en/s)$ in
Lemma \ref{lem:boundedempiricalprocess}. With probability at least
$1-\exp(-s\log(en/s))$,
\begin{align}
Z_{\rm bd}
&\lesssim
\mathbb{E}Z_{\rm bd}
+
\sqrt{\frac{s\log(en/s)}{m}}
+
\frac{\tau s\log(en/s)}{m} \lesssim
\tau
\sqrt{\frac{s\log(en/s)}{m}},
\label{boundZbd}
\end{align}
where the last inequality uses
$m\gtrsim s\log(en/s)$ and $\tau\ge1$.

\paragraph{Bounding $Z_{\rm tail}$.}  using $0\le\sfT_2(|a|)\le2$ and
\(
    a^2-(a^2\wedge\tau)
    =
    (a^2-\tau)_+
    \le
    a^2\mathbbm{1}(a^2>\tau),\)
we have
\begin{align} \label{Ztaildecompose}
Z_{\rm tail}
&\le
2\underbrace{\sup_{\bu\in\Sigma^{n,*}_s}
\frac{1}{m}
\sum_{i=1}^m
(\ba_i^\top\bu)^2
\mathbbm{1}\big(
(\ba_i^\top\bu)^2>\tau
\big)}_{:=\Xi_1}+
2\underbrace{\sup_{\bu\in\Sigma^{n,*}_s}
\mathbb{E}\Big[
(\ba_i^\top\bu)^2
\mathbbm{1}\big(
(\ba_i^\top\bu)^2>\tau
\big)
\Big]}_{:=\Xi_2}.
\end{align}

\paragraph{Bounding $\Xi_1.$}
Let
\(
    k
    =
    \big\lceil
    s\log\frac{en}{s}
    \big\rceil.\)
Under a sufficiently large implied constant in
$m\gtrsim s\log(en/s)$, we have $k\le m$. Applying
Lemma \ref{lem:maxlsum} with
$\eta=k/m$ and $\calU=\Sigma^{n,*}_s$, and then using
(\ref{gwsparse}), we obtain, with probability at least
$1-2\exp(-cs\log(en/s))$,
\begin{align}
&\sup_{\bu\in\Sigma^{n,*}_s}
\max_{\substack{\calG\subset[m]\\|\calG|=k}}
\frac{1}{m}
\sum_{i\in\calG}
(\ba_i^\top\bu)^2 \lesssim
\frac{\omega^2(\Sigma^{n,*}_s)}{m}
+
\frac{k}{m}\log\frac{em}{k} \lesssim
\frac{s\log(en/s)}{m}
\log\bigg(
\frac{em}{s\log(en/s)}
\bigg).
\label{largesttailbound}
\end{align}
For any $\bu\in\Sigma^{n,*}_s$, let
\(
    \calI_{\bu}
    :=
    \big\{
    i\in[m]:
    (\ba_i^\top\bu)^2>\tau
    \big\},\)
    then
    \begin{align}
        \Xi_1 \le \sup_{\bu\in\Sigma^{n,*}_s}\frac{1}{m}\sum_{i\in\calI_{\bu}}(\ba_i^\top\bu)^2.\label{Xi1setbound}
    \end{align}
We claim that, on the event (\ref{largesttailbound}),
$|\calI_{\bu}|<k$ for every $\bu\in\Sigma^{n,*}_s$. Indeed, if
$|\calI_{\bu}|\ge k$, we can choose a size-$k$ set
$\calG\subset\calI_{\bu}$, which together with (\ref{tauvalue}) gives
\[
    \frac{1}{m}
    \sum_{i\in\calG}
    (\ba_i^\top\bu)^2
    >
    \frac{k\tau}{m}
    \ge
    \frac{C_0s\log(en/s)}{m}
    \log\bigg(
    \frac{em}{s\log(en/s)}
    \bigg),
\]
contradicting (\ref{largesttailbound}) when $C_0$ is sufficiently
large.

In view of (\ref{Xi1setbound}), since all the summands are nonnegative, we may enlarge
$\calI_{\bu}$ to a size-$k$ subset of $[m]$ and conclude that
\begin{align}
&\Xi_1\le \sup_{\bu\in\Sigma^{n,*}_s}
\max_{\substack{\calG\subset[m]\\|\calG|=k}}
\frac{1}{m}
\sum_{i\in\calG}
(\ba_i^\top\bu)^2 \stackrel{{\rm (\ref{largesttailbound})}}{\lesssim}
\frac{s\log(en/s)}{m}
\log\bigg(
\frac{em}{s\log(en/s)}
\bigg).
\label{sampletailbound}
\end{align}

\paragraph{Bounding $\Xi_2.$} By rotational invariance and elementary Gaussian calculation,
\begin{align}
     \Xi_2:=\mathbb{E}\big[
    g^2\mathbbm{1}(g^2>\tau)
    \big]
    \stackrel{(a)}{\lesssim}
    (\sqrt{\tau}+1)\exp(-\tau/2)
    \stackrel{(b)}{\lesssim}
    \frac{s\log(en/s)}{m}
    \log\bigg(
    \frac{em}{s\log(en/s)}
    \bigg), \label{Xi2bound}
\end{align}
where $(a)$ follows from integration by parts, and $(b)$ holds because 
\[
    \exp(-\tau/2)
    =
    \bigg(
    \frac{em}{s\log(en/s)}
    \bigg)^{-C_0/2}
    \lesssim
    \frac{s\log(en/s)}{m},
\]
and
\(
    \sqrt{\tau}+1
    \lesssim
    \log\big(
    \frac{em}{s\log(en/s)}
    \big)
\)
when $C_0$ is a fixed sufficiently large constant. 
 
Therefore, substituting (\ref{sampletailbound}) and  (\ref{Xi2bound})
into (\ref{Ztaildecompose}) yields
\begin{align}
Z_{\rm tail}
&\lesssim
\frac{s\log(en/s)}{m}
\log\bigg(
\frac{em}{s\log(en/s)}
\bigg) \lesssim
\log\bigg(
\frac{em}{s\log(en/s)}
\bigg)
\sqrt{\frac{s\log(en/s)}{m}}.
\label{boundZtail}
\end{align}
Combining (\ref{secondbound1}), (\ref{boundZbd}), and
(\ref{boundZtail}), we conclude that the second term in
(\ref{decomPhix}) is bounded by
\[
    C
    \log\bigg(
    \frac{em}{s\log(en/s)}
    \bigg)
    \sqrt{\frac{s\log(en/s)}{m}}
\]
with probability at least
$1-C\exp(-cs\log(en/s))$. Combining this with
(\ref{firstbound1}) proves (\ref{Phixuniformbound}).
\end{proof}

 We are ready to combine the previous components to  establish (\ref{sprinix0}), and therefore conclude with (\ref{uniinispr}).
\begin{lem}[Uniform initialization] \label{proini} Under small enough $c_*$, if
\begin{align}
    m\gtrsim  s^3\log\frac{en}{s}
    \log^2\bigg(
    \frac{em}{s\log(en/s)}
    \bigg),\label{finalsamsprpr}
\end{align}
 then (\ref{sprinix0}) holds with probability at least $1-C\exp(-cs\log\frac{en}{s})$.
\end{lem}
\begin{proof}
By Lemma \ref{lem:inideter}, it is sufficient to verify (\ref{smalll2})--(\ref{maxconcen}) with small enough $\bar{c}$ for all $\bx\in \calX\setminus\{0\}$. (\ref{smalll2}) is guaranteed by using small enough $c_*$ in $\calX$.
    Under $m\gtrsim s\log(en/s)$, Lemma \ref{lem:l1l2} implies    (\ref{normaccc}) over all $\bx\in \calX\setminus\{0\}$. To ensure (\ref{operaconcen})--(\ref{maxconcen}) over $\bx\in \calX\setminus\{0\}$, in light of Lemma \ref{lem:Phixuniformbound}, with the promised probability we only need to guarantee 
    \[    \log\bigg(\frac{em}{s\log(en/s)}\bigg)
    \sqrt{\frac{s\log(en/s)}{m}}\le \frac{c}{s}\]
    for small enough $c$, which is exactly (\ref{finalsamsprpr}) with large enough hidden constant.  
\end{proof} 

\subsubsection{Projection Has No Effect on $\bx\in\mathcal{X}$}
As discussed in the proof outline in Section \ref{sec:iopspr}, the previous analysis already yields the following: if $\{\bx_t\}_{t\ge 0}$ is obtained by running 
\[
    \bx_{t+1}=H_s\big(\bx_t-\bh_{\bx}(\bx_t)\big),\quad t\ge 0
\]
with  $\bx_0$ obeying (\ref{sprinix0})   satisfies (\ref{sprcompressible}), then (\ref{sprcompressible}) holds. Now we want to show that  \emph{(\ref{sprcompressible}) remains true for   Algorithm \ref{alg:btaf}  which incorporates the additional projection  onto $\mathbb{B}_2^n(2\lambda_{\bx})$}. To this end, it is enough to show that \emph{the projection has no impact on the iterates satisfying (\ref{sprcompressible})}. More specifically, we proceed to show the following: if (\ref{sprcompressible}) holds, then  
\begin{align}\label{getproject}
    \|\bx_t\|_2 \le 2\lambda_{\bx},\quad \forall t\ge 0,~\forall \bx\in \calX. 
\end{align}

\subsubsection*{(vi) Establishing (\ref{getproject})} 
Under small enough $c,c_*$, (\ref{sprcompressible}) yields
\begin{align*}
    \|\bx_t-\bx\|_2\le c\|\bx\|_2 + 22 c_*\|\bx\|_2 < \frac{\|\bx\|_2}{2} ,
\end{align*}   
and therefore $\|\bx_t\|_2\le \frac{3}{2}\|\bx\|_2$. 
On the other hand,   (\ref{normaccc}) with small $c_1$ leads to $  \lambda_{\bx}\ge \frac{4\|\bx\|_2}{5}$, and hence 
\begin{align*}
    2\lambda_{\bx} \ge \frac{8}{5}\|\bx\|_2\ge \frac{3}{2}\|\bx\|_2 >\|\bx_t\|_2, \quad \forall\bx\in\calX.
\end{align*}
 Combining the previous two displays yields (\ref{getproject}). 

\subsubsection{A Separate Argument for  $\bx\in\mathcal{X}^c$}

We show that  $P_{\mathbb{B}_2^n(2\lambda_{\bx})}$  can guarantee the desired bound for \(\bx\in\calX^c= \big\{\bu \in \mathbb{R}^n:\tau_s(\bu)>c_*\|\bu\|_2\big\}.\)

\begin{lem} 
\label{proincom}
If $m\gtrsim s\log(\frac{en}{s})$, then with probability at least $1-2\exp(-s \log(\frac{en}{s}))$, 
\begin{align*}
    \dist(\bx_t,\bx)\le \Big(\frac{5}{c_*} + 4  \Big) \tau_s(\bx),\quad \forall \bx\in \calX^c .
\end{align*} 
\end{lem}
\begin{proof}
    For all $\bx\in \calX^c$, 
    \begin{align}\nn
        \lambda_{\bx} = \sqrt{\frac{\pi}{2}}\frac{\|\bA\bx\|_1}{m}& \le \sqrt{\frac{\pi}{2}}\frac{\|\bA\bx\|_2}{\sqrt{m}}=\sqrt{\frac{\pi}{2}}\frac{\|\bA\bx_{[s]}+\sum_{i\ge 1}\bA(\bx_{[(i+1)s]}-\bx_{[is]})\|_2}{\sqrt{m}}
        \\\nn &\stackrel{(a)}{\le} \sqrt{\frac{\pi}{2}}\frac{\|\bA\bx_{[s]}\|_2}{\sqrt{m}}+ \sqrt{\frac{\pi}{2}}\sum_{i\ge 1}\frac{\|\bA(\bx_{[(i+1)s]}-\bx_{[is]})\|_2}{\sqrt{m}} \\\nn&\stackrel{(b)}{\le} \sqrt{\frac{\pi}{2}}\cdot\frac{3}{2}\bigg(\|\bx_{[s]}\|_2+\tau_s(\bx)\bigg)
        \\&\stackrel{(c)}{\le} \frac{3}{2}\sqrt{\frac{\pi}{2}}\Big(\frac{1}{c_*}+1\Big)\tau_s(\bx)
        \,, \label{sprbounduseful} 
    \end{align}
    where $(a)$ is due to triangle inequality, $(b)$ is due to $\Gamma_{s}(\frac{\bA}{\sqrt{m}})\le \frac{3}{2}$ from Lemma \ref{ripupper}, $(c)$ holds because $\bx\in \calX^c$, which gives $\|\bx_{[s]}\|_2\le \|\bx\|_2\le \frac{1}{c_*}\tau_s(\bx)$. Due to the projection onto $\mathbb{B}_2^n(2\lambda_{\bx})$,  $\|\bx_t\|_2\le 2\lambda_{\bx}$ always holds, and therefore  
    $$\dist(\bx_t,\bx)\le \|\bx\|_2+2\lambda_{\bx} \le \Big(\frac{1}{c_*}+3\sqrt{\frac{\pi}{2}}(c_*^{-1}+1)\Big) \tau_s(\bx)$$
    for all $\bx\in \calX^c$. This leads to the desired claim. 
\end{proof}
\subsection{Proof of Theorem \ref{thm:nonuni} (Non-Uniform Instance Optimality)}\label{thm2}
We focus on nonzero $\bx$. Following the unified framework in Section \ref{sec:framenonuiop}, the proof consists of the following three steps:
\begin{enumerate} 
    \item {\bf (Decomposition of $\mathbb{R}^n$)} We choose \(\calX=\{\bu\in \mathbb{R}^n:\delta_s(\bu)\le c_*\|\bu\|_2\}\)
    for some small enough universal constant $c_*$;

    \item  {\bf (Instance Optimality if $\bx\in\calX$)}  
    
    \textbf{RAIC.} For some universal constant $c>0$, we establish   
    \begin{align}
        \bh_{\bx}(\bu)\sim {\rm RAIC}\bigg(\Sigma^{n,*}_{2s};\Sigma^n_s\cap \mathbb{B}_2^n(\bx;c\|\bx\|_2),\bx,\frac{\|\bu-\bx\|_2}{4}+6\|\bx-\bx_{[s]}\|_2\bigg) \label{sprnonuniraic}
    \end{align}
  Under small enough $c_*$, this renders (\ref{con1iht})--(\ref{con3iht}) with 
    \[\mu_1=\frac{1}{4},~~\mu_2=0,~~e(\bx)=6\|\bx-\bx_{[s]}\|_2,\quad\text{and}~~R_{\bx}^{\rm loc}=c\|\bx\|_2.\]
    
    \textbf{Initialization.} Under $m\gtrsim s^2$ (up to logarithmic factor), we further show that, w.h.p.,
    \begin{align}
        \bx_0\in\Sigma^n_s\cap \mathbb{B}_2^n(\bx;c\|\bx\|_2). \label{sprnonuniini}
    \end{align} 
    (Recall that we assume $\dist(\bx_0,\bx)=\|\bx_0-\bx\|_2$ with no loss of generality.)
    In turn, suppose that $\{\bx_t\}$ is generated by running  Algorithm \ref{alg:btaf} without $P_{\mathbb{B }_2^n(2\lambda_{\bx})}$, Lemma \ref{raiciht} guarantees 
    \begin{align}\label{non-uniformX}
        \|\bx_t-\bx\|_2\le \frac{c\|\bx\|_2}{2^t} + C\|\bx-\bx_{[s]}\|_2,\quad \forall t\ge 0. 
    \end{align}

        \textbf{Incorporating $P_{\mathbb{B}_2^n(2\lambda_{\bx})}$.} We further show that $P_{\mathbb{B}_2^n(2\lambda_{\bx})}$ has no effect on the iterates, hence the iterates of Algorithm \ref{alg:btaf} satisfy (\ref{sprcompressible}). 

    \item {\bf (Instance Optimality if $\bx\in\calX^c$)}   In this case, the desired error bound (\ref{l2l2iop}) reduces to a $\Theta(1)$ bound, and a simple argument suffices. 
\end{enumerate}


\subsubsection{Proof of the RAIC in (\ref{sprnonuniraic}) if $\bx\in\calX$}
If $\bx\in\calX\setminus\{0\}$, we shall first prove the RAIC in (\ref{sprnonuniraic}), i.e., 
\begin{align}\label{sprraicnonu}
    \|\bu-\bx- \bh_{\bx}(\bu)\|_{2s,2}\le \frac{\|\bu-\bx\|_2}{4}+6\|\bx-\bx_{[s]}\|_2 ,~~ \forall \bu\in \Sigma^n_s\cap \mathbb{B}_2^n(\bx;c\|\bx\|_2).
\end{align}
 We notice that (\ref{sprdecompose}) and (\ref{prraic_sparse}) remain valid under our current definition of $\calX$. This yields 
\begin{align}\label{nonuraicdecom}
    \|\bu-\bx-\bh_{\bx}(\bu)\|_{2s,2} \le \frac{\|\bu-\bx\|_2}{4} + \frac{5\|\bx-\bx_{[s]}\|_2}{4} + \|\bh_{\bx}(\bu)-\bh_{\bx_{[s]}}(\bu)\|_{2s,2}
\end{align}
for all $\bu\in \Sigma^n_s\cap \mathbb{B}_2^n(\bx;c\|\bx\|_2)$. To bound $\|\bh_{\bx}(\bu)-\bh_{\bx_{[s]}}(\bu)\|_{2s,2}$, 
\begin{align}\nn
    \|\bh_{\bx}(\bu)-\bh_{\bx_{[s]}}(\bu)\|_{2s,2} \nn&=
    \bigg\|\frac{1}{m}\sum_{i=1}^m\big(|\ba_i^\top\bx|-|\ba_i^\top\bx_{[s]}|\big)\sign(\ba_i^\top\bu)\ba_i\bigg\|_{2s,2}
    \\\nn&= \sup_{\bw\in\Sigma^{n,*}_{2s}}\frac{1}{m}\sum_{i=1}^m \big(|\ba_i^\top\bx|-|\ba_i^\top\bx_{[s]}|\big)\sign(\ba_i^\top\bu)\ba_i^\top\bw\\\nn
    &\stackrel{(a)}{\le}  \sup_{\bw\in\Sigma^{n,*}_{2s}}\frac{1}{m}\sum_{i=1}^m |\ba_i^\top(\bx-\bx_{[s]})| |\ba_i^\top\bw| 
    \\\nn
    &\stackrel{(b)}{=} \|\bx-\bx_{[s]}\|_2 \sup_{\bw\in \Sigma^{n,*}_{2s}}\frac{1}{m}\sum_{i=1}^m |\ba_i^\top\bx_e||\ba_i^\top\bw| 
    \\\nn
    &\stackrel{(c)}{\le} \|\bx-\bx_{[s]}\|_2 \frac{\|\bA\bx_e\|_2}{\sqrt{m}} \sup_{\bw\in\Sigma^{n,*}_{2s}}\frac{\|\bA\bw\|_2}{\sqrt{m}}
    \\\label{nonumismatchspr}&\stackrel{(d)}{\le} 2\|\bx-\bx_{[s]}\|_2 \frac{\|\bA\bx_e\|_2}{\sqrt{m}}, 
\end{align}
where $(a)$ is due to triangle inequality, in $(b)$ we suppose $\bx\ne \bx_{[s]}$ and let $\bx_e = \frac{\bx-\bx_{[s]}}{\|\bx-\bx_{[s]}\|_2}$, $(c)$ is due to Cauchy--Schwarz inequality, and $(d)$ holds because of  $\Gamma_{2s}(\frac{\bA}{\sqrt{m}})\le 2$ from Lemma \ref{ripupper}.

Moreover, since $\bA\bx_e \sim N(0,\bI_m)$,   Equation (3.7) of \cite{vershynin2018high} yields
\begin{align}\label{vergaunorm}
    \mathbb{P}\bigg(\frac{\|\bA\bx_e\|_2}{\sqrt{m}}\le 2 \bigg)\ge 1-2\exp(-c'm).
\end{align}
In turn, $\|\bh_{\bx}(\bu)-\bh_{\bx_{[s]}}(\bu)\|_{2s,2}\le 4\|\bx-\bx_{[s]}\|_2$. Combining with (\ref{nonuraicdecom}) yields (\ref{sprraicnonu}). 

\subsubsection{Proof of the Initialization Guarantee (\ref{sprnonuniini}) if $\bx\in\calX$} 
We shall use   Lemma \ref{lem:inideter} again. Since (\ref{smalll2}) is ensured by small enough $c_*$ in the definition of $\calX$, it remains to guarantee (\ref{normaccc})--(\ref{maxconcen}). 

\subsubsection*{Establishing (\ref{normaccc})}

Under $m\gtrsim \log n$, (\ref{normaccc}) is guaranteed by the following Lemma \ref{fixnormestimate} which gives a norm estimate bound tighter than Lemma \ref{lem:l1l2}.  

\begin{lem}\label{fixnormestimate}
    For any $\bx\in \mathbb{R}^n\setminus\{0\}$, 
    \[\mathbb{P}\bigg(|\lambda_{\bx}-\|\bx\|_2|\le C\sqrt{\frac{\log n}{m}}\|\bx\|_2\bigg) \ge 1-C_1n^{-1}.\]
\end{lem}
\begin{proof}
   This is a simple outcome of $\|\lambda_{\bx}/\|\bx\|_2-1\|_{\psi_2}=O(m^{-1/2})$ and the sub-Gaussian tail bound (see, e.g., \cite[Section 2]{vershynin2018high}). We omit further details.  
\end{proof}

 \subsubsection*{Establishing (\ref{operaconcen})--(\ref{maxconcen})}
 It remains to bound  $\Phi(\bx)$ and $\Psi(\bx)$ for a fixed $\bx\in\calX\setminus\{0\}$, in order to guarantee (\ref{operaconcen}) and (\ref{maxconcen}) with small enough $\bar{c}$. As we shall see, the transition  to a fixed $\bx$ allows for essentially tighter bounds than Lemma \ref{lem:Phixuniformbound}.  We first introduce a concentration inequality for product process. 

  \begin{lem}[{\cite[Theorem 1.13]{mendelson2016upper}}, see also {\cite[Page 938]{genzel2023unified}}] 
     \label{menproduct}Let $\ba_1,\dots,\ba_m$ be i.i.d. copies of a random vector $\ba$, and
let $\{g_{\bu}(\ba)\}_{\bu\in \calU}$ and $\{h_{\bv}(\ba)\}_{\bv\in\calV}$ be real-valued stochastic processes indexed by 
$\calU,\calV\subset \mathbb{R}^n$, respectively. Assume that there exist $K_1,K_2,r_1,r_2\geq 0$ such that
\[
\|g_{\bu}(\ba)-g_{\bu'}(\ba)\|_{\psi_2}\leq K_1\|\bu-\bu'\|_2,\quad
\|g_{\bu}(\ba)\|_{\psi_2} \leq r_1,\quad \forall\,\bu,\bu'\in \calU,
\]
and
\[
\|h_{\bv}(\ba)-h_{\bv'}(\ba)\|_{\psi_2} \leq K_2\|\bv-\bv'\|_2,\quad
\|h_{\bv}(\ba)\|_{\psi_2}\leq r_2,\quad \forall\,\bv,\bv'\in \calV.
\]
Then there exist universal constants $C>c>0$ such that, for every $t\geq 1$,  
\begin{align*}
&\sup_{\bu\in\calU}\sup_{\bv\in\calV} 
\bigg| \frac{1}{m}\sum_{i=1}^m g_{\bu}(\ba_i)h_{\bv}(\ba_i)
- \mathbb{E}\big[g_{\bu}(\ba)h_{\bv}(\ba)\big]\bigg|\\
&\leq C\bigg(\frac{(K_1\omega(\calU)+t r_1) \cdot(K_2 \omega(\calV)+t r_2)}{m}
+\frac{r_1 K_2\omega(\calV)+r_2 K_1\omega(\calU)+t r_1r_2}{\sqrt{m}}\bigg).
\end{align*}
holds with probability at least $1-2\exp(-ct^2)$. 
 \end{lem}

The following lemma shows that $m\gtrsim s^2\log n$ guarantees (\ref{operaconcen}) and (\ref{maxconcen}) with small enough $\bar{c}$, with the promised probability. 

\begin{lem}
\label{proonepoint} Let $\bx\in\calX\setminus\{0\}$. Under small enough $c_*$, if $m\gtrsim s\log\frac{en}{s}$, then \[\Phi(\bx) \lesssim\sqrt{\frac{s\log(en/s)}{m}}\quad \textrm{and}\quad \Psi(\bx)\lesssim \sqrt{\frac{\log n}{m}}\] hold with probability at least $1-C_1 n^{-1}-C_2\exp(-c_3s\log(en/s))$. 
\end{lem}

 \begin{proof}
We seek to bound 
\(\|\bx\|^{-1}_2 \sup_{\bu\in \calU_1}|\bu^\top(\hat{\bS}_{\bx}-\mathbb{E}[\tilde{\bS}_{\bx}])\bu|\,,~\text{where}~\,\calU_1=\Sigma^{n,*}_{s}~\textrm{or}~\{\be_j\}_{j=1}^n.\) These correspond to $\Phi(\bx)$ and $\Psi(\bx)$ in (\ref{operaconcen})--(\ref{maxconcen}), respectively.
For the fixed $\bx$ and $\bar{\bx}=\bx/\|\bx\|_2$, let \[F(\bu):=\frac{1}{m}\sum_{i=1}^m\big[\!\sfT_2(|\ba_i^\top\bar{\bx}|)(\ba_i^\top\bu)^2- \mathbb{E}\sfT_2(|\ba_i^\top\bar{\bx}|)(\ba_i^\top\bu)^2\big].\]
Revisiting the arguments in (\ref{decomPhix}) and (\ref{secondbound1})  yields 
    \begin{align*}
        \frac{1}{\|\bx\|_2}\sup_{\bu\in\calU_1}\big|\bu^\top(\hat{\bS}_{\bx}-\mathbb{E}[\tilde{\bS}_{\bx}])\bu\big|\le    \frac{1}{\|\bx\|_2}\sup_{\bu\in\calU_1}\big|\bu^\top(\hat{\bS}_{\bx}- \tilde{\bS}_{\bx})\bu\big| +  \sup_{\bu\in\calU_1}|F(\bu)|:=\Xi_3+\Xi_4
    \end{align*}
    for $\calU_1=\Sigma^{n,*}_s$ and $\{\be_j\}_{j=1}^n.$

    \paragraph{Bounding $\Xi_3$.} We revisit the arguments in (\ref{firstbound1}), but now we   consider a fixed $\bx$ and therefore can tighten the bound $|\frac{\lambda_{\bx}}{\|\bx\|_2}-1|\lesssim\sqrt{\frac{s\log(en/s)}{m}}$ used in $(c)$ therein to 
    \(
        \big|\frac{\lambda_{\bx}}{\|\bx\|_2}-1\big|\lesssim \sqrt{\frac{\log(n)}{m}}\)
    from Lemma \ref{fixnormestimate}. This yields 
    \begin{align*}
        \frac{1}{\|\bx\|_2}\sup_{\bu\in\calU_1}\big|\bu^\top(\hat{\bS}_{\bx}- \tilde{\bS}_{\bx})\bu\big|\le \bigg|\frac{\lambda_{\bx}}{\|\bx\|_2}-1\bigg|\sup_{\bu\in\Sigma^{n,*}_{s}} \frac{2}{m}\sum_{i=1}^m (\ba_i^\top\bu)^2\lesssim\sqrt{\frac{\log (n)}{m}} 
    \end{align*}
    with the promised probability.

    \paragraph{Bounding $\Xi_4$.} Since $\bar{\bx}$ is fixed, we  can apply Lemma \ref{menproduct} to \[\Xi_4:=\sup_{\bu\in\calU_1}|F(\bu)|\le \sup_{\bu,\bv\in\calU_1}\bigg|\frac{1}{m}\sum_{i=1}^m\big[\!\sfT_2(|\ba_i^\top\bar{\bx}|)\ba_i^\top\bu\ba_i^\top\bv- \mathbb{E}\sfT_2(|\ba_i^\top\bar{\bx}|)\ba_i^\top\bu\ba_i^\top\bv\big]\bigg|\] with ``$\calU=\{\bar{\bx}\}\times \calU_1,g_{\bu}(\ba_i)=\sfT_2(|\ba_i^\top\bar{\bx}|)\ba_i^\top\bu,\,K_1=O(1),\,r_1=O(1),\,\calV=\calU_1,\,h_{\bv}(\ba_i)=\ba_i^\top\bv,\,K_2=O(1),\,r_2=O(1)$'' yields 
    \begin{align*}
       \mathbb{P}\bigg(\sup_{\bu\in\calU_1}|F(\bu)|\lesssim \frac{\omega^2(\calU_1)+u}{m}+\sqrt{\frac{\omega^2(\calU_1)+u}{m}}\bigg) \ge 1-2\exp(-c_1u)
    \end{align*}
    for any $u\ge 1$.    
    Under   $m\gtrsim s^2\log n$,  we set  $u\asymp s\log(en/s)$ for $\calU_1=\Sigma^{n,*}_s$, and  $u\asymp \log n$ for $\calU_1=\{\be_j\}_{j=1}^n$, along with the estimates (\ref{gwej}) and (\ref{gwsparse}) to yield 
    $$\mathbb{P}\bigg(\sup_{\bu \in \Sigma^{n,*}_{s}} |F(\bu)|\lesssim  \sqrt{\frac{s\log(en/s)}{m}}\,,\, \sup_{\bu \in \{\be_j\}_{j=1}^n} |F(\bu)|\lesssim \sqrt{\frac{\log n}{m}}\,\bigg)\ge 1-2n^{-1}\,.$$
  Combining these pieces, we obtain
  \(\Phi(\bx) \lesssim\sqrt{\frac{s\log(en/s)}{m}}~\textrm{and}~\Psi(\bx)\lesssim \sqrt{\frac{\log n}{m}}\), as claimed.  
\end{proof}

We now invoke Lemma \ref{lem:inideter} to establish (\ref{sprnonuniini}) for the case of $\bx\in\calX$. Specifically, if $m\gtrsim s^2\log n$ (which is dictated by $\sqrt{\frac{\log n}{m}}\lesssim\frac{1}{2}$), then (\ref{normaccc})--(\ref{maxconcen}) hold for the fixed $\bx\in\calX$ with the promised probability. Then, Lemma \ref{lem:inideter} yields (\ref{sprnonuniini}).

\subsubsection{The Projection Has No Impact if $\bx\in\calX$}

We are ready to complete the analysis for the case of $\bx\in\calX$. 
Suppose $\{\bx_t\}_{t\ge 0}$ is generated by Algorithm \ref{alg:btaf} without $P_{\mathbb{B}_2^n(2\lambda_{\bx})}$, then the previous analysis guarantees (\ref{non-uniformX}) with high probability. To establish that the iterates of the original Algorithm \ref{alg:btaf} satisfy (\ref{non-uniformX}), we only need to show that the projection onto $\mathbb{B}_2^n(2\lambda_{\bx})$ does not affect the iterates. In fact, it is enough to show the following: if $\{\bx_t\}_{t\ge 0}$ satisfies (\ref{non-uniformX}), then \(\|\bx_t\|_2\le 2\lambda_{\bx}\) holds for any $t\ge 0$.  
To this end, we can set $c_*,c $ sufficiently small, so that (\ref{non-uniformX}) ensures
\begin{align*}
    \|\bx_t\|_2 \le \|\bx\|_2+ \|\bx_t-\bx\|_2 \le  c\|\bx\|_2+O(c_*)\|\bx\|_2 +\|\bx\|_2 <\frac{3}{2}\|\bx\|_2. 
\end{align*}
Also, on the event of Lemma \ref{fixnormestimate}, $\lambda_{\bx}\ge \frac{3}{4}\|\bx\|_2$ and hence $2\lambda_{\bx}\ge \frac{3}{2}\|\bx\|_2$. This delivers $\|\bx_t\|_2\le 2\lambda_{\bx}$ for all $t\ge 0$. 
   
\subsubsection{A Separate Argument if $\bx\in\calX^c$}

We assume that we are on the event of Lemma \ref{fixnormestimate}. Under the sample complexity $m\gtrsim s^2\log n$,  $
    \|\bx_t\|_2 \le 2\lambda_{\bx} \le 3\|\bx\|_2$ holds for any $t\ge 0$. Thus, 
\(
    \dist(\bx_t,\bx)\le \|\bx_t\|_2+\|\bx\|_2 \le  4\|\bx\|_2 \le \frac{4\|\bx-\bx_{[s]}\|_2}{c_*}.\)
This addresses the case of $\bx\in \calX^c$.

\section{Proofs for One-Bit Compressed Sensing}\label{sec:proof1bcs}
In this appendix, we work with $\bh_{\bx}(\bu) = \frac{1}{2m}\sum_{i=1}^m(\sign(\ba_i^\top\bu)-\sign(\ba_i^\top\bx))\ba_i$.  We allow $\bu,\bx$ to take values in $\mathbb{R}^n\setminus\{0\}$, and notice that 
\[\bh_{\bu}(\bx) = \bh_{\bu/\|\bu\|_2}\Big(\frac{\bx}{\|\bx\|_2}\Big),\quad\forall \bu,\bx\in\mathbb{R}^n\setminus\{0\}.\]
\subsection{Proof of Theorem \ref{thm:1bcsiop}  (Instance Optimal 1bCS)}\label{app:iop1bcsproof}
 We set \(\calX = \big\{\bu\in \mathbb{S}^{n-1}: \tau_s(\bu)\le c_*\big\}\)
 for some small enough constant $c_*$. In the outline of proof in Section \ref{sec:iop1bcs1}, it remains to prove (\ref{1bcsraic}). By definition, this is equivalent to 
 \begin{align}\nn
     &\|\bu-\bx-\sqrt{2\pi}\bh_{\bx}(\bu)\|_{2s,2}\le \frac{\|\bu-\bx\|_2}{10}+\frac{C_1s}{m}\log\Big(\frac{mn}{s^2}\Big)\log^{1/2}\Big(\frac{m}{s}\Big)+e(\bx),~\forall \bu\in\Sigma^{n,*}_s,~\bx\in\calX,\\
     &\textrm{where}~~e(\bx) \asymp \tau_s(\bx) \log\bigg(\frac{1}{\tau_s(\bx)}\bigg):=e_0(\bx). \label{1bcswantedraic}
 \end{align}

 \subsubsection{Proof of the RAIC in (\ref{1bcswantedraic})} 
  By $\bx=\bx_{[s]}+\bx-\bx_{[s]}$, 
  \begin{align}  
      \nn \bigg\|\bu-\bx-\sqrt{2\pi} \bh_{\bx}(\bu)\bigg\|_{2s,2}&\le  \underbrace{\bigg\|\bu-\bx_{[s]}-\sqrt{2\pi} \bh_{\bx_{[s]}}(\bu)\bigg\|_{2s,2}}_{:=\Xi_5}\\&+\|\bx-\bx_{[s]}\|_{2} + \sqrt{2\pi} \underbrace{\big\|\bh_{\bx}(\bu)-\bh_{\bx_{[s]}}(\bu)\big\|_{2s,2}}_{:=\Xi_6}  \label{1bcsdecomm}
  \end{align}

\subsubsection*{(i) Bounding $\Xi_5:=\big\|\bu-\bx_{[s]}-\sqrt{2\pi} \bh_{\bx_{[s]}}(\bu)\big\|_{2s,2}$} 
By triangle inequality,
  \begin{align}\nn
    \Xi_5\le \bigg\|\bx_{[s]}-\frac{\bx_{[s]}}{\|\bx_{[s]}\|_2}\bigg\|_2+ \bigg\|\bu-\frac{\bx_{[s]}}{\|\bx_{[s]}\|_2}-\sqrt{2\pi} h_{\bx_{[s]}/\|\bx_{[s]}\|_2}(\bu)\bigg\|_{2s,2}:=\bigg\|\bx_{[s]}-\frac{\bx_{[s]}}{\|\bx_{[s]}\|_2}\bigg\|_2+ \Xi_5',\label{2ndsparseraic}
  \end{align}
  where the first term is controlled by 
  \[\bigg\|\bx_{[s]}-\frac{\bx_{[s]}}{\|\bx_{[s]}\|_2}\bigg\|_2 = |1-\|\bx_{[s]}\|_2|\le \|\bx-\bx_{[s]}\|_2.\] 
  In light of $\bu,\,\bx_{[s]}/\|\bx_{[s]}\|_2\in \Sigma^{n,*}_s,$  
 $\Xi_5'$ can be controlled by the RAIC in  \cite[Theorem 3.3]{matsumoto2024binary}. For a formulation of RAIC more aligned with our setting, we refer to \cite{chen2024optimal}.

  \begin{lem}[{\cite[Theorem 5]{chen2024optimal}}, see also the restatement in {\cite[Theorem B.11]{chen2025unified}}]\label{lem:1bcsraic1}If $m\gtrsim s\log\frac{en}{s}$, then with probability at least $1-\exp(-cs\log\frac{en}{s})$,     
  \begin{align*}
      \big\|\bu-\bx-\sqrt{2\pi}\cdot \bh_{\bx}(\bu)\big\|_{2s,2}\le \frac{1}{10}\|\bu-\bx\|_2 +\frac{C_1s}{m}\log\Big(\frac{mn}{s^2}\Big)\log^{1/2}\Big(\frac{m}{s}\Big),\quad \forall \bu,\bx\in\Sigma^{n,*}_s.   
  \end{align*}
  \end{lem}
  \begin{proof}
      By Definition \ref{def:sparseraic}, the desired claim is equivalent to 
      \[\sqrt{2\pi}\bh_{\bx}(\bu)\sim {\rm RAIC}\Big(\Sigma^{n,*}_{2s};\Sigma^{n,*}_s,\bx,\frac{1}{10}\|\bu-\bx\|_2 +\frac{C_1s}{m}\log\Big(\frac{mn}{s^2}\Big)\log^{1/2}\Big(\frac{m}{s}\Big)\Big),~~\forall \bx\in\Sigma^{n,*}_s.\]
      Note that the general RAIC   in \cite[Definition B.1]{chen2025unified} is consistent with Definition \ref{def:sparseraic}. Therefore, we can apply \cite[Theorem B.11]{chen2025unified} with $\calC=\Sigma^n_s$ to obtain the following: if \(m\gtrsim \frac{s}{\delta}\log (\frac{12n}{s\delta})\) for some small enough $\delta$, then with probability at least $1-C\exp(-cs\log\frac{en}{s})$, 
      \[\sqrt{2\pi}\bh_{\bx}(\bu)\sim {\rm RAIC}\Big(\Sigma^{n,*}_{2s};\Sigma^{n,*}_s,\bx,O\Big(\sqrt{\delta\|\bu-\bx\|_2}+\delta\sqrt{\log(1/\delta)}\Big)\Big),~~\forall \bx\in\Sigma^{n,*}_s.\]
      Therefore, we can set $\delta\asymp \frac{s}{m}\log\frac{mn}{s^2}$ to establish the RAIC with error function $ \sqrt{\frac{C_2s}{m}\log(\frac{mn}{s^2})\|\bu-\bx\|_2}+\frac{C_3s}{m}\log(\frac{mn}{s^2})\sqrt{\log\frac{m}{s}}$. This is then relaxed to 
      claimed error function by 
      \[ \sqrt{\frac{C_2s}{m}\log\big(\frac{mn}{s^2}\big)\|\bu-\bx\|_2}\le \frac{\|\bu-\bx\|_2}{10}+ O\Big(\frac{s}{m}\log\frac{mn}{s^2}\Big).\]
      The proof is complete. 
  \end{proof}
 
On the event of Lemma \ref{lem:1bcsraic1}, for all $\bu\in\Sigma^{n,*}_s$ and $\bx\in\calX$, 
  \begin{align*}\nn \Xi_5' &\stackrel{(a)}{\le}  \frac{1}{10}\bigg\|\bu-\frac{\bx_{[s]}}{\|\bx_{[s]}\|_2}\bigg\|_2 + \frac{C_1s}{m}\log\Big(\frac{mn}{s^2}\Big) \log^{1/2}\Big(\frac{m}{s}\Big) 
  \\& \stackrel{(b)}{\le} \frac{1}{10}\|\bu-\bx\|_2 + \frac{\|\bx-\bx_{[s]}\|_2}{5} +\frac{C_1s}{m}\log\Big(\frac{mn}{s^2}\Big) \log^{1/2}\Big(\frac{m}{s}\Big), \nn 
  \end{align*}
where $(a)$ is due to the RAIC in Lemma \ref{lem:1bcsraic1}, $(b)$ is due to triangle inequality. Combining  these bounds, we arrive at  
 \begin{align}
     \label{T1final1bcs}
     \Xi_5\le \frac{1}{10}\|\bu-\bx\|_2 + \frac{6\|\bx-\bx_{[s]}\|_2}{5} + \frac{C_1s}{m}\log\Big(\frac{mn}{s^2}\Big)\log^{1/2}\Big(\frac{m}{s}\Big).
 \end{align}
  
  \subsubsection*{(ii) Bounding $\Xi_6:=\big\|\bh_{\bx}(\bu)-\bh_{\bx_{[s]}}(\bu)\big\|_{2s,2}$} By (\ref{1bcsgradient}),   
  \begin{align}&\Xi_6:= \|\bh_{\bx}(\bu)-\bh_{\bx_{[s]}}(\bu)\|_{2s,2} \nn
  \\\nn
      &=\bigg\|\frac{1}{2m}\sum_{i=1}^ m (\sign(\ba_i^\top\bu)-\sign(\ba_i^\top \bx))\ba_i -\frac{1}{2m}\sum_{i=1}^ m (\sign(\ba_i^\top\bu)-\sign(\ba_i^\top \bx_{[s]}))\ba_i\bigg\|_{2s,2}\\\nn
      &= \bigg\|\frac{1}{2m}\sum_{i=1}^m (\sign(\ba_i^\top\bx)-\sign(\ba_i^\top\bx_{[s]}))\ba_i\bigg\|_{2s,2}\\
      & = \sup_{\bu\in \Sigma^{n,*}_{2s}} \frac{1}{2m}\sum_{i=1}^m \big(\sign(\ba_i^\top\bx)-\sign(\ba_i^\top\bx_{[s]})\big)\ba_i^\top \bu\label{whattobound}
  \end{align}
  The key is to establish uniform and signal-dependent bound on \[L_{\bx}:=|\{i\in[m]:\sign(\ba_i^\top\bx)\ne \sign(\ba_i^\top\bx_{[s]})\}|,\] 
   the number of the hyperplanes $\{\bu:\ba_i^\top\bu=0\},~i\in[m]$ that separate $\bx$ and $\bx_{[s]}$.
We establish the following  novel hyperplane tessellation result. It will be proved along the subsequent analysis. 
\begin{lem}[Instance-dependent hyperplane tessellation] \label{lem:iophyper} Let $\calX=\{\bu\in\mathbb{S}^{n-1}:\sum_{i\ge 1}\|\bu_{[(i+1)s]}-\bu_{[is]}\|_2\le c_*\}$ for some small universal constant $c_*$. 
If $m\gtrsim s\log\frac{mn}{s^2}$, then with probability at least $1-C\exp(-cs\log\frac{en}{s})$, $L_{\bx}=|\{i\in[m]:\sign(\ba_i^\top\bx)\ne\sign(\ba_i^\top\bx_{[s]})\}|$ satisfies 
\[\frac{L_{\bx}}{m} \lesssim \frac{s}{m}\log\Big(\frac{mn}{s^2}\Big) +   \tau_s(\bx) \log^{1/2}\frac{1}{ \tau_s(\bx)},~~\forall\bx\in\calX .\]
\end{lem}

  We shall note a uniform bound on \(I_{\bx}(\eta):=\big|\big\{i\in[m]:\big|\ba_i^\top\frac{\bx_{[s]}}{\|\bx_{[s]}\|_2}\big|\le\eta\big\}\big|\) over $\bx\in\calX$.

  \begin{lem}\cite[Lemma B.1]{chen2024robust}\label{lem:numsmallmea}
     If $m\gtrsim s\log\frac{mn}{s^2}$, then  for any $\eta\in \frac{1}{m}\mathbb{Z}_s$, 
      \begin{align}\label{Ixetabound}
          \mathbb{P}\bigg(\sup_{\bx\in\calX}\,I_{\bx}(\eta) \le \eta m\bigg)\ge 1-\exp(-c\eta m).
      \end{align}
  \end{lem}
  \begin{proof}
       While Lemma B.1 in \cite{chen2024robust} is stated for $m$ complex Gaussian vectors, its proof yields the identical statement for   real Gaussian vectors  $a_i$'s, up to minor changes. Since $\bx_{[s]}/\|\bx_{[s]}\|_2 \in\Sigma^{n,*}_s$, it follows that 
       \begin{align*}
           \sup_{\bx\in\calX} \,I_{\bx}(\eta) \le \sup_{\bx\in\Sigma^{n,*}_s}\, I_{\bx}(\eta).
       \end{align*}
       Then by Lemma B.1 in \cite{chen2024robust}---along with the estimates (\ref{sparsecovering}) and (\ref{gwsparse})---we conclude that  (\ref{Ixetabound}) holds if
      $
           m\gtrsim  \frac{s}{\eta}\log\frac{n}{s\eta}$, which can be ensured by $\eta \gtrsim \frac{s}{m}\log\frac{mn}{s^2}$. Under $m\gtrsim s\log\frac{mn}{s^2}$, this is satisfied by $\eta\in\frac{\mathbb{Z}_s}{m}$; see (\ref{etarange}).  
  \end{proof}

 \begin{proof}[Proof of Lemma \ref{lem:iophyper}]
To start the proof, note that for any $\eta>0$ and $\bx\in\calX$, 
  \begin{align}
  \label{decomposeIJ}
      L_{\bx}
      &\le \underbrace{\bigg|\bigg\{i\in[m]: \bigg|\ba_i^\top\frac{\bx_{[s]}}{\|\bx_{[s]}\|_2}\bigg|\le \eta\bigg\}\bigg|}_{I_{\bx}(\eta)} + \underbrace{\bigg|\bigg\{i\in[m]:\bigg|\ba_i^\top\bigg(\bx-\frac{\bx_{[s]}}{\|\bx_{[s]}\|_2}\bigg)\bigg|>\frac{\eta}{2}\bigg\}\bigg|}_{J_{\bx}(\eta)}.
  \end{align}
To see this, if some $i$ is not counted in either $I_{\bx}(\eta)$ or $J_{\bx}(\eta)$, that is,
\begin{align*}
    \bigg|\ba_i^\top\frac{\bx_{[s]}}{\|\bx_{[s]}\|_2}\bigg|>\eta
    \quad\text{and}\quad
    \bigg|\ba_i^\top\Big(\bx-\frac{\bx_{[s]}}{\|\bx_{[s]}\|_2}\Big)\bigg|\le \frac{\eta}{2},
\end{align*}
then it follows that 
\[
\sign(\ba_i^\top\bx)
= \sign\!\bigg(\ba_i^\top\frac{\bx_{[s]}}{\|\bx_{[s]}\|_2}\bigg)
= \sign(\ba_i^\top\bx_{[s]}),
\]
and this index $i$ does not contribute to $L_{\bx}$. In addition, for each $\bx\in\calX$, $I_{\bx}(\eta)$ is non-decreasing while $J_{\bx}(\eta)$ is non-increasing in $\eta$, that is,
\[I_{\bx}(\eta_1)\le I_{\bx}(\eta_2),\quad J_{\bx}(\eta_1)\ge J_{\bx}(\eta_2),\qquad \forall  \eta_1\le \eta_2.\]

   To get signal-dependent bound,   we will choose different $\eta$ for different $\bx$. Thus,  we shall control  $I_{\bx}(\eta)$ and $J_{\bx}(\eta)$ for $\eta$ in the following  range: \begin{align}\label{etarange}
       \eta\in \frac{\mathbb{Z}_{s}}{m},\quad \text{where}~~\mathbb{Z}_s= \bigg\{l\in \mathbb{Z}: C^*s\log\frac{mn}{s^2}\le l\le m\bigg\}
   \end{align} 
   for some large enough $C^*$.

  By using Lemma \ref{lem:numsmallmea} and taking a union bound, the event 
  \begin{align}\label{boundIxeta}
       I_{\bx}(\eta)\le \eta m,\quad \forall \eta \in \frac{\mathbb{Z}_s}{m},~\forall \bx\in \calX 
  \end{align}
  holds with probability at least 
  \begin{align*}
      1-\sum_{\eta\in \mathbb{Z}_s/m}\exp(-c\eta m)\ge 1- \sum_{\mathbb{Z}\ni i\ge C^*s\log\frac{mn}{s^2}}\exp(-ci)\ge 1- C\exp\Big(-cs\log\frac{mn}{s^2}\Big). 
  \end{align*}
  We proceed on the event (\ref{boundIxeta}).

  We then turn to bounding $J_{\bx}(\eta)$, again for all $\eta\in\frac{\mathbb{Z}_s}{m}$. To this end, we instead look at the quantity 
  \begin{align*} 
     E_s(\eta,\bx):=\max_{|S|=\eta m}\,\bigg(\frac{1}{\eta m}\sum_{i\in S}\bigg|\ba_i^\top\bigg(\bx-\frac{\bx_{[s]}}{\|\bx_{[s]}\|_2}\bigg)\bigg|^2\bigg)^{1/2},\quad \eta\in \frac{\mathbb{Z}_s}{m}.
  \end{align*}
  For any $\bx$, it is not hard to see that 
  \begin{align}\label{Esimply}
       E_s(\eta,\bx)  \le \frac{\eta}{2}~~\Longrightarrow ~~J_{\bx}(\eta)\le \eta m.
  \end{align}
  (In fact, $J_{\bx}(\eta)>\eta m$ necessarily implies $E_s(\eta,\bx)>\frac{\eta}{2}$.)
  Therefore, we seek to establish a signal-dependent bound on  $E_s(\eta,\bx)$, which holds uniformly for all $\bx\in\calX$ and all $\eta\in \frac{\mathbb{Z}_s}{m}$.

  \paragraph{Bounding $E_s(\eta,\bx)$.} Applying Lemma \ref{lem:maxlsum} to $\calU= \Sigma^{n,*}_{2s}$ and using (\ref{gwsparse}), we obtain 
  \[\mathbb{P}\left(\sup_{\bu\in \Sigma^{n,*}_{2s}}\max_{\substack{S\subset [m]\\|S|= \eta m}}\,\bigg(\frac{1}{\eta m}\sum_{i\in S}|\ba_i^\top\bu|^2\bigg)^{1/2} \lesssim \sqrt{\frac{s\log(en/s)}{\eta m}}+ \sqrt{\log\frac{e}{\eta}}\right)\ge 1-2\exp(-c\eta m)\]
  for each $\eta\in \frac{\mathbb{Z}_s}{m}$. 
 Then   taking a union bound over $\eta\in \frac{\mathbb{Z}_s}{m}$,
  \begin{align}\label{maxlsum}
     \sup_{\bu\in \Sigma^{n,*}_{2s}}\max_{\substack{\calG\subset [m]\\|\calG|= \eta m}}\,\bigg(\frac{1}{\eta m}\sum_{i\in \calG}|\ba_i^\top\bu|^2\bigg)^{1/2} \lesssim \sqrt{\frac{s\log(en/s)}{\eta m}} + \sqrt{\log\frac{e}{\eta}},\quad \forall\eta\in\frac{\mathbb{Z}_s}{m}
  \end{align}
  holds with  probability at least 
  \(
      1-2\sum_{\eta \in\mathbb{Z}_s/m}\exp(-c\eta m) \ge 1-C\exp\big(-cs\log\frac{mn}{s}\big).\)
  We proceed
  on the event   (\ref{maxlsum}). For any $S\subset [m]$, we let $\bA_S\in \mathbb{R}^{|S|\times n}$ denote the submatrix of $\bA$ constituted by the rows in $S$ only. 
   Uniformly for all $\eta\in\frac{\mathbb{Z}_s}{m}$ and $\bx\in \calX$, 
  \begin{align*}
      & E_s(\eta,\bx) = \max_{|S|=\eta m} \frac{\|\bA_S(\bx-\frac{\bx_{[s]}}{\|\bx_{[s]}\|_2})\|_2}{\sqrt{\eta m}}\\
      &\stackrel{(a)}{\le}  \max_{|S|=\eta m} \frac{\|\bA_S(\bx_{[s]}-\frac{\bx_{[s]}}{\|\bx_{[s]}\|_2})\|_2}{\sqrt{\eta m}}+\max_{|S|=\eta m} \frac{\|\bA_S(\bx-\bx_{[s]})\|_2}{\sqrt{\eta m}} 
      \\
      &\stackrel{(b)}{\le} \Big|1-\|\bx_{[s]}\|_2\Big| \max_{|S|=\eta m}\frac{\|\bA_S\frac{\bx_{[s]}}{\|\bx_{[s]}\|_2}\|_2}{\sqrt{\eta m}} + \max_{|S|=\eta m}\sum_{i\ge 1}\frac{\|\bA_S(\bx_{[(i+1)s]}-\bx_{[is]})\|_2}{\sqrt{\eta m}} \\
      &\stackrel{(c)}{\le} \Big|1-\|\bx_{[s]}\|_2\Big| \sup_{\bu\in\Sigma^{n,*}_s}\max_{|S|=\eta m}\frac{\|\bA_S\bu\|_2}{\sqrt{\eta m}} +  \sum_{i\ge 1} \big\|\bx_{[(i+1)s]}-\bx_{[is]}\big\|_2\sup_{\bu\in\Sigma^{n,*}_s}\max_{|S|=\eta m}\frac{\|\bA_S\bu\|_2}{\sqrt{\eta m}} 
      \\
      &\stackrel{(d)}{\le} 2 \sum_{i\ge 1} \big\|\bx_{[(i+1)s]}-\bx_{[is]}\big\|_2\sup_{\bu\in\Sigma^{n,*}_s}\max_{|S|=\eta m}\frac{\|\bA_S\bu\|_2}{\sqrt{\eta m}}\\ 
      &\stackrel{(e)}{\lesssim} \sqrt{\log\frac{e}{\eta}}\cdot\tau_s(\bx), 
  \end{align*}
  where $(a)$ is due to triangle inequality, in $(b)$ we use $\bx-\bx_{[s]}=\sum_{i\ge 1}(\bx_{[(i+1)s]}-\bx_{[is]})$ and triangle inequality, in $(c)$ we use $\frac{\bx_{[s]}}{\|\bx_{[s]}\|_2},\,\frac{\bx_{[(i+1)s]}-\bx_{[is]}}{\|\bx_{[(i+1)s]}-\bx_{[is]}\|_2}\in\Sigma^{n,*}_s$,   $(d)$ holds because $|1-\|\bx_{[s]}\|_2|\le \|\bx-\bx_{[s]}\|_2\le \tau_s(\bx)$, $(e)$ is due to the event in Equation (\ref{maxlsum}).

    In light of (\ref{Esimply}), for any $\bx\in \calX$ and $\eta\in\frac{\mathbb{Z}_s}{m}$, there exists some small enough universal constant $c_0$ such that  
  \begin{align}\label{Jxetabound}
      \tau_s(\bx)\le \frac{c_0\eta}{\sqrt{\log (e/\eta)}}~\Longrightarrow~ |J_{\bx}(\eta)|\le \eta m.
  \end{align}
  Therefore, for each $\bx\in\calX$, we shall set  
  \begin{align}\label{etax11}
      &\eta =\eta_{\bx}: = \max\bigg\{ \frac{\lceil C^*s\log\frac{mn}{s^2}\rceil}{m},C_0  \tau_s(\bx) \log^{1/2}\frac{1}{ \tau_s(\bx)}\bigg\}
  \end{align}
  for some large enough $C_0$ such that $m\eta_{\bx}\in \mathbb{Z}$ and  
  the left-hand side of (\ref{Jxetabound}) holds. In turn,
\begin{align}\label{Jxunibound}
    J_{\bx}(\eta_{\bx})\le \eta_{\bx}m ,\quad\forall\bx\in\calX. 
\end{align}
By (\ref{boundIxeta}), we also have
\begin{align}\label{Ixunibound}
    I_{\bx}(\eta_{\bx})\le \eta_{\bx}m ,\quad \forall\bx\in\calX.
\end{align}
Combining (\ref{decomposeIJ}), (\ref{Jxunibound}) and (\ref{Ixunibound}) yields 
\begin{align}\label{signalboundLx}
    L_{\bx}\le 2\eta_{\bx} m,\quad\forall\bx\in\calX, \quad\text{where }\eta_{\bx}\textrm{ is defined in (\ref{etax11})}.  
\end{align}
The proof of Lemma \ref{lem:iophyper} is now complete. 
 \end{proof}

We are now ready to establish a final bound on $\Xi_6$ in (\ref{1bcsdecomm}). Using the expression in (\ref{whattobound}), it holds 
for all $\bx\in\calX$ that   
\begin{align}\nn
    \Xi_6 &\stackrel{(a)}{\le} \sup_{\bu\in\Sigma^{n,*}_{2s}}\frac{1}{2m}\sum_{i=1}^m |\sign(\ba_i^\top\bx)-\sign(\ba_i^\top\bx_{[s]})||\ba_i^\top\bu| \\ \nn&\stackrel{(b)}{\le} \sup_{\bu\in \Sigma^{n,*}_{2s}}\frac{1}{m}\sum_{\substack{i\in[m]: \sign(\ba_i^\top\bx)\\\ne\sign(\ba_i^\top\bx_{[s]})}}|\ba_i^\top\bu|\\\nn
    &\stackrel{(c)}{\le} \sup_{\bu\in\Sigma^{n,*}_{2s}}\max_{\substack{I\subset [m]\\|I|= 2\eta_{\bx}m}}\frac{1}{m}\sum_{i\in I}|\ba_i^\top\bu|     
    \\\nn
    &\stackrel{(d)}{\le} 4\eta_{\bx}  \sup_{\bu\in\Sigma^{n,*}_{2s}}\max_{\substack{I\subset [m]\\|I|= 2\eta_{\bx}m}}\bigg(\frac{\sum_{i\in I}|\ba_i^\top\bu|^2}{2\eta_{\bx}m}\bigg)^{1/2} \\ 
    &\stackrel{(e)}{\lesssim} \eta_{\bx}\sqrt{\log \frac{e}{\eta_{\bx}}} \nn\\& \stackrel{(f)}{\lesssim} \frac{s}{m}\log\Big(\frac{mn}{s^2}\Big)\log^{
    1/2}\Big(\frac{m}{s}\Big)  + e_0(\bx), \label{finalboundT21bcs} 
\end{align}  
where $(a)$ follows from triangle inequality, in $(b)$ we restrict to the nonzero summands, $(c)$ is due to  the event in Equation (\ref{signalboundLx}), $(d)$ is due to Cauchy--Schwarz inequality,  $(e)$ holds because of the event in (\ref{maxlsum}), in $(f)$ we substitute $\eta_{\bx}$ in (\ref{etax11}) and $e_0(\bx)$ in (\ref{1bcswantedraic}).

Substituting (\ref{T1final1bcs}) and (\ref{finalboundT21bcs}) into (\ref{1bcsdecomm}), along with $\|\bx-\bx_{[s]}\|_2\le e_0(\bx)$ for $\bx\in \calX$, we arrive at (\ref{1bcswantedraic}).

\subsection{Proof of Theorem \ref{thm:1bcsnonuiop} (Non-Uniform Instance Optimality)}\label{app:prove1bcsnonuni}
Following the unified framework in Section \ref{sec:framenonuiop}, the proof consists of the following three steps: 
\begin{enumerate}
    \item {\bf (Decomposition of $\mathbb{S}^{n-1}$)} For some small enough universal constant $c_*>0$, we set 
    \(\calX=\{\bu\in \mathbb{S}^{n-1}:\delta_s(\bu)\le c_*\}.\)    
    \item  {\bf (Instance optimality if $\bx\in\calX$)} We establish the RAIC 
    \begin{align}
        &\sqrt{2\pi}\bh_{\bx}(\bu)\sim {\rm RAIC}\bigg(\Sigma^{n,*}_{2s};\Sigma^{n,*}_s,\bx,\frac{\|\bu-\bx\|_2}{10}+\frac{C_1s}{m}\log\Big(\frac{mn}{s^2}\Big)\log^{1/2}\Big(\frac{m}{s}\Big)+C_2\delta_s(\bx)\bigg)\label{desiredraicnonu1bcs}
    \end{align}
    This renders (\ref{con1niht})--(\ref{ininiht}) with 
    \(\mu_1=\frac{1}{10},~R_{\bx}^{\rm loc}=\infty,~\bx_0=\be_1,\) 
    and hence Theorem \ref{raicniht} yields the desired estimate. 
    \item {\bf (Instance optimality if $\bx\in\calX^c$)} If $\bx\in\calX^c$, the crude bound $\|\bx_t-\bx\|_2\le 2$ suffices: $\|\bx_t-\bx\|_2\le \frac{2\delta_s(\bx)}{c_*}$. Note that this step is already complete.  
\end{enumerate}

It remains to establish (\ref{desiredraicnonu1bcs}), which by definition is equivalent to 
\begin{align}
    \|\bu-\bx - \sqrt{2\pi}\bh_{\bx}(\bu)\|_{2s,2} \le \frac{1}{10}\|\bu-\bx\|_2 +\frac{C_1s}{m}\log\Big(\frac{mn}{s^2}\Big)\log^{1/2}\Big(\frac{m}{s}\Big) + C_2\|\bx-\bx_{[s]}\|_2,\quad \label{nonuraic1bcs}
    \forall \bu \in \Sigma^{n,*}_s. 
\end{align} 

\subsubsection{Proof of the RAIC in (\ref{nonuraic1bcs})}
We   start with the decomposition  in (\ref{1bcsdecomm}) and use the bound on $T_1$ in (\ref{T1final1bcs})\,\footnote{Note that (\ref{T1final1bcs}) remains valid under the current definition of $\calX$.} and (\ref{whattobound}), we obtain  
\begin{align}\nn
    &\bigg\|\bu-\bx -\sqrt{2\pi}\bh_{\bx}(\bu)\bigg\|_{2s,2} \le \frac{\|\bu-\bx\|_2}{10} + \frac{C_1s}{m}\log\Big(\frac{mn}{s^2}\Big)\log^{1/2}\Big(\frac{m}{s}\Big)\\\label{195}
    & + \frac{11\|\bx-\bx_{[s]}\|_2}{5} + \sqrt{2\pi}\bigg\|\frac{1}{2m}\sum_{i=1}^m (\sign(\ba_i^\top\bx)-\sign(\ba_i^\top\bx_{[s]}))\ba_i\bigg\|_{2s,2}.
\end{align}
All that remains is to show that 
\begin{align} \label{desirednonu1bcs}
    \bigg\|\frac{1}{2m}\sum_{i=1}^m (\sign(\ba_i^\top\bx)-\sign(\ba_i^\top\bx_{[s]}))\ba_i\bigg\|_{2s,2}\lesssim \|\bx-\bx_{[s]}\|_2 + \frac{s\log(en/s)}{m}
\end{align}
holds with the promised probability. With no loss of generality, we assume $\bx\ne \bx_{[s]}$. By $\sign(\ba_i^\top\bx_{[s]})=\sign(\ba_i^\top\frac{\bx_{[s]}}{\|\bx_{[s]}\|_2})$, we seek to perform a sharp analysis with two unit vectors, $\bx$ and $\bx_{[s]}/\|\bx_{[s]}\|_2$. We shall define $(\bbeta_1,\bbeta_2)$, an orthogonal basis of ${\rm span}(\bx,\bx_{[s]})$, as  
\[\bbeta_1 := \frac{\bx- \bx_{[s]}/\|\bx_{[s]}\|_2}{\|\bx- \bx_{[s]}/\|\bx_{[s]}\|_2\|_2}\,,\quad\textrm{and}\quad \bbeta_2:=\frac{\bx+ \bx_{[s]}/\|\bx_{[s]}\|_2}{\|\bx+ \bx_{[s]}/\|\bx_{[s]}\|_2\|_2}.\]
We also introduce the shorthand \[E^{(i)}:=\{\sign(\ba_i^\top\bx)\ne \sign(\ba_i^\top\bx_{[s]})\},\quad i\in [m].\]
Under these conventions, for $i\in[m]$, 
\begin{align*}
    &\sign(\ba_i^\top\bx)-\sign(\ba_i^\top\bx_{[s]})=\sign(\ba_i^\top\bx)-\sign(\ba_i^\top\bx_{[s]}/\|\bx_{[s]}\|_2)\\&= 2 \mathbbm{1}(E^{(i)})\sign(\ba_i^\top(\bx-\bx_{[s]}/\|\bx_{[s]}\|_2))=  2 \mathbbm{1}(E^{(i)})\sign(\ba_i^\top\bbeta_1),
\end{align*}
and $\ba_i$ can be decomposed into 
\[\ba_i = (\ba_i^\top\bbeta_1)\bbeta_1 + (\ba_i^\top\bbeta_2)\bbeta_2 + \underbrace{\big[\ba_i - (\ba_i^\top\bbeta_1)\bbeta_1 - (\ba_i^\top\bbeta_2)\bbeta_2 \big]}_{:=\ba_i^\perp}.\]
By the preceding two identities,
\begin{align*}
    &\frac{1}{2m}\sum_{i=1}^m(\sign(\ba_i^\top\bx)-\sign(\ba_i^\top\bx_{[s]}))\ba_i = \frac{1}{m}\sum_{i=1}^m \mathbbm{1}(E^{(i)})\sign(\ba_i^\top\bbeta_1) \ba_i \\
    &= \frac{1}{m}\sum_{i=1}^m \mathbbm{1}(E^{(i)})|\ba_i^\top\bbeta_1|\bbeta_1 +\frac{1}{m}\sum_{i=1}^m \mathbbm{1}(E^{(i)})\sign(\ba_i^\top\bbeta_1)(\ba_i^\top\bbeta_2)\bbeta_2 + \frac{1}{m}\sum_{i=1}^m \mathbbm{1}(E^{(i)})\sign(\ba_i^\top\bbeta_1) \ba_i^\perp.
\end{align*}
Taking $\|\cdot\|_{2s,2}$ norm and using triangle inequality,
\begin{align}
    \nn&\bigg\|\frac{1}{2m}\sum_{i=1}^m(\sign(\ba_i^\top\bx)-\sign(\ba_i^\top\bx_{[s]}))\ba_i \bigg\|_{2s,2} \\\nn&\le \bigg|\frac{1}{m}\sum_{i=1}^m \mathbbm{1}(E^{(i)})|\ba_i^\top\bbeta_1|\bigg|\\\nn  
    &+ \bigg|\frac{1}{m}\sum_{i=1}^m \mathbbm{1}(E^{(i)})\sign(\ba_i^\top\bbeta_1)(\ba_i^\top\bbeta_2)\bigg| \\&+ \bigg\|\frac{1}{m}\sum_{i=1}^m \mathbbm{1}(E^{(i)})\sign(\ba_i^\top\bbeta_1) \ba_i^\perp\bigg\|_{2s,2}:=\Xi_7+\Xi_8+\Xi_9.\label{nonu1bcsT1to3}
\end{align}

\subsubsection*{Bounding $\Xi_7$ and $\Xi_8$} 
The main   tool to bound $\Xi_7,\Xi_8$ is the following moment-based Bernstein's inequality. 
\begin{lem}
    \cite[Theorem 2.10]{13concen} \label{lem:bernstein210} Let $X_1,...,X_n$ be independent random variables, and assume that for some $v,c>0$, 
    \(
        \sum_{i=1}^n \mathbb{E}|X_i|^q\le \frac{q!}{2}vc^{q-2}\) holds \(\textrm{for all integers }q\ge 2,\) 
    then  
    \[
        \mathbb{P}\bigg(\bigg|\sum_{i=1}^n(X_i-\mathbb{E}X_i)\bigg|\ge \sqrt{2vt}+ct\bigg) \le 2\exp(-t),\quad \forall t>0. 
    \]
\end{lem}

The moment bounds and expectations of $\mathbbm{1}(E^{(i)})|\ba_i^\top\bbeta_1|$ and $\mathbbm{1}(E^{(i)})\sign(\ba_i^\top\bbeta_1)(\ba_i^\top\bbeta_2)$ have been established in \cite{chen2024optimal}. 
\begin{lem}
\cite[Facts 1, 2]{chen2024optimal} For some universal constant $C$, 
\begin{gather*}
    \mathbb{E}\big(\mathbbm{1}(E^{(i)})|\ba_i^\top\bbeta_1|^p \big)\le C\|\bx-\bx_{[s]}\|_2 \frac{p!}{2}, \\
    \mathbb{E}\big(\mathbbm{1}(E^{(i)})|\ba_i^\top\bbeta_2|^p\big)\le C\|\bx-\bx_{[s]}\|_2 \frac{p!}{2}
\end{gather*}
hold for any integer $p\ge 2$. Moreover, 
\[\mathbb{E}(\mathbbm{1}(E^{(i)})|\ba_i^\top\bbeta_1|)=\frac{\|\bx-\bx_{[s]}\|_2}{\sqrt{2\pi}}\,,\quad \textrm{and}\quad \mathbb{E}(\mathbbm{1}(E^{(i)})\sign(\ba_i^\top\bbeta_1)\ba_i^\top\bbeta_2)=0.\]
\end{lem}

Therefore,  Lemma \ref{lem:bernstein210} with $X_i=\mathbbm{1}(E^{(i)})|\ba_i^\top\bbeta_1|$ yields
    \[\mathbb{P}\bigg(\bigg|\frac{1}{m}\sum_{i=1}^m\big\{\mathbbm{1}(E^{(i)})|\ba_i^\top\bbeta_1|-\mathbb{E}[\mathbbm{1}(E^{(i)})|\ba_i^\top\bbeta_1|]\big\}\bigg| \lesssim \sqrt{\frac{\|\bx-\bx_{[s]}\|_2t}{m}}+ \frac{t}{m}\bigg)\le 2\exp(-t)\]
for any $t>0.$ Setting $t= s\log(\frac{en}{s})$ and using triangle inequality along with $\mathbb{E}[\mathbbm{1}(E^{(i)})|\ba_i^\top\bbeta_1|]=\frac{\|\bx-\bx_{[s]}\|_2}{\sqrt{2\pi}}$, we obtain that 
\begin{align}
    \label{nonu1bcsT1}
    \mathbb{P}\bigg(\Xi_7  \lesssim \frac{s\log(en/s)}{m} + \|\bx-\bx_{[s]}\|_2 \bigg) \ge 1-2\exp(-s\log\frac{en}{s}). 
\end{align}

Analogously, a straightforward application of Lemma \ref{lem:bernstein210} with $X_i=\mathbbm{1}(E^{(i)})\sign(\ba_i^\top\bbeta_1)(\ba_i^\top\bbeta_2)$ yields
\begin{align}
    \label{nonu1bcsT2}
    \Xi_8 \lesssim \sqrt{\frac{\|\bx-\bx_{[s]}\|_2 s\log(en/s)}{m}}+\frac{s\log(en/s)}{m} \lesssim \|\bx-\bx_{[s]}\|_2+ \frac{s\log(en/s)}{m}
\end{align}
with probability at least $1-2\exp(-s\log\frac{en}{s}).$

\subsubsection*{Bounding $\Xi_9$} 
By rotational invariance, $\ba_i^\perp$ is independent of $(\ba_i^\top\bbeta_1,\ba_i^\top\bbeta_2)$. To bound 
\[\Xi_9= \bigg\|\frac{1}{m}\sum_{i=1}^m \mathbbm{1}(E^{(i)})\sign(\ba_i^\top\bbeta_1)\ba_i^\perp\bigg\|_{2s,2}=\sup_{\bu\in\Sigma^{n,*}_{2s}}\frac{1}{m}\sum_{i=1}^m\mathbbm{1}(E^{(i)})\sign(\ba_i^\top\bbeta_1)\langle\ba_i^\perp, \bu\rangle,\]
we first treat the randomness of $(\ba_i^\perp)_{i=1}^m$ and then deal with the randomness of $(\ba_i^\top\bbeta_1,\ba_i^\top\bbeta_2)_{i=1}^m$.  Conditioning on $(\ba_i^\top\bbeta_1,\ba_i^\top\bbeta_2)_{i=1}^m$, then $\mathbbm{1}(E^{(i)})$'s are deterministic, and a calculation based on \cite[Proposition 2.6.1]{vershynin2018high} yields 
\[\bigg\|\frac{1}{m}\sum_{i=1}^m \mathbbm{1}(E^{(i)})\sign(\ba_i^\top\bbeta_1)\ba_i^\perp\bigg\|_{\psi_2}\lesssim \frac{\sqrt{\sum_{i=1}^m \mathbbm{1}(E^{(i)})}}{m}.\]
In turn, \cite[Exercise 8.6.5]{vershynin2018high} (along with (\ref{gwsparse})) yields
\begin{align}
    \label{conditionbound}
    \mathbb{P}\left(\Xi_9 \lesssim \frac{\sqrt{\big[\sum_{i=1}^m \mathbbm{1}(E^{(i)})\big]s\log(en/s)}}{m}\bigg| (\ba_i^\top\bbeta_1,\ba_i^\top\bbeta_2)_{i=1}^m\right) \ge 1-2\exp\Big(-s\log\frac{en}{s}\Big). 
\end{align}
We now bound $\sum_{i=1}^m \mathbbm{1}(E^{(i)})=L_{\bx}$. In light of (e.g.,  \cite[Lemma 3.2]{goemans1995improved})
\begin{align*}
    &\mathbb{P}\big(\sign(\ba_i^\top\bx)\ne \sign(\ba_i^\top\bx_{[s]})\big) = \frac{\arccos(\langle \bx,\bx_{[s]}/\|\bx_{[s]}\|_2\rangle)}{\pi} \le \frac{1}{2}\bigg\|\bx-\frac{\bx_{[s]}}{\|\bx_{[s]}\|_2}\bigg\|_2 \le \|\bx-\bx_{[s]}\|_2,
\end{align*}
by letting 
\begin{align*}
    \eta = \max\bigg\{\frac{\lceil C_*s\log\frac{en}{s}\rceil}{m},\tilde{C}\|\bx-\bx_{[s]}\|_2\bigg\}\quad\textrm{for some $C_*,\tilde{C}\ge 2$} 
\end{align*}
such that $\eta m$ is an integer, 
 Chernoff bound yields 
\begin{align*}
    \mathbb{P}\big(L_{\bx}\le \eta m\big) \ge \mathbb{P}\Big({\rm Binomial}\Big(m,\frac{\eta}{2}\Big)\le \eta m\Big) \ge 1-\exp(-c'\eta m) \ge 1-\exp\Big(-c''s\log\frac{en}{s}\Big). 
\end{align*}
Substituting this into (\ref{conditionbound}) yields that
\begin{align}
    \label{nonu1bcsT3}
    \Xi_9\lesssim \|\bx-\bx_{[s]}\|_2 + \frac{s\log(en/s)}{m}
\end{align}
holds with the promised probability. Combining (\ref{nonu1bcsT1}), (\ref{nonu1bcsT2}), (\ref{nonu1bcsT3}), and (\ref{nonu1bcsT1to3}), we arrive at (\ref{desirednonu1bcs}) and complete the proof. 

\section{Proofs for Sparse ReLU Regression} \label{proof:relu}
We proceed to the proofs of sparse ReLU regression. Throughout this section, we define $\bh_{\bx}(\bu)$ as in (\ref{hxrelu}), i.e., $\bh_{\bx}(\bu)=\frac{1}{m}\sum_{i=1}^m[\relu(\ba_i^\top\bu)-\relu(\ba_i^\top\bx)][1+\sign(\ba_i^\top\bu)]\ba_i$. As the algorithm achieves exact recovery for $\bx=0$, the subsequent analysis treats nonzero $\bx$ only. 

\subsection{Proof of Theorem \ref{thm:reluuniform} (Instance Optimal Sparse ReLU Regression)}\label{app:uniioprelu}
Following the unified framework in Section \ref{sec:frameuniiop}, the main steps are outlined in the following:
\begin{enumerate}
    \item {\bf (Decomposition of $\mathbb{R}^n$)} Set $\calX = \{\bu\in\mathbb{R}^n:\tau_s(\bu)\le c_*\|\bu\|_2\}$ for some small enough universal constant $c_*>0$;

    \item {\bf (Instance Optimality for $\bx\in\calX$)} 

    \textbf{RAIC.} We establish the following signal-dependent RAIC: for some universal constant $c>0$, 
    \begin{align}\label{reluraic}
        \bh_{\bx}(\bu)\sim {\rm RAIC}\bigg(\Sigma^{n,*}_{2s};\Sigma^n_s\cap \mathbb{B}_2^n(\bx;c\|\bx\|_2),\bx,\frac{\|\bu-\bx\|_2}{4}+6\tau_s(\bx)\bigg),~\forall \bx\in\calX\setminus\{0\}.
    \end{align}
    By definition, this reads
     \begin{align}
        &\|\bu-\bx-\bh_{\bx}(\bu)\|_{2s,2}\le \frac{\|\bu-\bx\|_2}{4} + 6\tau_s(\bx), \quad\forall \bx\in\calX\setminus\{0\},~\forall \bu\in \Sigma^n_s\cap \mathbb{B}_2^n(\bx;c\|\bx\|_2).  
        \label{reluraicdesired}
    \end{align}
    By the conventions in Section \ref{sec:unifieda}, this renders 
    \(\mu_1= \frac{1}{4},~  \mu_2=0,~ e(\bx) = 6\tau_s(\bx),~
    R_{\bx}^{\rm loc}=c\|\bx\|_2,\)
    thus satisfying (\ref{con1iht})--(\ref{con3iht}) as long as $c_*$ in the definition of $\calX$ is chosen small enough. 
    
    \textbf{Initialization.} We then establish 
 \begin{align}
        \label{reluinitialbound}
        \sup_{\bx\in\calX\setminus\{0\}} \frac{\|\bx_0-\bx\|_2}{\|\bx\|_2}\le C' \sqrt{\frac{s\log(en/s)}{m}} +3c_*. 
    \end{align}
    Combining  with $\bx_0\in\Sigma^n_s$, we arrive at (\ref{iniiht}) under $m\gtrsim s\log\frac{en}{s}$ and small enough $c_*$. Therefore, by Theorem \ref{thm:ihtiopX}, $\{\bx_t\}_{t\ge 0}$ produced by   Algorithm \ref{alg:reluiht} without $P_{\mathbb{B}_2^n(2\|\bx_0\|_2)}$ satisfies 
    \begin{align}\label{reluiopX}
        \|\bx_t-\bx\|_2 \le \frac{c\|\bx\|_2}{2^t} + 30 \tau_s(\bx), \quad \forall t\ge 0,~\forall\bx\in\calX\setminus\{0\}.
    \end{align}
    
    \textbf{Incorporating $P_{\mathbb{B}_2^n(2\|\bx_0\|_2)}$.} Built upon (\ref{reluiopX}), we further show that $P_{\mathbb{B}_2^n(2\|\bx_0\|_2)}$ has no effect on the iterates of $\bx\in\calX\setminus\{0\}$, thus (\ref{reluiopX}) is satisfied by the original Algorithm \ref{alg:reluiht}.

    \item {\bf (Instance Optimality for $\bx\in\calX^c$)} By using $P_{\mathbb{B}_2^n(2\|\bx_0\|_2)}$ and the definition of $\calX$, a simple argument suffices. 
\end{enumerate}


\subsubsection{Proof of the RAIC in (\ref{reluraicdesired})}
We start with the following RAIC over exactly sparse signals. 
\begin{lem}\cite[Theorem B.10]{chen2025unified}\label{lem:reluraic}
If $m\gtrsim s\log\frac{en}{s}$, then with probability at least $1-c_1\exp(-c_2m)$,
\[\|\bu-\bx- \bh_{\bx}(\bu)\|_{2,2s}\le \frac{1}{4}\|\bu-\bx\|_2,\quad \forall \bu,\bx\in \Sigma^n_s\setminus \{0\}~\textrm{obeying}~\|\bu-\bx\|_2\le \tilde{c}\|\bx\|_2\]
holds for some small enough universal constant $\tilde{c}>0$.
\end{lem}
\begin{proof}
    By Definition \ref{def:sparseraic}, the desired claim is equivalent to 
    \[\bh_{\bx}(\bu)\sim {\rm RAIC}\Big(\Sigma^{n,*}_{2s};\Sigma^n_s\cap \mathbb{B}_2^n(\bx;\tilde{c}\|\bx\|_2),\bx,\frac{\|\bu-\bx\|_2}{4}\Big),\quad \forall \bx\in\Sigma^n_s\setminus\{0\}.\]
    Therefore, the result follows from \cite[Theorem B.10]{chen2025unified} with $\calK=\Sigma^{n,*}_{2s}$, $\calC=\Sigma^n_s$. 
\end{proof}
We now seek an extension to $\bx\in \calX\setminus\{0\}$. We set $c$ in (\ref{reluraicdesired}) as $c=\frac{\tilde{c}}{2}$, where $\tilde{c}$ is the universal constant in Lemma \ref{lem:reluraic}. For any $\bx\in \calX\setminus\{0\}$ and $\bu\in\Sigma^n_s\cap \mathbb{B}_2^n(\bx;c\|\bx\|_2)$, triangle inequality gives
\begin{align}
    \|\bu-\bx-\bh_{\bx}(\bu)\|_{2s,2}\le \|\bu-\bx_{[s]}-\bh_{\bx_{[s]}}(\bu)\|_{2s,2}+\|\bx-\bx_{[s]}\|_2+ \|\bh_{\bx}(\bu)-\bh_{\bx_{[s]}}(\bu)\|_{2s,2}.\label{reuserelu11}
\end{align}
By $\bu,\bx_{[s]}\in\Sigma^n_s$ and $\|\bu-\bx_{[s]}\|_2 \le \|\bu-\bx\|_2+\|\bx-\bx_{[s]}\|_2\le (c+c_*)\|\bx\|_2\le \tilde{c}\|\bx\|_2$ (the last inequality is ensured by setting $c_*\le \frac{\tilde{c}}{2}$), Lemma \ref{lem:reluraic} gives
\begin{align}
    \|\bu-\bx_{[s]}-\bh_{\bx_{[s]}}(\bu)\|_{2s,2}\le \frac{\|\bu-\bx_{[s]}\|_2}{4} \le \frac{\|\bu-\bx\|_2}{4} + \frac{\|\bx-\bx_{[s]}\|_2}{4}.\label{reuserelu22}
\end{align}
It remains to bound $ \|\bh_{\bx}(\bu)-\bh_{\bx_{[s]}}(\bu)\|_{2s,2}$, which can be accomplished by 
\begin{align}\nn
    &\|\bh_{\bx}(\bu)-\bh_{\bx_{[s]}}(\bu)\|_{2s,2}=\sup_{\bw\in\Sigma^{n,*}_{2s}}\langle \bw, \bh_{\bx}(\bu)-\bh_{\bx_{[s]}}(\bu)\rangle\\\nn
    &= \sup_{\bw\in\Sigma^{n,*}_{2s}}\frac{1}{m}\sum_{i=1}^m[\relu(\ba_i^\top\bx_{[s]})-\relu(\ba_i^\top\bx)][1+\sign(\ba_i^\top\bu)]\ba_i^\top\bw \\\nn
    &\stackrel{(a)}{\le} \sup_{\bw\in\Sigma^{n,*}_{2s}}\frac{2}{m}\sum_{i=1}^m|\relu(\ba_i^\top\bx_{[s]})-\relu(\ba_i^\top\bx)||\ba_i^\top\bw| 
    \\\nn&\stackrel{(b)}{\le} \sup_{\bw\in\Sigma^{n,*}_{2s}}\frac{2}{m}\sum_{i=1}^m |\ba_i^\top(\bx-\bx_{[s]})||\ba_i^\top\bw|\\ 
    &\stackrel{(c)}{\le} 4\sum_{j\ge 1}\|\bx_{[(j+1)s]}-\bx_{[js]}\|_2=4\tau_s(\bx), \label{mismatchrelu} 
\end{align}
where in $(a)$ we use triangle inequality and $|1+\sign(\ba_i^\top\bu)|\le 2$, $(b)$ holds because of $|\relu(a)-\relu(b)|\le|a-b|$, $(c)$ is due to the bound derived in   (\ref{gradientmismatch}). 
Putting the pieces together, along with $\|\bx-\bx_{[s]}\|_2\le\tau_s(\bx)$, we arrive at (\ref{reluraicdesired}).  

\subsubsection{Proof of the Initialization Guarantee (\ref{reluinitialbound})}
In view of $\bx_0 = H_s(\frac{1}{m}\sum_{i=1}^m 2y_i\ba_i) = H_s(\frac{2}{m}\sum_{i=1}^m\relu(\ba_i^\top\bx)\ba_i)$, 
\begin{align}\nn
    \sup_{\bx\in \calX\setminus\{0\}}\frac{\|\bx_0-\bx\|_2}{\|\bx\|_2}&=\sup_{\bx\in\calX\setminus\{0\}}\bigg\|H_s\bigg(\frac{2}{m}\sum_{i=1}^m\relu(\ba_i^\top\frac{\bx}{\|\bx\|_2})\ba_i\bigg)-\frac{\bx}{\|\bx\|_2}\bigg\|_2 \\\nn
    &\stackrel{(a)}{=}\sup_{\bx\in\calX^*}\bigg\|H_s\bigg(\frac{2}{m}\sum_{i=1}^m\relu(\ba_i^\top\bx)\ba_i\bigg)-\bx\bigg\|_2 \\\nn
    &\stackrel{(b)}{\le} \sup_{\bx\in\calX^*} \bigg\|H_s\bigg(\frac{2}{m}\sum_{i=1}^m\relu(\ba_i^\top\bx)\ba_i\bigg)-\bx_{[s]}\bigg\|_2 + \sup_{\bx\in\calX^*}\|\bx-\bx_{[s]}\|_2 \\\nn
    &\stackrel{(c)}{\le} \sup_{\bx\in\calX^*}2\bigg\|\frac{2}{m}\sum_{i=1}^m\relu(\ba_i^\top\bx)\ba_i- \bx_{[s]} \bigg\|_{2s,2} + c_* \\ \label{reluinieq1}
    &\stackrel{(d)}{\le}  \sup_{\bx\in\calX^*}2\bigg\|\frac{2}{m}\sum_{i=1}^m\relu(\ba_i^\top\bx)\ba_i- \bx \bigg\|_{2s,2} + 3c_*,
\end{align}
where $(a)$ is due to $\frac{\bx}{\|\bx\|_2}\in\calX^*=\calX\cap \mathbb{S}^{n-1}$, $(b)$ is due to triangle inequality, $(c)$ follows from Lemma \ref{dualbound} and the definition of $\calX$, 
$(d)$ follows from triangle inequality and again the definition of $\calX$.

Moreover, the first term is bounded by 
\begin{align}\nn
    &\sup_{\bx\in\calX^*}\bigg\|\frac{2}{m}\sum_{i=1}^m\relu(\ba_i^\top\bx)\ba_i- \bx \bigg\|_{2s,2} \\&= \nn\sup_{\bx\in\calX^*} \sup_{\bw\in\Sigma^{n,*}_{2s}}\frac{2}{m}\sum_{i=1}^m\relu(\ba_i^\top\bx)\ba_i^\top\bw- \bx^\top\bw \\\nn
    &\stackrel{(a)}{=} \sup_{\bx\in\calX^*} \sup_{\bw\in\Sigma^{n,*}_{2s}} \frac{1}{m}\sum_{i=1}^m\bigg[2\relu(\ba_i^\top\bx)\ba_i^\top\bw- \mathbb{E}(2\relu(\ba_i^\top\bx)\ba_i^\top\bw)\bigg] \\\label{reluinieq2}
    & \stackrel{(b)}{\lesssim} \sqrt{\frac{s\log(en/s)}{m}},
\end{align}
where $(a)$ is due to $\mathbb{E}(2\relu(\ba_i^\top\bx)\ba_i^\top\bw)= \bx^\top\bw$, $(b)$ follows from Lemma \ref{menproduct}, the estimates (\ref{gwsparse}) and (\ref{Xstargwbound}), and $m\gtrsim s\log\frac{en}{s}$, and it holds with probability at least $1-2\exp(-c's\log\frac{en}{s})$. 
The claim follows. 

\subsubsection{The Projection Has No Effect on $\bx\in\calX\setminus\{0\}$}
Suppose that $\{\bx_t\}_{t=0}^\infty$ is generated by running Algorithm \ref{alg:reluiht} without $P_{\mathbb{B}_2^n(2\|\bx_0\|_2)}$, then the RAIC and initialization bound guarantees (\ref{reluiopX}), which along with $\tau_s(\bx)\le c_*\|\bx\|_2$ implies
\begin{align*}
    \|\bx_t-\bx\|_2\le (c+30c_*)\|\bx\|_2 ~~\Longrightarrow~~ \|\bx_t\|_2 \le (1+c+30c_*)\|\bx\|_2,\quad \forall t\ge 0,~\bx\in \calX\setminus\{0\}.
\end{align*}
On the other hand, under $m\gtrsim s\log\frac{en}{s}$, 
\begin{align*}
    {\rm(\ref{reluinitialbound})}~\Longrightarrow~ \sup_{\bx\in\calX\setminus\{0\}}\frac{\|\bx_0-\bx\|_2}{\|\bx\|_2}\le 4c_*~\Longrightarrow~ \|\bx_0\|_2  \ge (1-4c_*)\|\bx\|_2,~~\forall \bx\in\calX\setminus\{0\}.
\end{align*}
Thus, so long as $c$ and $c_*$ are small enough, 
\begin{align*}
    \|\bx_t\|_2\le 2\|\bx_0\|_2,\quad \forall t\ge 0,~\bx\in \calX\setminus\{0\}.
\end{align*}
This shows that adding the projection $P_{\mathbb{B}_2^n(2\|\bx_0\|_2)}$ does not affect $\{\bx_t\}_{t\ge 0}$ satisfying (\ref{reluiopX}), and therefore $\{\bx_t\}_{t\ge 0}$ obtained from Algorithm \ref{alg:reluiht} also satisfies (\ref{reluiopX}).

\subsubsection{A Separate Argument for $\bx\in\calX^c$}
Finally, we handle $\bx\in\calX^c $ by a different argument based on $P_{\mathbb{B}_2^n(2\|\bx_0\|_2)}$. For all $\bx\in \calX^c$, due to the projection, $\{\bx_t\}_{t\ge 0}$ (the iterates of Algorithm \ref{alg:reluiht})  satisfy 
\begin{align*}
    \|\bx_t\|_2&\le 2\|\bx_0\|_2 = \bigg\|\frac{4}{m}\sum_{i=1}^m\relu(\ba_i^\top\bx)\ba_i\bigg\|_{s,2}=\sup_{\bw\in\Sigma^{n,*}_s}\frac{4}{m}\sum_{i=1}^m\relu(\ba_i^\top\bx)\ba_i^\top\bw\\
    &\stackrel{(a)}{\le} \sup_{\bw\in\Sigma^{n,*}_s}\frac{4}{m}\sum_{i=1}^m|\ba_i^\top\bx||\ba_i^\top\bw| \stackrel{(b)}{\le} \frac{4\|\bA\bx\|_2}{\sqrt{m}}\sup_{\bw\in\Sigma^{n,*}_s}\frac{\|\bA\bw\|_2}{\sqrt{m}}\stackrel{(c)}{\le} \frac{6\|\bA\bx\|_2}{\sqrt{m}} \stackrel{(d)}{\le} 9\Big(\frac{1}{c_*}+1\Big) \tau_s(\bx), 
\end{align*}
where $(a)$ is due to triangle inequality, $(b)$ is due to Cauchy--Schwarz inequality, $(c)$ is due to Lemma \ref{ripupper}, and $(d)$ is obtained by using the bound established in (\ref{sprbounduseful}).
Therefore, by triangle inequality and the definition of $\calX^c$,
\[\|\bx_t-\bx\|_2 \le \|\bx_t\|_2+\|\bx\|_2 \le \Big(\frac{10}{c_*}+9\Big)\tau_s(\bx),\quad\forall t\ge 0,~\bx\in \calX^c,\]
as desired.

\subsection{Proof of Theorem \ref{thm:relunon-uniform} (Non-Uniform Instance Optimality)}\label{app:nonuioprelu}
Following the unified framework in Section \ref{sec:framenonuiop}, the proof consists of the following steps: 
\begin{enumerate}
    \item {\bf (Decomposition of $\mathbb{R}^n$)} Let $\calX = \{\bu\in \mathbb{R}^n:\delta_s(\bu)\le c_*\|\bu\|_2\}$ for some small enough universal constant $c_*>0$. 
    \item {\bf (Instance Optimality if $\bx\in\calX$)} For the fixed $\bx$, we establish the RAIC
    \begin{align}
        \bh_{\bx}(\bu)\sim {\rm RAIC}\bigg(\Sigma^{n,*}_{2s};\Sigma^n_s\cap \mathbb{B}_2^n(\bx;c\|\bx\|_2),\bx,\frac{\|\bu-\bx\|_2}{4}+10\|\bx-\bx_{[s]}\|_2\bigg). \label{nonuraicrelu}
    \end{align}  By the notation in Section \ref{sec:unifieda}, the RAIC renders $\mu_1 = \frac{1}{4}$, $\mu_2=0$, $e(\bx)=10\|\bx-\bx_{[s]}\|_2$, and $R_{\bx}^{\rm loc}=c\|\bx\|_2$, thus (\ref{con1iht})--(\ref{con3iht}) are satisfied by the $\bx\in\calX$ as long as $c_*$ is sufficiently small. Then, we reuse the initialization bound from Theorem \ref{thm:reluuniform} to ensure (\ref{iniiht}), i.e., 
    \begin{align}
        \label{relunonuini}
        \bx_0 \in \Sigma^n_s\cap \mathbb{B}_2^n(\bx;c\|\bx\|_2).
    \end{align}
    As such, Theorem \ref{raiciht} concludes that Algorithm \ref{alg:reluiht} without $P_{\mathbb{B}_2^n(2\|\bx_0\|_2)}$ satisfies the desired bound. Finally, we incorporate  $P_{\mathbb{B}_2^n(2\|\bx_0\|_2)}$ by showing that it has no impact.

    \item  {\bf (Instance Optimality if $\bx\in\calX^c$)} This case is handled by a different argument. 
\end{enumerate}
\subsubsection{Proof of the RAIC (\ref{nonuraicrelu}) if $\bx\in\calX$}
Following the proof outline in Section \ref{sec:nonurelu}, the main bulk of technical work lies in the proof of the RAIC in (\ref{nonuraicrelu}). Particularly, for a fixed nonzero $\bx\in \calX = \{\bu\in \mathbb{R}^n:\delta_s(\bu)\le c_*\|\bu\|_2\}$, we seek to show that
\begin{align}
    \label{nonuraic2relu}
    \|\bu-\bx-\bh_{\bx}(\bu)\|_{2s,2}\le \frac{\|\bu-\bx\|_2}{4} + 10\|\bx-\bx_{[s]}\|_2 ,\quad \forall \bu\in\Sigma^n_s\cap \mathbb{B}_2^n(\bx;c\|\bx\|_2)
\end{align}
holds for some universal constant $c>0$. To this end, we reuse (\ref{reuserelu11}) and (\ref{reuserelu22}) to obtain 
\[\|\bu-\bx-\bh_{\bx}(\bu)\|_{2s,2}\le \frac{\|\bu-\bx\|_2}{4} + \frac{5\|\bx-\bx_{[s]}\|_2}{4} + \|\bh_{\bx}(\bu)-\bh_{\bx_{[s]}}(\bu)\|_{2s,2},\]
and the only difference arises in bounding $ \|\bh_{\bx}(\bu)-\bh_{\bx_{[s]}}(\bu)\|_{2s,2}$: 
\begin{align*}
     \|\bh_{\bx}(\bu)-\bh_{\bx_{[s]}}(\bu)\|_{2s,2} \stackrel{(a)}{\le} \sup_{\bw\in\Sigma^{n,*}_{2s}}\frac{2}{m}\sum_{i=1}^m |\ba_i^\top(\bx-\bx_{[s]})||\ba_i^\top\bw| \stackrel{(b)}{\le} 8\|\bx-\bx_{[s]}\|_2 ,
\end{align*}
where $(a)$ holds due to the first two inequalities in (\ref{mismatchrelu}), $(b)$ is by (\ref{nonumismatchspr}) and (\ref{vergaunorm}). Combining these pieces establishes (\ref{nonuraic2relu}), which is exactly the desired RAIC (\ref{nonuraicrelu}).

\subsubsection{Proof of the Initialization Guarantee (\ref{relunonuini})  if $\bx\in\calX$}
Since $\bx_0\in\Sigma^n_s$ by construction, we only need to further ensure $\|\bx_0-\bx\|_2\le c\|\bx\|_2$.  
For any fixed $\bx\in\mathbb{R}^n$ (need not be in $\calX$), with probability at least $1-2\exp(-cs\log\frac{en}{s})$ we have    
\begin{align}\nn
    \frac{\|\bx_0-\bx\|_2}{\|\bx\|_2} &\stackrel{(a)}{\le} \sup_{\bw\in \Sigma^{n,*}_{2s}}\frac{2}{m}\sum_{i=1}^m \bigg[2\relu(\ba_i^\top\bx)\ba_i^\top\bw-\mathbb{E}(2\relu(\ba_i^\top\bx)\ba_i^\top\bw)\bigg] +\frac{3\|\bx-\bx_{[s]}\|_2}{\|\bx\|_2}  \\\label{iniboundrelu}
    &\stackrel{(b)}{\le} C_1\sqrt{\frac{s\log(en/s)}{m}}+\frac{3\|\bx-\bx_{[s]}\|_2}{\|\bx\|_2}, 
\end{align}
where $(a)$ is obtained by adapting the argument in Equations (\ref{reluinieq1}) and (\ref{reluinieq2}), $(b)$ holds with probability at least $1-2\exp(-cs\log\frac{en}{s})$ due to Lemma \ref{menproduct} (note that $\bx$ is   fixed). Since $\frac{\|\bx-\bx_{[s]}\|_2}{\|\bx\|_2}\le c_*$ when $\bx\in\calX$, $\|\bx_0-\bx\|_2\le c\|\bx\|_2$ is satisfied under $m\gtrsim s\log\frac{en}{s}$ and small enough $c_*$.

\subsubsection{The Projection Has No Effect if $\bx\in\calX$}
Suppose that $\{\bx_t\}_{t\ge 0}$ is obtained by Algorithm \ref{alg:reluiht} without $P_{\mathbb{B}_2^n(2\|\bx_0\|_2)}$, then under the established components Theorem \ref{raiciht} yields
\begin{align} \label{fixedXboundrelu}
    \|\bx_t-\bx\|_2\le \frac{c\|\bx\|_2}{2^t} + 46\|\bx-\bx_{[s]}\|_2,\quad\forall t\ge 0.
\end{align}
This implies $\|\bx_t\|_2 \le \|\bx_t-\bx\|_2+ \|\bx\|_2 \le (1+c+46c_*)\|\bx\|_2$ for any $t\ge 0$.  On the other hand, (\ref{iniiht}) with $R_{\bx}^{\rm loc}=c\|\bx\|_2$ yields $\|\bx_0\|_2\ge (1-c)\|\bx\|_2$. Hence, under small enough $c$ and $c_*$, we have $\|\bx_t\|_2\le 2\|\bx_0\|_2$ for any $t\ge 0$. Hence, adding $P_{\mathbb{B}_2^n(2\|\bx_0\|_2)}$ has no effect on the iterates. We conclude that (\ref{fixedXboundrelu}) holds for Algorithm \ref{alg:reluiht}. 

\subsubsection{Separate Argument if $\bx\in\calX^c$}
All that remains is to address the case of $\bx\in\calX^c$, where we continue to assume the high-probability event (\ref{iniboundrelu}). This event implies
\[\|\bx_0\|_2\le \|\bx\|_2+\|\bx-\bx_0\|_2\le \bigg(1+C_1\sqrt{\frac{s\log(en/s)}{m}}\bigg)\|\bx\|_2+3\|\bx-\bx_{[s]}\|_2\le 2\|\bx\|_2+3\|\bx-\bx_{[s]}\|_2.\]
Due to $P_{\mathbb{B}_2^n(2\|\bx_0\|_2)}$, the iterates of Algorithm \ref{alg:reluiht} satisfy $\|\bx_t\|_2\le 2\|\bx_0\|_2\le 4\|\bx\|_2+6\|\bx-\bx_{[s]}\|_2$. Combining   with $\|\bx\|_2\le \frac{\|\bx-\bx_{[s]}\|_2}{c_*}$, we conclude that 
\[\|\bx_t-\bx\|_2\le \|\bx_t\|_2+\|\bx\|_2 \le 5\|\bx\|_2+6\|\bx-\bx_{[s]}\|_2 \le (5c_*^{-1}+6)\|\bx-\bx_{[s]}\|_2,\quad\forall t\ge 0,\]
as desired.

\end{document}

%% file: preamble.tex
\usepackage{amssymb,amsmath,amsfonts,latexsym}
\usepackage{amsmath,graphicx,bm,xcolor,url}
\usepackage[caption=false]{subfig} 
\usepackage{array}
\usepackage{verbatim}
\usepackage{bm}
\usepackage{verbatim}
\usepackage{textcomp}
\usepackage{mathrsfs}
\usepackage{relsize}
\usepackage{subfig}
 \usepackage{amsthm}

\catcode`~=11 \def\UrlSpecials{\do\~{\kern -.15em\lower .7ex\hbox{~}\kern .04em}} \catcode`~=13 

\allowdisplaybreaks[3]

\newcommand{\nn}{\nonumber}

\newcommand{\calC}{\mathcal{C}}

\newcommand{\calE}{\mathcal{E}}
\newcommand{\calF}{\mathcal{F}}
\newcommand{\calG}{\mathcal{G}}

\newcommand{\calI}{\mathcal{I}}

\newcommand{\calK}{\mathcal{K}}
\newcommand{\calL}{\mathcal{L}}

\newcommand{\calN}{\mathcal{N}}

\newcommand{\calP}{\mathcal{P}}
\newcommand{\calQ}{\mathcal{Q}}

\newcommand{\calS}{\mathcal{S}}
\newcommand{\calT}{\mathcal{T}}
\newcommand{\calU}{\mathcal{U}}
\newcommand{\calV}{\mathcal{V}}
\newcommand{\calW}{\mathcal{W}}
\newcommand{\calX}{\mathcal{X}}

\newcommand{\ba}{\mathbf{a}}
\newcommand{\bA}{\mathbf{A}}

\newcommand{\be}{\mathbf{e}}

\newcommand{\bg}{\mathbf{g}}

\newcommand{\bh}{\mathbf{h}}

\newcommand{\bI}{\mathbf{I}}

\newcommand{\bM}{\mathbf{M}}

\newcommand{\bS}{\mathbf{S}}

\newcommand{\bu}{\mathbf{u}}

\newcommand{\bv}{\mathbf{v}}

\newcommand{\bw}{\mathbf{w}}

\newcommand{\bx}{\mathbf{x}}
\newcommand{\bX}{\mathbf{X}}
\newcommand{\by}{\mathbf{y}}

\newcommand{\bbeta}{\bm{\beta}}

\DeclareMathOperator{\supp}{supp}

\newcommand{\qednew}{\nobreak \ifvmode \relax \else
      \ifdim\lastskip<1.5em \hskip-\lastskip
      \hskip1.5em plus0em minus0.5em \fi \nobreak
      \vrule height0.75em width0.5em depth0.25em\fi}

